\documentclass{article}

\usepackage{fullpage}
\usepackage{appendix}
\usepackage{amsthm,amsmath,amsfonts,amssymb,amsxtra,bookmark,dsfont,bm,xcolor}
\numberwithin{equation}{section}

\usepackage{mathtools}
\usepackage{enumerate}
\allowdisplaybreaks
\usepackage[auth-lg]{authblk}

\theoremstyle{plain}
\newtheorem{theorem}{Theorem}[section]
\newtheorem{proposition}[theorem]{Proposition}
\newtheorem{lemma}[theorem]{Lemma}
\newtheorem{corollary}[theorem]{Corollary}

\theoremstyle{remark}
\newtheorem{remark}[theorem]{Remark}

\def\Tr{{\rm Tr}}

\def\eps{\varepsilon}

\def\Re{\operatorname{Re}}
\def\Im{\operatorname{Im}}
\def\1{{\ensuremath {\mathds 1} }}

\def\cF {\mathcal{F}}

\def\cH{\mathcal{H}}

\def\R {\mathbb{R}}

\def\cN {\mathcal{N}}

\def\cU {\mathcal{U}}

\def\cR{\mathcal{R}}

\def\rd{\text{d}}
\def\cF {\mathcal{F}}

\def\bR {\mathbb{R}}
\def\bC {\mathbb{C}}

\def\bN {\mathbb{N}}

\def\cL{\mathcal{L}}

\def\d{{\rm d}}

\def\ii{\mathrm{i}}

\def\dx{\mathrm{d}x}
\def\dy{\mathrm{d}y}
\def\dz{\mathrm{d}z}
\def\dr{\mathrm{d}r}
\def\ds{\mathrm{d}s}
\def\dt{\mathrm{d}t}
\def\dk{\mathrm{d}k}
\def\dxi{\mathrm{d}\xi}
\def\cW{\mathcal{W}}

\DeclareMathOperator{\ImPart}{Im}

\title{Microscopic  derivation of a Schr\"odinger equation in dimension one with a nonlinear point interaction}

\author[1]{Riccardo Adami}
\author[2]{Jinyeop Lee}

\affil[1]{Dipartimento di Scienze Matematiche ``G.L. Lagrange'', Politecnico di Torino, \newline
Corso Duca degli Abruzzi, 24, 10129, Torino, Italy\newline }

\affil[2]{Departement Mathematik und Informatik, Universität Basel, \newline
Spiegelgasse 1,CH-4051, Basel, Switzerland}

\date{}

\begin{document}

\maketitle

\begin{abstract}
	We derive an effective equation for the dynamics of many identical bosons in dimension one in the presence of a tiny impurity. The interaction between every pair of bosons is mediated by the impurity through a positive three-body potential. Assuming a simultaneous mean-field and short-range scaling with the short-range proceeding slower than the mean-field, and choosing an initial fully condensed state, we prove propagation of chaos and obtain an effective one-particle Schr\"odinger equation with a defocusing nonlinearity concentrated at a point. More precisely, we prove the convergence of one-particle density operators in the trace-class topology and estimate the fluctuations as superexponential. This is the first derivation of the so-called nonlinear delta model, widely investigated in the last decades, as a phenomenological model for several physical phenomena.
\end{abstract}

\section{Introduction}
\label{sec:intro}

Concentrated nonlinearities for the Schr\"odinger equation were introduced in the nineties of the twentieth century with the aim of effectively modeling several physical phenomena.
In \cite{PJ-LC91}, a model with a concentrated nonlinearity was proposed to describe the nonlinear effects arising when a bunch of electrons experiences resonant tunneling through a double-well potential. 
In that case, the effect of the trapped electrons inside the well was described by a nonlinearity that vanishes outside the well. 
Since the double well arises along one direction only, the problem resulted in the analysis of a one-dimensional system. 
In \cite{MA93}, the spatial range of the nonlinearity was reduced to a single point, so that the interaction was described as a nonlinear point interaction also called nonlinear delta, namely a delta potential whose strength depends on the function to which it applies. 
In \cite{N93}, a short-range nonlinearity was studied in the context of open quantum systems using a more abstract mathematical framework.

Generally speaking, a nonlinear point interaction is supposed to be useful for the sake of describing the action of a nonlinear layer whose transverse dimension is much smaller than the typical wavelength of the incoming particle.

Formally, the model is realized by
the Schr\"odinger equation
\begin{equation} \label{eq:concnlsformal}
	\ii \partial_t u_t \ = \ {\mbox{``}\,} - u_t'' + \mu | u_t|^p \delta u_t {\;\mbox{''}}.
\end{equation}
Here the Dirac delta on the right-hand side acts as a multiplicative factor, working as a potential concentrated at a single point. It is often said that such a term describes the effect of an impurity characterized by a range that is much smaller than the wavelength of the particle described by the wave function $u_t$.

Despite its apparent simplicity, this model proves to be nontrivial to define in dimensions two and three, where it is necessary to apply the theory of self-adjoint extensions of symmetric operators \cite{AGH-KH88}. The result of the rigorous construction does not coincide with the ordinary notion of a delta distribution, and this is the reason for the use of quotation marks in \eqref{eq:concnlsformal}.

The rigorous analysis of the Schrödinger equation with a nonlinear point interaction was initially carried out in dimension one  \cite{AdamiTeta2001class}, three  \cite{ADFT03}, and finally in dimension two \cite{CCT18}. 
In these works well-posedness is established in appropriate functional spaces and it is demonstrated that the main feature of the models is the reduction to nonlinear integral equations with singular kernels. Further investigations provided results on the existence of blow-up solutions \cite{Adami2020blow,ADFT04, Holmer2020blow}, on the presence of standing waves with their orbital and asymptotic stability \cite{Adami2021Stability2D, ANO13, ANO16, BKKS08, KKS12}, and on supercritical scattering \cite{AFH21}.

In contrast to the case of the standard nonlinear Schr\"odinger equation (NLS) or Hartree equation, however, a rigorous derivation of \eqref{eq:concnlsformal} as the effective dynamics for a quantum system has not been previously accomplished. 
Indeed, in the only available mathematical results of this kind, i.e. \cite{CFNT14} for dimension one and \cite{CFNT17} for dimension three, equation \eqref{eq:concnlsformal} is derived as the limiting behavior of a predefined short-range nonlinearity. 
Thus the question of how the concentrated nonlinearity arises from an underlying fundamental linear model remained untouched.

Here we present the first rigorous derivation of \eqref{eq:concnlsformal} from
a many-body quantum dynamics in the special case $p=2$, i.e., we rigorously deduce a cubic equation with a pointwise nonlinearity. We limit our analysis to the one-dimensional setting, however we plan to explore the corresponding problem in dimensions two and three in the near future. Furthermore, here we treat the case of defocusing nonlinearity. This choice is due to the fact that the target equation is $L^2$-critical and for the moment we prefer to focus on the derivation of the pointwise nonlinearity avoiding to deal with issues related to the existence of blow-up solutions.

The purpose of this paper is to demonstrate that, analogously to the standard NLS and Hartree equations, \eqref{eq:concnlsformal} can be constructed as the effective dynamics for a quantum system consisting of a large number $N$ of identical bosons. This result links the equation for an $N$-body linear quantum system with the equation for a one-body nonlinear dynamics.

Let us point out that the choice of bosonic symmetry in the initial state does not prevent the model from being applied to the resonant tunneling of electrons. As explained in \cite{PJ-LC91}, one can assume that the initial state of the electrons factorizes into a longitudinal 
and a transversal component.

One can assume that the longitudinal component exhibits bosonic symmetry, thus it can be factorized in the single-electron variables. Specifically, for $\Phi_{N,0} (\mathbf{x}_1,\dots, \mathbf{x}_N)$, the initial state of the system of $N$ electrons (where the boldface denotes $\mathbf{x}_j = (x_j, y_j, z_j)$, the position variable in three dimensions of the $j$-th electron), we have
\[
\Phi_{N,0} (\mathbf{x}_1,\dots, \mathbf{x}_N) = \varphi(x_1) \dots \varphi(x_N) \, \Omega(y_1, z_1, \dots, y_N, z_N).
\]
Here, $\varphi$ represents the longitudinal component of the initial state of each electron, while $\Omega$ is the collective transversal component, which is assumed to be totally antisymmetric under variable exchange, i.e.,
\[
\Omega(y_{\pi(1)}, z_{\pi(1)}, \dots, y_{\pi(N)}, z_{\pi(N)}) = (-1)^{\rm{sgn}(\pi)} \Omega(y_1, z_1, \dots, y_N, z_N),
\]
for every permutation $\pi$ of the set $\{1, \dots, N\}$. 
This ensures that the overall initial wave function $\Phi_{N,0}$ to be antisymmetric, thereby respecting the Pauli's principle for an $N$-electron system.

This shows that the one-dimensional nature in this model is intrinsic and not just a simplification of higher dimensional model.
Throughout the paper, we focus on describing the longitudinal evolution only, under the assumption that it decouples from the transversal component.

The characteristic feature of the present model is the way bosons interact with one another. 
Indeed, instead of experiencing a usual two-body interaction, every pair of particles undergoes an interaction mediated by a third body, that is an impurity, whose position is fixed. 
In other words, in order to derive a nonlinear point interaction, we use a three-body interaction potential that acts on triplets made of a pair of bosons and
the impurity.

In terms of scaling, we consider for the potential a mean-field regime together with a short-range limit. This means that the strength of the interaction scales as the inverse of the number of particles $N$, so that the kinetic and potential energies scale in the same way as $N$ grows and simultaneously the range of the interaction shrinks to a single point. However, the short-range limit is slower than the mean-field.

Concerning the initial data,
we make use of a factorized state, so that all bosons share the same quantum state $\varphi$. Then
the $N$-body dynamics we consider is
\begin{equation} \label{eq:manybody}
	\left\{
	\begin{array}{c}
		\ii \partial_{t} \Psi_{N,t} (X_N) = - \displaystyle \sum_{j=1}^N \Delta_{x_j} \Psi_{N,t} (X_N) + \frac \mu N \sum_{1 \leq k < \ell \leq N} W_\varepsilon (c, x_k, x_\ell)\Psi_{N,t} (X_N)
		\\
		\\
		\Psi_{N,0} (X_N) \ = \ \varphi^{\otimes N} (X_N) \ = \ \varphi(x_1) \dots \varphi (x_N)
	\end{array}
	\right.
\end{equation}
where:
\begin{itemize}
	\item $x_j$ is the spatial coordinates of the $j$-th boson and $\Delta_{x_j} = \partial_{x_j}^2$;
	
	\item
	$X_N = (x_1, \dots, x_N) \in \bR^N$ is the string of the coordinates of the $N$ bosons;
	\item $\Psi_{N,t}$ is the wave function at time $t$ of the system made of $N$ identical bosons;
	\item The initial data is the $N$.th tensor power of the same function $\varphi \in H^1(\R)$, that satisfies the normalization condition $\int_\R \d x\,| \varphi (x) |^2 = 1$;
	
	\item $\mu > 0$ denotes the strength of the interaction;
	\item The short-range parameter $\varepsilon$ depends on $N$. In order for our techniques to work, we assume that
	\[
		\varepsilon^{-1} = o(\log N).
	\]
	So, to fix ideas, one can take, e.g.,
	\begin{equation}
		\label{eq:fixepsilon}
		\varepsilon : = (\log N)^{-\frac 1 2};
	\end{equation}
	\item $W_\varepsilon$ is the three-body potential that describes the interaction between couples of bosons and the impurity, located at the point $c$. One can think of it as of
	\[
		W_\varepsilon (c,x_k,x_\ell) : = w_\varepsilon (c-x_k) w_\varepsilon (c-x_\ell),
		\quad\text{with}\quad
		w_{\varepsilon}(x):=
		\varepsilon^{-1} w ( \varepsilon^{-1} x ),
	\]
	where $w$ is positive, even, and in Schwartz class.
	
	Furthermore, placing the origin of the coordinates at the position of the impurity, one gets $c=0,$ so the potential simplifies to
	\begin{equation}
		\label{eq:simplepot}
		W_\varepsilon (x_k,x_\ell) :
		= W_\varepsilon (0,x_k,x_\ell)
		= w_\varepsilon (x_k) w_\varepsilon (x_\ell) ;
	\end{equation}
	\item Denoting by $H_N$ the Hamiltonian operator that generates the dynamics of the $N$-body system in the presence of the impurity, namely
	\begin{equation}\label{eq:manybodyH}
		H_N = - \sum_{j=1}^N \Delta_{x_j} + \frac \mu N \sum_{1 \leq k < \ell \leq N} W_\varepsilon (x_k, x_\ell) ,
	\end{equation}
	one can express the system \eqref{eq:manybody} in the shorthand way
	\[
	 \Psi_{N,t} : = e^{-\mathrm{i}tH_N} \varphi^{\otimes N};
	\]
	\item Since both the initial data and the Hamiltonian are symmetric under exchange of coordinates, at any time the solution $\Psi_{N,t}$ preserves the same symmetry.
	
	\item As we choose the one-particle initial data $\varphi$ as an element of $H^1 (\R)$, we are interested in the so-called mild solution to the problem \eqref{eq:manybody}, namely the solution in the energy space $H^1 (\R^N),$ that strictly speaking solves the integral version of \eqref{eq:manybody}. 
\end{itemize}
We can now state our main result.

\begin{theorem}\label{thm:main}
	Given a function $\varphi \in H^1 (\bR)$, normalized in $L^2(\R)$, let $\Psi_{N,t}$ be the mild solution to the Cauchy problem \eqref{eq:manybody} with $\mu>0$.
	Define the one-particle reduced density matrix $\gamma_{N,t}^{(1)}$ as the integral operator whose kernel is defined by
	\begin{equation}
		\label{reduceddensity}
		\gamma_{N,t}^{(1)} (x,y)
		: = \int_{\bR^{N-1}} \d Z \,
		\overline{\Psi_{N,t} (y,Z)}
		\Psi_{N,t} (x,Z) .
	\end{equation}
	Let $\varphi_{t}$ be the mild solution to the equation
	\begin{equation}
		\label{eq:limiteq}
		\ii \partial_t \varphi_t \ = \ - \varphi_t'' + \mu | \varphi_t |^2 \delta \varphi_t
	\end{equation}
	with initial data $\varphi_{0}=\varphi$, namely, the solution to the integral equation
	\begin{equation}\label{eq:integralequationdelta}
		\varphi_t \ = \ U(t) \varphi - \ii \mu \int_0^t \d s \, U (t-s) |\varphi_s|^2 \delta \varphi_s,
	\end{equation}
	where $U(t)$ denotes the free Schr\"odinger propagator in dimension one.
	
	Then, for $0 < \varepsilon < 1$, there exist positive constants $C$ and $K$ such that the following bound holds:
	\begin{equation}\label{eq:main}
		\begin{split}
			&\Tr\left|\gamma_{N,t}^{(1)}-|\varphi_{t}\rangle\langle\varphi_{t}|\right|
			\leq
			\frac{C}{N}\,\Big(1+\frac{1}{\varepsilon^{3/2}}+\frac{1}{\varepsilon^{\frac{37}{2}}N}+\frac{1}{\varepsilon^{37}N^2}\Big) \exp\left(K (1+\varepsilon^{-1}) t \right)
			+ C \varepsilon^{1/4} \, e^{Kt}
	\end{split}\end{equation}
	where $|\varphi_{t}\rangle\langle\varphi_{t}|$ denotes the projection operator to $\varphi_{t}$ in $L^2(\bR)$.
\end{theorem}

\vspace{1em}

The following immediate Corollary points out the convergence in trace norm:
\begin{corollary}\label{cor:tr-convergence}
	With the same assumptions as stated in Theorem \ref{thm:main}, we have
	\[
		\Tr\left|\gamma_{N,t}^{(1)}-|\varphi_{t}\rangle\langle\varphi_{t}|\right| \to 0
		\text{ as }
		N \to \infty
	\]
	when $\varepsilon \to 0$ as $N\to\infty$ satisfies that
	$\varepsilon^{-1} = o(\log N)$
	and $0\leq t = o(\log \varepsilon)$.
\end{corollary}

\vspace{2em}
\noindent\textit{Let us point out the following remarks:}
\begin{enumerate}[1)]
	
	\item As usual the notation $X = o(Y)$ is used to express that $X$ grows at a much smaller rate than $Y$ as $N\to \infty$.
	This is crucial for obtaining the convergence of the first term in \eqref{eq:main}.
	
	\item 
	The existence and uniqueness of the mild solution to \eqref{eq:integralequationdelta} have been proven in \cite{AdamiTeta2001class}.
	
	\item 
	For the second term in \eqref{eq:main}, the exponent $1/4$ of $\varepsilon$ can be chosen to be any $\eta\in (0,1/2)$ so that
	\begin{equation}\label{eq:rmk}
		\begin{split}
			&\Tr\left|\gamma_{N,t}^{(1)}-|\varphi_{t}\rangle\langle\varphi_{t}|\right|
			\leq
			\frac{C}{N}\,\Big(1+\frac{1}{\varepsilon^{3/2}}+\frac{1}{\varepsilon^{\frac{37}{2}}N}+\frac{1}{\varepsilon^{37}N^2}\Big) \exp\left(K (1+\varepsilon^{-1}) t \right)
			+ C \varepsilon^\eta \, e^{Kt}.
	\end{split}\end{equation}
	
	\item
	To describe the three-body interaction that takes place between the bosons and the impurity, one may employ an interaction potential that involves direct interaction between the bosons, like e.g.
	\begin{align*}
		&W_{\varepsilon}(c,x_{k},x_{\ell})\\
		&\quad:=
		q_c^2 q_k q_\ell\,
		w_{\varepsilon}(c-x_{k})w_{\varepsilon}(c-x_{\ell})
		+
		q_c q_k^2 q_\ell\,
		w_{\varepsilon}(c-x_{k})v_{\varepsilon}(x_{k}-x_{\ell})
		+
		q_c q_k q_\ell^2\,
		w_{\varepsilon}(c-x_{\ell})v_{\varepsilon}(x_{k}-x_{\ell})
	\end{align*}
	where $q_c, q_k, q_\ell$ denote the charge with which the impurity located at $c$ and the particles at $x_k$ and $x_\ell$ interact.
	Moreover, we considered two different potentials: $w_{\varepsilon}$ for the interaction between the impurity and the bosons, and  $v_{\varepsilon}$ the interaction among the bosons, to distinguish between  different physical cases.
	
	If we assume that the interaction is electromagnetic and the charge of the impurity considerably surpasses that of the bosons, e.g. $q_c = N$ and $q_k,q_\ell=N^{-1}$, we can disregard the interaction terms with $v_{\varepsilon}$. This results in the interaction potential described above.
	Furthermore, we expect that for the case where  $q_c$ and all $q_k$ are of same order, one can follow the strategy given in this paper to derive \eqref{eq:limiteq} with a coefficient slightly different from $\mu$.
	
	\item One can add to the $N$-body model a standard mean-field term as well as a short-scale potential, and obtain, as a result of the scaling limit, equation \eqref{eq:limiteq} with additional Hartree and NLS term, namely $(V * |\varphi_t^2)\varphi_t$ and $|\varphi_t|^2\varphi_t$. However, our aim is to rigorously derive from a microscopic dynamics the nonlinear point interaction, which is the novelty of our work, therefore we omit such terms.

	\item
	Analogolously, in order to focus on the mathematically rigorous derivation of the nonlinear delta term, we restricted to an interaction between the impurity and couples of particles, neglecting a possible action between the impurity and single particles.
	
	Including such a term in the dynamics one obtains the $N$-body Hamiltonian
	\[
		\widetilde H_N := 
		- \sum_{j=1}^N \Delta_{x_j}
		+ \sum_{j=1}^N \upsilon_\varepsilon (c-x_j)
		+ \frac \mu N \sum_{1 \leq k < \ell \leq N} W_\varepsilon (c, x_k, x_\ell),
	\]
	where the potential $\upsilon_\varepsilon$ describes the interaction of the impurity with every
	individual particle and tends to a Dirac's delta 
	in the limit $\varepsilon \to 0$. It is well-known
	from the theory of the approximation of delta interactions by regular potentials \cite[Section 1.3.2]{AGH-KH88}, that the one-particle Hamiltonian $- \Delta + \upsilon_\varepsilon$ converges in the strong resolvent sense to the delta-interaction Hamiltonian $- \Delta + \alpha \delta $, where
	$\alpha = \int \upsilon. $
	Including such interaction in the model 
	reduce then to adapting a known argument, so we decided not to consider this issue.
	
	\item 
	We model the fact that the {interaction range is small} by letting $\varepsilon$ go to $0$. A similar procedure was often followed in the literature on two-body interaction models in 2D or 3D, e.g. in \cite{BCS2016AHP,BOS2015,nam2020derivation,Pickl2011LMP}, where the gap between the mean-field and Gross-Pitaevskii regimes had to be bridged by considering an intermediate regime. More specifically, for the Hamiltonian
	\[
		H_{N\!,\varepsilon} := - \sum_{j=1}^N \Delta_{x_j} + \frac{\varepsilon^{-d}}{N} \sum_{1 \leq k \leq \ell \leq N} V(\varepsilon^{-1} (x_k - x_\ell)),
	\]
	the most common regimes are defined as follows:
	\begin{itemize}
		\item \textbf{Mean-field regime}: When $\varepsilon = 1$ for any dimension $d = 1, 2, 3$.
		\item \textbf{Gross-Pitaevskii regime}:
		\begin{itemize}
			\item For 2D: $d = 2$, $\varepsilon = e^{-N}$.
			\item For 3D: $d = 3$, $\varepsilon = N^{-1}$.
		\end{itemize}
		\item \textbf{Intermediate regime}:
		\begin{itemize}
			\item For 2D: $d = 2$, $\varepsilon = N^{-\alpha}$ with $\alpha > 0$.
			\item For 3D: $d = 3$, $\varepsilon = N^{-\alpha}$ with $0 < \alpha < 1$.
		\end{itemize}
	\end{itemize}
	
	We adopt a similar approach, replacing $N^{-\alpha}$ with $\varepsilon$. Then, since there is no agreed-upon Gross-Pitaevskii scaling in 1D, we prove that when $\varepsilon^{-1} \sim \sqrt{\log N}$, we achieve the limiting equation as shown in Corollary \ref{cor:tr-convergence}.
	
	\item The hypothesis of factorization is often referred to as condensation, as it describes the collective state of a Bose-Einstein condensate, one of the most prominent research areas where factorization naturally applies. From a theoretical perspective, the concept of exactly factorized states has been explored for several decades (see, e.g., \cite{S80} and the subsequent research on mean-field limits). Note that the factorization holds exactly only in the infinite-particle limit, although it can approximate well for a sufficiently large number of particles.
	
	In one-dimensional systems, where condensation is a more delicate phenomenon, the existence of factorized states was proven in \cite{LSY04}. The proof is provided for fully factorized initial data, and it is expected that this result could be extended to other initial data that approximate fully factorized states.
	
\end{enumerate}

\vspace{1em}

Let us show heuristically how equation
\eqref{eq:limiteq} is related to \eqref{eq:manybody}. Putting $c = 0$ as in Theorem \ref{thm:main},
the $N$-body energy for the factorized state
$\varphi_{t}^{\otimes N}$ reads
\begin{align*}
	\mathcal{E}(\varphi_{t}^{\otimes N})&= \frac 1 2 \sum_{j=1}^{N}\int\mathrm{d}x_{j}\,|\varphi'_{t}(x_{j})|^{2}
	+\frac{\mu}{2N}\sum_{1\leq k<\ell\leq N}\int\mathrm{d}x_{k}\mathrm{d}x_{\ell}\,W_{\varepsilon}(x_{k},x_{\ell})|\varphi_{t}(x_{k})|^{2}|\varphi_{t}(x_{\ell})|^{2}.
\end{align*}
Now we equally distribute the energy among the $N$ particles. To this aim, we rewrite the interaction term in the energy as follows:
\begin{align*}
	\mathcal{E}(\varphi_{t}^{\otimes N})&=\frac 1 2 	\sum_{j=1}^{N}\int\mathrm{d}x_{j}\,|\varphi'_{t}(x_{j})|^{2}
	+\frac{\mu}{4N}\sum_{
		\substack{
			1\leq k , \ell\leq N \\
			k \neq \ell}}
	\int\mathrm{d}x_{k}\mathrm{d}x_{\ell}\,W_{\varepsilon}(x_{k},x_{\ell})|\varphi_{t}(x_{k})|^{2}|\varphi_{t}(x_{\ell})|^{2}.
\end{align*}
Exploiting the symmetry of the system, we can focus on the contribution of the energy of the first particle, namely
\begin{align*}
	& \frac 1 2 \int\mathrm{d}x_{1}\,|\varphi'_{t}(x_{1})|^{2}+\frac{\mu}{4N}\sum_{\ell=2}^{N}\int\mathrm{d}x_{1}\mathrm{d}x_{\ell}\,W_{\varepsilon}(x_{1},x_{\ell})|\varphi_{t}(x_{1})|^{2}|\varphi_{t}(x_{\ell})|^{2}\\
	&=\frac 1 2 \int\mathrm{d}x_{1}\,|\varphi'_{t}(x_{1})|^{2}+\frac{\mu}{4}\Big(\int\mathrm{d}x_{1}\,w_{\varepsilon}(x_{1})|\varphi_{t}(x_{1})|^{2}\Big)\Big(\frac{1}{N}\sum_{\ell=2}^{N}\int\mathrm{d}x_{\ell}\,w_{\varepsilon}(x_{\ell})|\varphi_{t}(x_{\ell})|^{2}\Big)\\
	&= \frac 1 2\int\mathrm{d}x_{1}\,|\varphi'_{t}(x_{1})|^{2}+\frac{\mu}{4}\,\frac{N-1}{N} \Big(\int\mathrm{d}x_{1}\,w_{\varepsilon}(x_{1})|\varphi_{t}(x_{1})|^{2}\Big)^{2}.
\end{align*}
Now, taking the limit $N\to \infty$, and recalling that $\varepsilon \to 0$ as $N\to \infty$, the previous expression converges to:
\[
	\frac{1}{2}\int\mathrm{d}x_{1}\,\big|\varphi'_{t}(x_{1})\big|^{2}+\frac{\mu}{4}\big|\varphi_{t}(0)\big|^{4},
\]
which coincides with the one-particle energy of the solution of \eqref{eq:limiteq}.

However, as well understood in this kind of problems (see \cite{Pickl2011LMP}), the limit to be performed is
conveniently formulated in terms of the \emph{reduced density matrix}, because proving $L^2$-convergence for the wave function, i.e.
$\Psi_{N,t}\approx \varphi_t^{\otimes N}$ in the  $L^2$-norm, is out of reach: if we consider an $N$-body wave function $\Psi_N=\varphi^{\otimes (N-1)} \vee \varphi^\perp$ with $\varphi^\perp$ orthogonal to $\varphi$, we find indeed that the distance between $\Psi_N$ and $\varphi^{\otimes (N-1)}\vee \varphi^\perp$ equals $\sqrt{2}$ because the two $N$-particle states are orthogonal. Therefore, an uncontrolled behaviour of one sole particle could result in the maximal distance between two quantum states, even though all other particles reside  in the same state $\varphi$.
On the other hand, if one proceeds like in Theorem \ref{thm:main} and considers the trace norm distance between the one-particle reduced density matrix associated to $\varphi^{\otimes (N-1)}\vee \varphi^\perp$ with the one associated to $\varphi^{\otimes N}$, then it can be checked that the trace norm distance is bounded by $C N^{-1}$. Moreover, from the physical point of view
the trace class norm proves to be meaningful as it provides convergence of the expectation values of bounded observables.

\medskip

Our strategy consists in separating the mean-field and the short-range limit. 
Specifically, the mean-field limit $N \to \infty$ is achieved by adapting to our problem the techniques of  \cite{Xchen2012second,Lee2020rate}. Eventually, one obtains the following one-body equation
\begin{equation} \label{eq:preintermediateintro}
	\ii \partial_t u_t = - u_t'' + w_\varepsilon (w_\varepsilon * |u_t|^2) (0) u_t,
\end{equation}
to which we will refer as to   the {\em concentrated Hartree equation}, since it is a Hartree-type equation whose nonlinear term is concentrated by the presence of the factor $w_\varepsilon$.  Exactly because of this localization of the interaction, \eqref{eq:preintermediateintro} cannot be straightforwardly derived using previous results. 
Furthermore, notice that in \eqref{eq:preintermediateintro} the convolution term is evaluated at the origin, where the impurity is located.  
Introducing the bracket notation $\langle \cdot, \cdot \rangle$ for the hermitian product in $L^2 (\bR)$, we can rewrite equation \eqref{eq:preintermediateintro} as
\begin{equation}
	\label{eq:intermediateintro}
	\ii \partial_t u_t = - u_t'' + w_\varepsilon \langle w_\varepsilon , |u_t|^2 \rangle u_t.
\end{equation}
We stress that, notwithstanding the localization at zero of the convolution term, equation \eqref{eq:intermediateintro} remains nonlocal, since the value of $u_t$ at every point contributes to the hermitian product.
At this stage, equation \eqref{eq:intermediateintro} can be considered as an intermediate problem between \eqref{eq:manybody} and \eqref{eq:limiteq}.

Once obtained \eqref{eq:intermediateintro}, as a second step we perform the limit $\varepsilon \to 0$. Such a step relies on comparing two different one-particle evolutions, and this is carried out by employing a method inspired by \cite{CFNT14}. Again, the limit we perform is new and techniques already known have been suitably adapted.

Consistently with the two steps just described, we estimate the error by splitting the limit in two parts, through the triangular inequality
\begin{equation}\label{eq:triangular}
	\Tr\Big| \gamma_{N,t}^{(1)}- | \varphi_{t} \rangle \langle \varphi_{t}| \Big| 
	\ \leq \
	\Tr\Big| \gamma_{N,t}^{(1)}- | u_{\varepsilon,t} \rangle \langle u_{\varepsilon,t}| \Big|
	+
	\Tr\Big| | u_{\varepsilon,t} \rangle \langle u_{\varepsilon,t}| - |\varphi_{t}\rangle\langle\varphi_{t}|\Big|
\end{equation}
where $u_{\varepsilon,t}$ is the mild solution to the intermediate equation \eqref{eq:intermediateintro} with initial data $\varphi$, and by estimating the two terms in the r.h.s. separately, following a strategy similar to that in \cite{BCS2016AHP,BOS2015,nam2020derivation,Pickl2011LMP}.

We warn the reader that, even though conceptually the limit in $N$ precedes that in $\varepsilon$, in order to make the paper more readable we proceed in the inverse order.

Concerning the techniques, we follow the recent achievements on mean-field and Gross-Pitaevskii limits. In particular, we make use of the breakthrough results in
\cite{RodnianskiSchlein2009} where, inspired by \cite{GV79-1,GV79-2,Hepp}, the authors employed the coherent state approach in Fock space to derive the mean-field limit of the dynamics of many-body quantum systems with two-body interactions.
The gain with respect to previous derivations (e.g. \cite{S80}) lies in the estimate of the width of the fluctuations around the limit. Later, in \cite{KnowlesPickl2010meanfield} another strategy to derive mean-field  limit with an estimate of the error was developed.
The optimal convergence rate for Gross-Pitaevskii limits was finally obtained in \cite{BS19}, while in \cite{grillakis2010second,grillakis2011second} the authors employed the Bogoliubov transformation to derive second-order corrections to mean-field evolution of weakly interacting bosons.
Results on approximation in norm were obtained in \cite{BNNS19,NN17-1,NN17-2}.

The large $N$-limit of the dynamics of $N$-body quantum systems with three-body interaction was rigorously studied in \cite{ChenPavlovic2011}  where the Gross–Pitaevskii limit of a Bose gas with three-body interactions was achieved using the method of the BBGKY hierarchy, obtaining  the quintic nonlinear Schrödinger equation (qNLS).
Later, in \cite{XChenHolmer2019derivation,Lee2020rate} the mean-field limit of three-body interacting Bose gas was obtained, and in \cite{nam2020derivation}, Gross–Pitaevskii limit of three-body interactions was obtained using Fock space methode.
Complementarily, in the recent paper \cite{NamRicaudTriay2022}  the ground state energy of a low-density Bose gas with three-body interactions was investigated.

From the technical point of view it is worth remarking that we employ only the Weyl transformation rather than Bogoliubov transformations in the language of second quantization. 
We believe that, together with the two-step strategy illustrated before, this choice streamlines the discussion, reduces computational costs, and necessitates fewer concepts. While this approach sacrifices the speed of the scaling limit, in our opinion it enhances clarity and simplicity of the derivation. 

The paper is organized as follows: Section \ref{sec:1body} contains results on the well-posedness and on the conservation laws of the dynamics generated by the the target equation \eqref{eq:limiteq} and that generated by the concentrated Hartree equation \eqref{eq:intermediateintro}; in fact, while for the former we just quote results from \cite{AdamiTeta2001class}, for the latter the results are new and their proof is given in detail. We notice that, owing to the one-dimensional character of the model, the theory can be easily developed in the energy space $H^1 (\bR)$.

Section \ref{sec:1bodyConvergence} is devoted to the proof of the convergence from \eqref{eq:intermediateintro} to \eqref{eq:limiteq} as $\varepsilon$ vanishes, that encodes the transition from nonlocal to local nonlinearity at the one-body level.

In Section \ref{sec:derivation} we introduce the second quantized formalism and provide several notions and results that prove useful  for microscopic derivation. 
Since second quantization is sometimes considered as complicated and cumbersome, we privilege self-containedness and assume sometimes a pedagogical attitude.
In particular,  we establish some \emph{a priori} estimates, define Fock spaces and recall some of their basic properties, delve into unitary operators and their generators, and introduce the  fluctuation dynamics. Part of the content of this section has the character of a quick review.

Lastly, Section \ref{sec:pfmainthm} presents the proof of Theorem \ref{thm:main}.

\vspace{0.5em}
Along the paper we denote by $\langle \cdot, \cdot \rangle$ either the standard Hermitian product in $L^2(\mathbb{R})$ or  the inner product in Fock space. 
Consistently, $\|\cdot\|$ may represent either the standard $L^2(\mathbb{R})$-norm or the norm in Fock space. The context will avoid every possible confusion or ambiguity caused by  this abuse of notation. As in the statement of Theorem \ref{thm:main}, we use the Dirac ket-bra notation $| f \rangle \langle f | $ to denote the orthogonal projection in $L^2 (\R)$ on the linear span of the fucntion $f$.

We use the symbol $U(t,x)$ to denote the integral kernel of the unitary group generated by the free Schr\"odinger equation in one dimension, i.e.,
\begin{equation}
	\label{freeschrod}
	U (t,x) \ = \ \frac 
	{e^{\mathrm{i} \frac{x^2}{4t}}}
	{\sqrt{4 \pi \mathrm{i} t}}, 
\end{equation}
so that for every $f \in L^2 (\R)$ it holds
\[
	(U(t) f) (x) \ = \ \int_\R \mathrm{d}y\, U (t, x-y) \, f (y) \ = \ \frac{1}{\sqrt{4 \pi \mathrm{i} t}} \int_\R \mathrm{d}y\,  e^{\ii \frac{(x-y)^2} {4t}}\, f (y).
\]
It is well-known that for every $t \in \R$, $U(t)$ is a unitary operator in all Sobolev spaces $H^s (\R).$

We define the Fourier transform $\widehat f$ of the function $f$ as follows:
\[
	\widehat f (k) \ : = \ \frac{1}{\sqrt{2 \pi}} \int_\R f(x) \, e^{-\ii k x} \, \d x, 
\]
so that it is a unitary operator in $L^2 (\R).$

\section*{Acknowledgements}
This project started from Oberwolfach Mini-workshop ``Zero-Range and Point-Like Singular Perturbations: For a Spillover to Analysis, PDE and Differential Geometry''.
We would like to thank the referees for careful review and helpful suggestions.
R.A. thanks the project of the Italian Ministery of Research PRIN 20225ATSTP - Nonlinear dispersive equations in presence of singularities.
R.A., in leave of absence from Politecnico di Torino, has been hosted at Laboratoire Jacques -Louis Lions, Université La Sorbonne, Paris.
J.L. was supported by the European Research Council (ERC CoG RAMBAS, Project Nr. 101044249), 
the Swiss National Science Foundation through the NCCR SwissMAP,
the SNSF Eccellenza project PCEFP2\_181153, 
by the Swiss State Secretariat for Research and Innovation through the project P.530.1016 (AEQUA), and
Basic Science Research Program through the National Research Foundation of Korea(NRF) funded by the Ministry of Education (RS-2024-00411072).

\section{Properties of one-body dynamics}\label{sec:1body}

In this section, we give results on the two one-body nonlinear dynamics \eqref{eq:limiteq} and \eqref{eq:intermediateintro}. 
As anticipated in Section \ref{sec:intro} we are interested in finite energy states, therefore we aim at studying mild solution rather than strong solutions. 
Mild solutions are solutions to the integral version of the equations obtained through the Duhamel's formula, namely, for \eqref{eq:limiteq} one get \eqref{eq:integralequationdelta}, while for \eqref{eq:intermediateintro} one has
\begin{equation} \label{eq:duhamel}
	u_{\eps,t} \
	= \ U (t) \varphi - \ii \mu \int_0^t \ds \, U (t-s) w_\varepsilon \langle w_\varepsilon,
	| u_{\eps,s} |^2 \rangle u_{\eps,s}. 
\end{equation}
For the definition of the first model \eqref{eq:limiteq} as well as for the proof of global well-posedness in $H^1 (\bR)$ and of conservation laws of $L^2$-norm and energy, we refer to \cite{AdamiTeta2001class},
from which we borrow the following results:

\begin{lemma}[Global well-posedness and Conservation 
	Laws for the NLS with pointwise nonlinearity]
	Consider the Cauchy problem given by equation \eqref{eq:limiteq} with $\mu > 0$ and the initial data $\varphi \in H^1(\bR)$. 
	Then 
	\begin{enumerate}
		\item (Global well-posedness \textnormal{\cite[Theorem 14]{AdamiTeta2001class})} There exists a unique mild solution $\varphi_t \in C^0 (\R, H^1(\R))$.
		
		\item (Conservation of the $L^2$-norm \textnormal{\cite[Theorem 7]{AdamiTeta2001class})} The solution $\varphi_t$ satisfies
		\[
			\|\varphi_t\|_{L^2(\bR)}=\|\varphi\|_{L^2(\bR)}.
		\]
		
		\item (Conservation of the energy \textnormal{\cite[Theorem 13]{AdamiTeta2001class})} Defined the functional
		\[
			E(\varphi_t):=\frac{1}{2}\|\varphi'_t\|_{L^2(\bR)}^2 + \frac{\mu}{4} |\varphi_t(0)|^4,
		\]
		there holds
		\[
			E (\varphi_t) \ = \ E (\varphi).
		\]
	\end{enumerate}
\end{lemma}

In Section \ref{subsec:conchartree} we provide detailed proofs of well-posedness and conservation laws for the concentrated Hartree equation \eqref{eq:duhamel}.


\subsection{Concentrated Hartree equation in the energy domain}
\label{subsec:conchartree}

In this section we prove global existence and uniqueness of equation 
\eqref{eq:duhamel}  with the initial one-particle state $\varphi$ in $H^1 (\bR)$. We follow the usual path starting with local well-posedness, that we prove on $H^\sigma$ with $\sigma > 1/2$, then showing the conservation laws of $L^2$-norm end energy, and finally using them to prove global well-posedness. In order to prove the conservation of the energy we need to consider at first  initial data $\varphi \in H^2 (\bR)$, then through an approximation argument we obtain the conservation of the energy in $H^1 (\bR)$, from which we get uniform estimates  in $t$ which are used to show global well-posedness. 

\begin{lemma}[Local well-posedness] \label{lemma:localex}
	Let $\sigma > \frac{1}{2}$ and $\varphi \in H^\sigma(\mathbb{R})$. Moreover, let $\varepsilon > 0$ be fixed. Then the integral equation \eqref{eq:duhamel}  is locally well-posed, i.e. there exists $T>0$ such that a unique solution $u_\eps \in C^0 ([0,T), H^\sigma(\mathbb{R}))$ exists for \eqref{eq:duhamel}.
\end{lemma}

\begin{proof}
	Let $T > 0$ to be specified later.
	In  the space $X : =  C^0 ([0,T], H^\sigma(\bR)) $, consider the closed ball $X_\rho$ of radius $\rho > 0$ centred at $U(\cdot) \varphi$, namely
	\[
		X_\rho : = \{ v \ \in \ X, \  \| v - U(\cdot) \varphi \|_X \leq \rho \},
	\]
	where we considered the standard norm $\| u \|_X := \sup_{t \in [0,T)} \| u_t \|_{H^\sigma(\bR)}$. As well-known, $X_\rho$ is complete as a metric space.
	
	We introduce in $X$ the map $\Phi$, defined by
	\begin{equation}
		\label{eq:mapphi}
		(\Phi v)_t \ : = \ U (t) \varphi - \ii \mu \int_0^t \ds \, U (t-s) w_\varepsilon \langle w_\varepsilon,
		| v_s |^2 \rangle v_s.
	\end{equation}
	Since $\sigma > 1/2$, the space
	$H^\sigma(\bR)$ is an algebra, therefore
	\begin{equation}
		\begin{split}
			\left\| \int_0^t \ds \, U (t-s) w_\varepsilon \langle w_\varepsilon,
			| v_s |^2 \rangle v_s \right\|_{H^\sigma(\bR)} \ \leq \ &  \int_0^t \ds
			\langle w_\varepsilon,
			| v_s |^2 \rangle
			\| w_\varepsilon   v_s \|_{H^\sigma(\bR)}  \\
			\ \leq \ &  C \| w \|_1 \| w_\varepsilon \|_{H^\sigma(\bR)} \int_0^t \ds \, \| v_s \|_\infty^2
			\| v_s \|_{H^\sigma(\bR)}
			\\
			\ \leq \ &  C T \| v \|_X^3.
		\end{split}
	\end{equation}
	Thus by \eqref{eq:mapphi}  one obtains
	\begin{equation}
		\begin{split}
			\| \Phi v - U(\cdot) \varphi \|_{X} \ & \leq \    C \| v \|_X^3  T
		\end{split}
	\end{equation}
	and by
	\[
		\| v \|_X \ \leq \ \| v -  U(\cdot) \varphi \|_{X}  + \|  U(\cdot) \varphi \|_{X} \ \leq \ \rho + \|\varphi \|_{H^\sigma(\bR)}, 
	\]
	one gets
	\begin{equation}
		\begin{split}
			\| \Phi v - U(\cdot) \varphi \|_{X} \ & \leq \    C  \left(\rho + \|\varphi \|_{H^\sigma(\bR)}\right)^3  T
		\end{split}
	\end{equation}
	so that, provided
	\begin{equation} \label{eq:phimap}
		T \ < \ \frac{\rho}{C \left(\rho + \|\varphi \|_{H^\sigma(\bR)}\right)^3},
	\end{equation}
	one concludes that $\Phi$ maps $X_\rho$ into itself. We look for further conditions on $T$ that guarantee that $\Phi$ is a contraction of $X_\rho$. Consider $\xi$ and $\zeta$ elements of $X_\rho$ and notice that
	\begin{equation}
		\label{eq:xizeta}
		\| \xi \|_X \ \leq \ \| \xi - U (\cdot) \varphi \|_X
		+ \| U (\cdot) \varphi \|_X \ \leq \ \rho +
		\| \varphi \|_{H^\sigma (\R)}
	\end{equation}
	and the same for $\| \zeta \|_X.$
	
	Then
	\begin{equation} \label{eq:contraction1}
		\begin{split}
			\left\| \left( \Phi \xi \right)_t - \left( \Phi \zeta \right)_t \right\|_{H^\sigma(\bR)}
			\ \leq \ & C \| w_\varepsilon \|_{H^\sigma(\bR)}  \int_0^t \ds \,
			\left\| \langle  w_\varepsilon, | \xi_s |^2 \rangle  \xi_s  -
			\langle  w_\varepsilon, | \zeta_s |^2 \rangle  \zeta_s
			\right\|_{H^\sigma(\bR)}.
		\end{split}
	\end{equation}
	We split the last factor in the integral as follows
	\begin{equation} \label{eq:contraction2}
		\begin{split}
			\big\| \langle  w_\varepsilon, | \xi_s |^2 \rangle  \xi_s  -
			\langle  w_\varepsilon, | \zeta_s |^2 \rangle  \zeta_s
			\big\|_{H^\sigma(\bR)}
			\leq \, &  \left\| \langle  w_\varepsilon, | \xi_s |^2 \rangle  ( \xi_s  - \zeta_s )
			\right\|_{H^\sigma(\bR)} +  \left\| \langle  w_\varepsilon, | \xi_s |^2  -| \zeta_s |^2 \rangle  \zeta_s
			\right\|_{H^\sigma(\bR)} \\
			\leq \, & \langle  w_\varepsilon, | \xi_s |^2 \rangle \| \xi_s  - \zeta_s \|_{H^\sigma(\bR)}
			+ \left| \langle  w_\varepsilon, | \xi_s |^2  -| \zeta_s |^2 \rangle \right| \left\| \zeta_s
			\right\|_{H^\sigma(\bR)}.
		\end{split}
	\end{equation}
	Using $\int dx \, w_\varepsilon (x) = 1$
	\begin{equation}
		\begin{split}
			\left| \langle  w_\varepsilon, | \xi_s |^2  -| \zeta_s |^2 \rangle \right| \ = \ &
			\left| \int_\bR \dx \,  w_\varepsilon (x) \left( | \xi_s (x) | + | \zeta_s (x) | \right)
			\left( | \xi_s (x) | - | \zeta_s (x) | \right) \right|
			\\
			\ \leq \ &  \int_\bR \dx \,  w_\varepsilon (x) \left( | \xi_s (x) | + | \zeta_s (x) | \right)
			| \xi_s (x) - \zeta_s (x) |
			\\
			\ \leq \ & \left( \| \xi_s \|_\infty + \| \zeta_s \|_\infty \right) \| \xi_s  - \zeta_s  \|_\infty
			\\
			\ \leq \ & C \left( \| \xi \|_X + \| \zeta \|_X \right) \| \xi  - \zeta  \|_X,
		\end{split}
	\end{equation}
	that, together with \eqref{eq:contraction1} and \eqref{eq:contraction2} yields
	\begin{equation} \label{eq:contraction3}
		\begin{split}
			\left\| \left( \Phi \xi \right)_t - \left( \Phi \zeta \right)_t \right\|_{H^\sigma(\bR)}
			\ \leq \ & C \| w_\varepsilon \|_{H^\sigma(\bR)} T \left( \| \xi \|_X^2 \| \xi  - \zeta  \|_X
			+ \| \zeta \|_X \left( \| \xi \|_X + \| \zeta \|_X \right)\| \xi  - \zeta  \|_X
			\right) \\
			\ \leq \ & C  T \left( \| \xi \|_X^2 + \| \zeta \|_X^2 \right) \| \xi  - \zeta  \|_X
			\\
			\ \leq \ & C T \left( \rho + \| \varphi \|_{H^\sigma(\bR)} \right)^2  \| \xi  - \zeta  \|_X.
		\end{split}
	\end{equation}
	Thus, owing to  \eqref{eq:phimap} and\eqref{eq:contraction3}, if
	\begin{equation} \label{existencetime}
		T \ < \ \min \left( \frac{1}{C \left(\rho + \| \varphi \|_{H^\sigma(\bR)} \right)^2}, \frac{\rho}{C \left(\rho + \|\varphi \|_{H^\sigma(\bR)}\right)^3} \right),
	\end{equation}
	then the map $\Phi$ is a contraction of $X_\rho$ into itself and by Banach-Caccioppoli fixed point theorem it admits a unique fixed point, that is a unique function $u_{\eps,t}$ such that 
	\begin{equation} 
		u_{\eps,t} \ = \
		(\Phi u_\eps)_t
		\ =   \ U (t) \varphi - \ii \mu \int_0^t \ds \, U (t-s) w_\varepsilon \langle w_\varepsilon,
		| u_{\eps,s} |^2 \rangle u_{\eps,s} 
		,\end{equation}
	then $u_{\eps,t}$ is the unique solution to \eqref{eq:duhamel}, which is then well-posed in the space $H^\sigma (\bR)$ 
	and the proof is complete.
\end{proof}
Notice that the existence time $T$ depends on $\varepsilon$ through the constant $C$, which in turn contains the norm in $H^\sigma (\bR)$ of the function $w_\varepsilon$, that diverges as $\varepsilon$ vanishes. This is unavoidable as we need to obtain a delta potential in the limit. Of course, this feature is absent in standard Hartree models.

In order to pass from local to global existence, one needs to exploit the conservation laws of the $L^2$-norm and of the energy. For the latter it is required to have more regularity, namely it must be $\sigma \geq 1$.

\begin{lemma}[Conservation laws, Global well-posedness]
	Let $\varepsilon > 0$ be fixed. Consider $\varphi \in H^1(\mathbb{R})$ and let $u_{\varepsilon,t}$ denote the solution at time $t$ of equation \eqref{eq:duhamel} with initial data $u_{\varepsilon,0} = \varphi,$ whose existence is guaranteed by Lemma \ref{lemma:localex}. Then, the following statements hold:
	\begin{enumerate}
		\item The $L^2$-norm is conserved by the flow, i.e. $\| u_{\varepsilon,t} \| = \| \varphi \|.$
		
		\item  The energy
		\begin{equation}\label{energy}
			E_{\eps}(u_{\varepsilon,t}) = \frac{1}{2}\|u_{\varepsilon,t}'\|^2 + \frac{\mu}{4} \langle w_\varepsilon, |u_{\varepsilon,t}|^2 \rangle^2
		\end{equation}
		is conserved by the flow, i.e. $E_\eps(u_{\eps,t}) = E_\eps (\varphi).$
		
		\item The solution $u_{\varepsilon,t}$ can be extended to a global solution.
	\end{enumerate}
\end{lemma}

\begin{proof}
	As a solution to \eqref{eq:duhamel}, the  function $u_{\varepsilon,t}$ satisfies $\partial_t u_{\varepsilon,t} \in H^{-1} (\bR)$. Then, one can compute
	\begin{equation}
		\begin{split}
			\partial_t \| u_{\varepsilon,t} \|^2 \ = \ & 2 \Re \langle u_{\varepsilon,t}, \partial_t u_{\varepsilon,t} \rangle \ = \ 2 \Im \langle u_{\varepsilon,t}, - u_{\varepsilon,t}'' \rangle +
			2 \Im \mu \langle w_\varepsilon, | u_{\varepsilon,t} |^2 \rangle^2
			\ = \  0,
		\end{split}
	\end{equation}
	so conservation of the $L^2$-norm is proven at every time of existence of the solution.
	
	To prove the conservation of energy, let us first suppose more regularity for the initial data, say $\varphi \in H^2 (\bR)$. Then, given the definition \eqref{energy}, straightforward computations yield
	\begin{equation} \label{eq:dtenergy}
		\begin{split}
			\partial_t E_\eps (u_{\varepsilon,t}) \, : = \, & \Re \left( \langle u_{\varepsilon,t}',
			\partial_t u_{\varepsilon,t}' \rangle + \mu \langle w_\varepsilon, | u_{\varepsilon,t} |^2 \rangle
			\langle w_\varepsilon, \overline{ u_{\varepsilon,t}} \partial_t u_{\varepsilon,t} \rangle \right)
		\end{split}
	\end{equation}
	The first term in parentheses can be computed as follows
	\begin{equation} \label{eq:dtenergy1}
		\begin{split}
			\Re \langle u_{\varepsilon,t}',
			\partial_t u_{\varepsilon,t}' \rangle \ = \ & - \Im \langle u_{\varepsilon,t}'', -u_{\varepsilon,t}''
			+ \mu c w_\varepsilon  u_{\varepsilon,t} \rangle
			\ = \ - \mu c \Im \langle u_{\varepsilon,t}'', w_\varepsilon  u_{\varepsilon,t} \rangle
		\end{split}
	\end{equation}
	where we denoted $c:=\langle w_\varepsilon, | u_{\varepsilon,t} |^2 \rangle$ to avoid cumbersome expressions and used the fact that $u_{\varepsilon,t} \in H^2 (\bR)$. On the other hand, for the second term in the parentheses of \eqref{eq:dtenergy} one immediately has
	\begin{equation} \label{eq:dtenergy2}
		\begin{split}
			\mu c \Re \langle w_\varepsilon, \overline{ u_{\varepsilon,t}} \partial_t u_{\varepsilon,t} \rangle \ = \ &
			\mu c \Im  \langle w_\varepsilon, \overline{ u_{\varepsilon,t}} \left( -u_{\varepsilon,t}''
			+ \mu c w_\varepsilon  u_{\varepsilon,t}  \right) \rangle \\
			\ = \ &
			\mu c \Im  \langle w_\varepsilon, - \overline{ u_{\varepsilon,t}}  u_{\varepsilon,t}''\rangle \\
			\ = \ &
			\mu c \Im  \langle  u_{\varepsilon,t}'', w_\varepsilon  u_{\varepsilon,t} \rangle
		\end{split}
	\end{equation}
	Therefore, from \eqref{eq:dtenergy}, \eqref{eq:dtenergy1}, and \eqref{eq:dtenergy2}, one concludes
	\begin{equation}
		\partial_t E_\varepsilon (u_{\varepsilon,t}) \, \equiv \, 0.
	\end{equation}
	Owing to the positivity of the interaction, conservation of energy immediately gives that the solutions are global. Indeed, from \eqref{existencetime} it appears that a solution in $H^1$ can always be prolonged for a time $T$ that depends on the $H^1$-norm of the solution only. Now, due to the expression \eqref{energy} of the conserved energy of the system and to conservation of the $L^2$-norm, one has
	\begin{equation}
		\label{h1estimate}\|  u_{\varepsilon,t} \|_{H^1(\bR)} \ \leq \ \sqrt{2 E_\varepsilon (\varphi) + \| \varphi \|^2}
	\end{equation}
	so that the solution can be always extended for a further time given by
	\begin{equation} \label{eq:extension}
		T_\eps \ = \ \frac{1}{2}\min \left( \frac{1}{C \left(\rho + \sqrt{2 E_\varepsilon (\varphi) + \| \varphi \|^2}
			\right)^2}, \frac{\rho}{C \left(\rho +\sqrt{2 E_\varepsilon (\varphi) + \| \varphi \|^2}
			\right)^3} \right),
	\end{equation}
	which depends on $\varphi$ and $\eps$. Thus, iterating the extension, at each step the solution is extended by the same amount of time $T_\varepsilon$, so globality in time is eventually reached.
	
	To allow for initial data in $H^1(\bR)$ we exploit the continuity of the solution with respect to initial data. Let us drop the auxiliary hypothesis $\varphi \in H^2 (\R)$ and consider a sequence $\varphi^{(n)} \in H^2(\bR)$ such that $\varphi^{(n)} \to \varphi$ in the topology of $ H^1(\bR)$. 
	Denote by $u_{\varepsilon,t}^{(n)}$  the solution to \eqref{eq:duhamel} with initial data  $\varphi^{(n)}$ and as usual by $u_{\eps,t}$ the solution to \eqref{eq:duhamel} with initial data $\varphi$. We stress that the solution $u_{\varepsilon,t}$ is initially defined in an interval $[0,T]$, with  $T$ satisfying the condition \eqref{existencetime}. On the other hand, being solutions in $H^2 (\bR)$, the functions $u_{\varepsilon,t}^{(n)}$ are defined at every $t\in \bR$.
	Furthermore, the quantity $\| u_{\varepsilon,t}^{(n)} \|_{H^1(\bR)}$ can be bounded by a constant $C$ independent of $\eps, \, n$ and $t$:
	\[
		\| u_{\varepsilon,t}^{(n)} \|_{H^1(\bR)}^2
		\  \leq \ 2 E_\eps (\varphi^{(n)}) + \| \varphi^{(n)} \|^2 \ \leq \
		\| (\varphi^{(n)})' \|^2 +\frac \mu 2 \| \varphi^{(n)} \|_\infty^4 + \| \varphi^{(n)} \|^2
		\ \leq \ C,
	\]
	where we used \eqref{h1estimate}, $\int w_\eps \, dx = 1$, and the fact that the $H^1$- convergence of the sequence $\varphi^{(n)}$ entails the boundedness of $\| (\varphi^{(n)})' \|, \, \| \varphi^{(n)} \|_\infty$, and $\| \varphi^{(n)} \|.$
	
	On the other hand, by Lemma \ref{lemma:localex} one has  $\| u_{\varepsilon,t}\|_{H^1(\bR)} \ \leq \ C_\varepsilon$ for every $t \in [0,T]$, but the dependence on $\varepsilon$ cannot be ruled out until we prove conservation of energy.
	
	From \eqref{eq:duhamel} one  has
	\begin{equation*}
		\begin{split}
			\| u_{\varepsilon,t}^{(n)} &- u_{\varepsilon,t} \|_{H^1(\bR)} \\
			\leq \ & \| \varphi^{(n)} - \varphi \|_{H^1(\bR)} + \mu \int_0^t \ds \, \left\| w_\varepsilon \left( \langle w_\varepsilon, |  u_{\varepsilon,s}^{(n)} |^2 \rangle u_{\varepsilon,s}^{(n)} - \langle w_\varepsilon, |  u_{\varepsilon,s} |^2 \rangle u_{\varepsilon,s}
			\right)
			\right\|_{H^1(\bR)} \\
			\ \leq \ & C_n + \mu \|  w_\varepsilon \|_{H^1(\bR)} \int_0^t \ds \, \left( \langle w_\varepsilon, |  u_{\varepsilon,s}^{(n)} |^2\rangle \|
			u_{\varepsilon,s}^{(n)} -  u_{\varepsilon,s} \|_{H^1(\bR)} + \langle w_\varepsilon, |
			|  u_{\varepsilon,s}^{(n)} |^2 - |  u_{\varepsilon,s} |^2  | \rangle \|  u_{\varepsilon,s} \|_{H^1(\bR)}
			\right) \\
			\ \leq \ & C_n + C_\varepsilon  \int_0^t \ds \, \left( |  u_{\varepsilon,s}^{(n)} (0) |^2
			\|
			u_{\varepsilon,s}^{(n)} -  u_{\varepsilon,s} \|_{H^1(\bR)}  + (  \|
			u_{\varepsilon,s}^{(n)}\|_{H^1(\bR)}^2 +  \|
			u_{\varepsilon,s} \|_{H^1(\bR)}^2 )  |
			u_{\varepsilon,s}^{(n)} (0) -  u_{\varepsilon,s} (0)| \right)
			\\
			\ \leq \ & C_n + C_\varepsilon  \int_0^t \ds \, \|
			u_{\varepsilon,s}^{(n)} -  u_{\varepsilon,s} \|_{H^1(\bR)} , \qquad \forall t \in [0,T].
		\end{split}
	\end{equation*}
	Now, by Gr\"onwall's estimate
	\begin{equation*}
		\begin{split}
			\| u_{\varepsilon,t}^{(n)} - u_{\varepsilon,t} \|_{H^1(\bR)} \ \leq \ C_n e^{C_\varepsilon  t}
			, \qquad \forall t \in [0,T]
		\end{split}
	\end{equation*}
	and, since $C_n : = \| \varphi^{(n)} - \varphi \|_{H^1(\bR)}$ vanishes as $n \to \infty$, we conclude that
	\[
		u_{\varepsilon,t}^{(n)} \to u_{\varepsilon,t}, \qquad n \to \infty, \, , \qquad \forall t \in [0,T] . 
	\]
	Finally, since the energy functional is continuous in $H^1(\bR)$, one has
	\[ 
		E (u_{\varepsilon,t}) \ = \ \lim_{n \to \infty}  E (u_{\varepsilon,t}^{(n)}) \ = \
		\lim_{n \to \infty}  E (\varphi^{(n)}) \ = \  E (\varphi), \qquad \forall t \in [0,T],
	\]
	then conservation of energy is proven up to time $T$. It is now possible to proceed as we did for the
	$H^2$-solutions, namely extend the solution $u_{\varepsilon,t}$ in $H^1 (\bR)$ for a time $T$ given by \eqref{eq:extension} infinitely many times, obtaining then a global solution in $H^1 (\bR)$.
\end{proof}

\section{From nonlocal to local nonlinearity}\label{sec:1bodyConvergence}

In the present section we prove that, as the number $N$ of particles grows to infinity and then $\varepsilon$ vanishes, the solution $u_{\varepsilon,t}$ of \eqref{eq:intermediateintro} converges to the solution $\varphi_t$ of the mild version \eqref{eq:integralequationdelta} of the target equation
\eqref{eq:limiteq}. There are some similarities with the analogous result proved in \cite{CFNT14}, but 
first, here the starting equation \eqref{eq:duhamel} is different from that of \cite{CFNT14}, and second, here we aim at estimating the error for large times, whereas in that work the result was given for times of order 1.


Preliminarily, we highlight an important uniform estimate:

\begin{remark} \label{rem:epsuniform}
	For the sake of taking the limit of $u_{\varepsilon,t}$ as $\varepsilon \to 0$ while keeping fixed the initial data $\varphi$,
	it is essential to notice that
	\[
		\| u_{\varepsilon,t} \|_{H^1(\bR)} \ \leq \ \sqrt{ \| u_{\varepsilon,t} \|^2 + 2 E ( u_{\varepsilon,t}) }\ = \ \sqrt{ \| \varphi \|^2 + 2 E ( \varphi) } \ \leq \ \sqrt{ \| \varphi\|^2_{H^1(\bR)} + \frac 1 2
		\| \varphi \|_{\infty}^4},
	\]
	that is a uniform bound in $\varepsilon$ and $t$. As a consequence, one has the uniform estimate
	\[
		\langle w_\eps, |u_{\eps,t}|^2 \rangle \ \leq \ 
		\| u_{\eps,t} \|_\infty^2 \ \leq \ C.
	\]
\end{remark}

Before proving the main result of the section, we give a Gr\"onwall-type inequality, suited for our purposes.
\begin{lemma} \label{gronwalltype}
	Let $w : [0,+ \infty) \to \R$ be a non-negative and continuous function satisfying the bound
	\begin{equation}\label{gronwall-abel}
		w (t) \ \leq \ A (t) + B \int_0^t \d s \, \frac {w(s)}{\sqrt{t-s}},
	\end{equation}
	where $A : [0, + \infty) \to [0, + \infty) $ is  continuous and $B \geq 0.$
	
	Then, the following inequality holds:
	\begin{equation}
		\label{quasigronwall}
		w (t) \ \leq \ D (t) + \pi B^2 e^{\pi B^2 t} \int_0^t \d s \, D(s) e^{- \pi B^2 s} ,
	\end{equation}
	where
	\[
		D (t) = A (t) + B \int_0^t \d s \,\frac{A (s)} {\sqrt{t-s}} .
	\]
\end{lemma}

\begin{proof}
	Iterating the Gr\"onwall-Abel inequality \eqref{gronwall-abel} one finds
	\[
		w(t) \ \leq \ A (t) + B \int_0^t \d s \,\frac{A(s)}{\sqrt{t-s}} + B^2
		\int_0^t \frac{\d s}{\sqrt{t-s}} \int_0^s \d s' \,\frac{w (s')}{\sqrt{s-s'}}.  
	\]
	Exchanging the integrals in $s$ and $s'$ and recalling that
	\[ 
		\int_{s'}^t  \frac{\d s}{\sqrt{t-s}\sqrt{s-s'}} \ = \ \pi,
	\]
	one gets
	\[
		w(t) \ \leq \ D (t) + \pi B^2 \int_0^t \d s' \, w (s')
	\]
	that, by the classical Gr\"onwall inequality, yields the result.
\end{proof}

\noindent
We prove the following 
\begin{theorem}[Convergence of the one-body dynamics]\label{thm:Convergence}
	Let $\varphi \in H^1 (\R)$ and $w$ in the Schwartz class. Then, for any $t > 0$,
	\begin{equation}
		\label{eq:onedconvergence}
		\| u_{\varepsilon, t} - \varphi_t \|^2_{L^2 (\R)} \ \leq \ C \varepsilon^{2 \eta}  \, e^{Ct},
		\qquad \forall \eta, \, t , \ \, 0 < \eta < \frac 1 2, \, t \geq 0,
	\end{equation}
	where $C$ is independent of $\varepsilon$ and $t$.
\end{theorem}

\begin{proof}
	First, we split $u_{\varepsilon,t} - \varphi_t$ into \textcolor{black}{four} terms, namely
	\begin{equation} \begin{split}
			u_{\varepsilon,t} - \varphi_t & =  \mathrm{(I)} + \mathrm{(II)} + \mathrm{(III)} + \mathrm{(IV)},
		\end{split}
	\end{equation}
	where
	\begin{equation} \begin{split}\label{eq123}
			\mathrm{(I)} & \ = \ - \ii \mu \int_0^t \ds \, U (t-s)
			\left( w_\varepsilon u_{\varepsilon,s}
			\langle w_\varepsilon, | u_{\varepsilon,s}|^2 \rangle - \delta
			u_{\varepsilon,s}
			\langle w_\varepsilon, | u_{\varepsilon,s}|^2 \rangle \right)
			\\
			\mathrm{(II)} & \ = \ - \ii \mu \int_0^t \ds \, U (t-s)
			\left( \delta
			u_{\varepsilon,s}
			\langle w_\varepsilon, | u_{\varepsilon,s}|^2 \rangle
			- \delta
			u_{\varepsilon,s} |u_{\varepsilon,s} (0)|^2 \right) \\
			\mathrm{(III)}  & \ = \ - \ii \mu \int_0^t \ds \, U (t-s)
			\left( \delta
			u_{\varepsilon,s} |u_{\varepsilon,s} (0)|^2  - \delta \varphi_s | \varphi_s (0) |^2 \right) \\
			\mathrm{(IV)}  & \ = \ - \ii \mu \int_0^t \ds \, U (t-s)
			\left( \delta \varphi_s | \varphi_s |^2  - \delta \varphi_s | \varphi_s (0) |^2 \right).
		\end{split}
	\end{equation}
	We preliminarily show that
	that $\mathrm{(IV)} = 0$:
	\begin{equation}
		\begin{split}
			\mathrm{(IV)}&=-\ii\mu\int_{0}^{t}\ds\,U(t-s)\left(\delta\varphi_{s}|\varphi_{s}|^{2}-\delta\varphi_{s}|\varphi_{s}(0)|^{2}\right)\\
			&=-\ii\mu\int_{0}^{t}\ds\,\varphi_{s}\left(\int_{\bR}\dy\,U(t-s,x-y)\left(\delta(y)|\varphi_{s}(y)|^{2}-\delta(y)|\varphi_{s}(0)|^{2}\right)\right)\\
			&=-\ii\mu\int_{0}^{t}\ds\,\varphi_{s}\left(U(t-s,x)\left(|\varphi_{s}(0)|^{2}-|\varphi_{s}(0)|^{2}\right)\right)\\
			&=0,
		\end{split}
	\end{equation}
	thus the quantity we have to estimate reduces to 
	$\mathrm{(I)+(II)+(III)}.$

	Let us first estimate $\mathrm{(I)}$. To this aim, we use the strategy outlined in \cite{CFNT17}:
	\begin{equation} \begin{split} \label{eq:IaIb}
			\mathrm{(I)}&\ =\ -\ii
			\mu\int_{0}^{t}\ds\,\langle w_{\varepsilon},|u_{\varepsilon,s}|^{2}\rangle\int_{\R}\dy\,U(t-s,x-y)\,\big(w_{\varepsilon}(y)u_{\varepsilon,s}(y)-\delta(y)u_{\varepsilon,s}(y)\big)\\
			&\ =\ -\ii \mu\int_{0}^{t}\ds\,\langle w_{\varepsilon},|u_{\varepsilon,s}|^{2}\rangle\left(\Big(\int_{\R}\dy\, U(t-s,x-y)w_{\varepsilon}(y)u_{\varepsilon,s}(y)\Big)-U(t-s,x)u_{\varepsilon,s}(0)\right)\\
			&\ =\ -\ii \mu\int_{0}^{t}\ds\,\langle w_{\varepsilon},|u_{\varepsilon,s}|^{2}\rangle\left(\Big(\int_{\R}\dy\, U(t-s,x-\varepsilon y)w(y)u_{\varepsilon,s}(\varepsilon y)\Big)-U(t-s,x)u_{\varepsilon,s}(0)\right)
		\end{split}
	\end{equation}
	where we performed the change of variable $y \mapsto \varepsilon y$ and used the definition of $w_\varepsilon$. Furthermore, since $\int w = 1$, we have
	\begin{equation} \begin{split} \label{eq:IaIb2}
			\mathrm{(I)}  = & \  - \ii \mu \int_0^t \ds \langle w_\varepsilon, |u_{\varepsilon,s} |^2 \rangle \int_\R \dy \, w (y) u_{\varepsilon,s} (\varepsilon y)
			\left( U (t-s, x- \varepsilon y) - U (t-s,x)
			\right) \\
			& \ - \  \ii \mu \int_0^t \ds \langle w_\varepsilon, |u_{\varepsilon,s} |^2 \rangle U (t-s,x) \int_\R \dy \, w (y)
			\left( u_{\varepsilon,s} (\varepsilon y) -
			u_{\varepsilon,s} (0)
			\right) \\
			= & \ \mathrm{(Ia)} + \mathrm{(Ib)}
		\end{split}
	\end{equation}
	We estimate the $L^2$-norm in the variable $x$ of the terms $\mathrm{(Ia)}$ and $\mathrm{(Ib)}$.
	We preliminarily introduce the shorthand notation
	\begin{equation}
		\label{def:f}
		f (y) : = \mu\, w(y)\, u_{\varepsilon,s} (\varepsilon y)\, \langle w_\varepsilon, |u_{\varepsilon,s} |^2 \rangle,
	\end{equation}
	so that
	\begin{equation} \label{eq:Ia-init}
		\begin{split}
			\| \mathrm{(Ia)} \|^2 \ =  \ \int_\R \dx \int_{[0,t]^2} \ds \, \ds' \int_\R \dy \int_\R \dy' \overline{f(y)} f(y')
			& \ \left( \overline{U (t-s,x-\varepsilon y)} - \overline{U (t-s,x)} \right) \times \\
			& \ \times \left( {U (t-s',x-\varepsilon y')} - {U (t-s',x)} \right).
		\end{split}
	\end{equation}
	Since $f$ does not depend on $x$, by exchanging the integrals one is led to compute 
	\[
		\int_\R \left( \overline{U (t-s,x-\varepsilon y)} - \overline{U (t-s,x)} \right) \, \left( {U (t-s',x-\varepsilon y')} - {U (t-s',x)} \right) \, \dx. 
	\]
	We notice that
	\[
		U (\tau, \eta) = \frac{1}{\sqrt{2 \pi}} \int_\R \dk \, e^{\ii k \eta} \, e^{-\ii k^2 \tau}.
	\]
	Then, by standard computations
	\begin{align*}
		&\int_\R \overline{U (t-s,x-\varepsilon y)}  \, {U (t-s',x-\varepsilon y')}  \,
		\dx \\
		& = \frac{1}{2\pi} \int \dx \, \dk \, \dxi \, e^{-ik (x - \varepsilon y)} \, e^{\ii k^2 (t-s)} \, e^{\ii  \xi (x- \varepsilon y')} \, e^{-\ii \xi^2 (t-s')} \\
		& = \int \dk \, e^{\ii  \varepsilon k(y- y')} \, e^{-\ii  k^2 (s-s')} \\
		& = \sqrt{2\pi} \, U (s-s', \varepsilon (y -y')),
	\end{align*}
	thus
	\begin{eqnarray}
		\int_\R  \dx\, \overline{U (t-s,x-\varepsilon y)} \,   {U (t-s',x)}
		& = & \sqrt{2\pi} \, U (s-s', \varepsilon y) \nonumber \\
		\int_\R \dx\, \overline{U (t-s,x)} \,   {U (t-s',x - \varepsilon y')}
		& = & \sqrt{2\pi} \, U (s-s', -\varepsilon y') \nonumber \\
		\int_\R \dx\, \overline{U (t-s,x)} \,   {U (t-s',x)}
		& = & \sqrt{2\pi} \, U (s-s', 0). \label{comput-propag}
	\end{eqnarray}
	Plugging the previous estimates in  \eqref{eq:Ia-init}, one gets
	\begin{equation}
		\label{eq:Ia-est1} \begin{split}
			\| \mathrm{(Ia)} \|^2 \ = \ & \sqrt{2\pi}\int_{[0,t]^2} \ds \ds' \int \dy \, \dy' \, \overline{f(y)} f(y') \bigg( U(s-s',\varepsilon (y-y')) - U (s-s', \varepsilon y) \\
			&\hspace{7cm} - U (s-s', - \varepsilon y') + U (s-s',0) \bigg) \\
			\leq \ &  \sqrt{2\pi}\int_{[0,t]^2} \ds \ds' \int \dy \, \dy' \,
			{|f(y)|} |f(y')| \bigg( | U(s-s',\varepsilon (y-y')) - U (s-s', \varepsilon y) | \\
			&\hspace{7cm} + | U (s-s', - \varepsilon y') - U (s-s',0)| \bigg)
		\end{split}
	\end{equation}
	In order to estimate the two last terms we use the elementary inequality
	\[
		| e^{\ii z} - 1 | \leq | z |^\eta, \quad \forall z \in \R, \, 0 < \eta < \frac{1}{2}
	\]
	and obtain
	\begin{equation} \label{expsuboptimal} \begin{split}
			| U(s-s',\varepsilon (y-y')) - U (s-s', \varepsilon y) | & \, \leq \,  C \varepsilon^{2 
				\eta} \frac{ ( 2 |y||y'| + |y'|^2)^{\eta}}{|s-s'|^{\frac 1 2 + 
					\eta}}\, \leq \, C \varepsilon^{2 
				\eta} \frac{|y|^{2 \eta} + |y'|^{2 \eta}}{|s-s'|^{\frac 1 2 + 
					\eta}}
			\\
			| U (s-s', - \varepsilon y') - U (s-s',0)| & \, \leq \, C \varepsilon^{2 
				\eta} \frac{|y'|^{2 \eta}}{|s-s'|^{\frac 1 2 + 
					\eta}}
		\end{split}
	\end{equation}
	so that
	\begin{equation}
		\label{eq:Ia-est2}
		\| \mathrm{(Ia)} \|^2 \ \leq \  C  \varepsilon^{2 \eta} \int_{[0,t]^2} \frac{\ds \, \ds'}{|s-s'|^{\frac 1 2 + \eta}} \int \dy \, \dy' \,
		| f (y) | \, | f (y') | \, (| y |^{2 
			\eta} + | y' |^{2 \eta}).
	\end{equation}
	By \eqref{def:f} and since $\| u_{\varepsilon,s} \|_\infty \, \leq \, C$,
	\[
		| f (y ) | \ \leq \  C | w (y) |
	\]
	and, since $w$ is rapidly decaying, the integral in \eqref{eq:Ia-est2} is finite, so that we conclude
	\begin{equation}\label{eq:Ia-final}
		\| \mathrm{(Ia)} \|^2 \ \leq \  C  \varepsilon^{2 \eta} t^{\frac 3 2 - \eta}.
	\end{equation}
	Let us now estimate the term
	\[
		\mathrm{(Ib)} \ = \   \mu \int_0^t \ds \langle w_\varepsilon, |u_{\varepsilon,s} |^2 \rangle U (t-s,x) \int_\R \dy \, w (y)
		\left( u_{\varepsilon,s} (\varepsilon y) -
		u_{\varepsilon,s} (0)
		\right).
	\]
	Using Cauchy-Schwarz inequality we notice that
	\begin{equation} \label{eq:upoint}
		\begin{split}
			\left| u_{\varepsilon,s}(\varepsilon y) -u_{\varepsilon ,s}(0)
			\right| \  =  \ &  \left| \int_0^{\varepsilon y} \dt \, u'_{\varepsilon,s}(t)
			\right| \ \leq \ \sqrt{ \varepsilon | y |} \, \|  u'_{\varepsilon,s} \|
			\ \leq \ \sqrt{ \varepsilon | y |} \, \|  u_{\varepsilon,s} \|_{H^1}
			\ \leq \ C \, \sqrt{ \varepsilon | y |}\;,
		\end{split}
	\end{equation}
	so that
	\begin{equation} \label{eq:Ib}
		\| \mathrm{(Ib)} \| \ \leq \ C \sqrt \varepsilon
		\left| \int \dy \, \sqrt{|y|} \, w (y)\, \right| 
		\left\| \int_0^t \ds \langle w_\varepsilon, |u_{\varepsilon,s} |^2 \rangle U (t-s,\cdot)
		\right\|.
	\end{equation}

	To estimate the second factor in \eqref{eq:Ib} we pass to the Fourier space:
	\begin{equation} \label{eq:Fresneltype}
		\begin{split}
			\left\| \int_0^t \ds \langle w_\varepsilon, |u_{\varepsilon,s} |^2 \rangle U (t-s,\cdot)
			\right\|^2 \ = \ & \left\| \int_0^t \ds \langle w_\varepsilon, |u_{\varepsilon,s} |^2 \rangle
			e^{-\ii k^2 (t-s)}
			\right\|^2_{L^2_k} \\
			= \ & \int_0^t \ds \, \langle w_\varepsilon, |u_{\varepsilon,s} |^2 \rangle  \int_0^t \ds'
			\langle w_\varepsilon, |u_{\varepsilon,s'} |^2 \rangle \int \dk \,  e^{-\ii k^2 (s-s')} \\
			\leq \ & C \int_0^t \ds  \int_0^t \ds' \frac{\langle w_\varepsilon, |u_{\varepsilon,s} |^2 \rangle
				\langle w_\varepsilon, |u_{\varepsilon,s'} |^2 \rangle}
			{\sqrt{|s-s'|}}  \ \leq \  C  \int_{[0,t]^2} \frac{\ds \, \ds'} {\sqrt{|s-s'|}} 
			\\
			= \ & C \
			t^{\frac 3 2}  \;.
		\end{split}
	\end{equation}
	Therefore, by \eqref{eq:Ib} and \eqref{eq:Fresneltype} one gets
	\[
		\| \mathrm{(Ib)} \|^2 \ \leq \ C  \varepsilon \, t^{\frac 3 2}
	\]
	that, together with \eqref{eq:Ia-final}, yields
	\begin{equation}
		\label{eq:I-final}
		\| \mathrm{(I)} \|^2 \ \leq \  C  \varepsilon^{2 \eta} (t^{\frac 3 2 - \eta} + t^{\frac 3 2}), \qquad 0 < \eta < \frac 1 2.
	\end{equation}
	Let us now estimate term $\mathrm{(II)}$  in \eqref{eq123}. We rewrite it as
	\begin{equation} \label{eq:2rewrite}
		\begin{split}
			\mathrm{(II)} \ = \  &- \ii \mu \int_0^t \ds \int \dy \, U (t-s, x - y) \left( \delta (y) u_{\varepsilon,s} (y)
			\Big( \int \dz\, w_\varepsilon (z) |u_{\varepsilon,s} (z)|^2  \Big) - \delta (y) u_{\varepsilon,s} (y)
			| u_{\varepsilon,s} (0) |^2 \right) \\
			\ = \  &- \ii\mu \int_0^t \ds  \, U (t-s, x)  u_{\varepsilon,s} (0) \left( \Big( \int \dz\,  w_\varepsilon (z) |u_{\varepsilon,s}(z) |^2  \Big) - | u_{\varepsilon,s} (0) |^2 \right) \\
			\ = \ &
			- \ii \mu \int_\R \d s \, U(t-s,x) \, \varsigma_{\eps}^t (s)
			\;,
		\end{split}
	\end{equation}
	where we denoted 
	\[
		\varsigma_{\varepsilon}^t (s) :=  \chi_{[0,t]}
		u_{\varepsilon,s} (0) \left( \Big( \int \dz  w_\varepsilon (z) |u_{\varepsilon,s} (z)|^2 \Big) - | u_{\varepsilon,s} (0) |^2 \right) =\chi_{[0,t]}  u_{\varepsilon,s} (0)  \int \dz\,  w_\varepsilon (z) \left( |u_{\varepsilon,s} (z)|^2   - | u_{\varepsilon,s} (0) |^2 \right) \; ,
	\]
	with $\chi_I$ the characteristic function of the interval $I$.
	
	As a first step we prove a uniform bound for $\varsigma_{\varepsilon}^t:$
	\begin{equation} \label{eq:etauniform}
		\begin{split}
			|\varsigma_{\varepsilon}^t (s)| \ \leq \ & | u_{\varepsilon,s} (0) | \int \dz\,  w_\varepsilon (z) \left| |u_{\varepsilon,s}(z) |^2   - | u_{\varepsilon,s} (0) |^2 \right|  \\
			\ = \ & | u_{\varepsilon,s} (0) |  \int \dz\,  w_\varepsilon (z) \left( |u_{\varepsilon,s} (z)|   + | u_{\varepsilon,s} (0) |\right) \left| |u_{\varepsilon,s} (z)| - | u_{\varepsilon,s} (0) | \right|
			\\
			\ \leq \ & C \| u_{\varepsilon,s} \|_{H^1_x}^2  \int \dz\,  w_\varepsilon (z)\left| u_{\varepsilon,s} (z) -  u_{\varepsilon,s} (0) \right|.
	\end{split} \end{equation}
	Then, proceeding like in \eqref{eq:upoint},
	\begin{equation} \label{cuesta}
		\begin{split}
			|\varsigma_{\varepsilon}^t (s)| 	\ \leq \ & C \| u_{\varepsilon,s} \|_{H^1_x}^3  \int \dz \, w_\varepsilon (z)
			\sqrt{|z|} \\
			\ = \ & C \sqrt{\varepsilon}\| u_{\varepsilon,s} \|_{H^1_x}^3  \int \dz \, w (z)\sqrt{|z|}
			\\
			\ = \ & C \sqrt{\varepsilon},
		\end{split}
	\end{equation}
	Going back to the estimate of the $L^2$-norm of $\mathrm{(II)}$, one obtains
	\begin{equation} 
		\begin{split}
			\| \mathrm{(II)} \|^2 \ = \ & \mu^2 \int_\R \dx \, \left|
			\int_\R \ds \, U(t-s,x) \, \varsigma_\varepsilon^t (s)
			\right|^2 \ = \ \mu^2 \int_\R \dk \, \left| \int_\R \ds \, e^{\ii k^2s} \, \varsigma_\varepsilon^t (s) \right|^2
			\\
			\ = \ & 2 \mu^2 \int_0^\infty 
			\dk \, \left| \int_\R \ds \, e^{\ii k^2s} \, \varsigma_\varepsilon^t (s) \right|^2 \ = \ 
			2 \mu^2 \int_0^1 
			\dk \, \left| \int_\R \ds \, e^{\ii k^2s} \, \varsigma_\varepsilon^t (s) \right|^2 +
			2 \mu^2 \int_1^\infty 
			\dk \, \left| \int_\R \ds \, e^{\ii k^2s} \, \varsigma_\varepsilon^t (s) \right|^2 \\
			\ = \ & C \varepsilon t^2 +
			2 \pi \mu^2 \int_1^\infty \rd\omega \, \frac {\left| \widehat \varsigma_\varepsilon^t (-\omega)\right|^2} {\sqrt{\omega}},
		\end{split}
	\end{equation}  
	where for the first term we used \eqref{eq:etauniform} and the fact that $\varsigma_\varepsilon^t$ vanishes outside the interval $[0,t]$, while in the second 
	we denoted by $\widehat \varsigma_\varepsilon^t$ the Fourier transform of $\varsigma_\varepsilon^t.$ Now, exploiting the Plancherel's formula and the fact that $\omega \geq 1$,
	\begin{equation}
		\begin{split}
			\label{eq:IIestimate}
			\|\mathrm{(II)} \|^2 \ = \ &
			C \varepsilon \, t^2 + 2 \pi \mu^2 \| \varsigma_\varepsilon^t \|^2 \ \leq \
			C \varepsilon (t + t^2),
		\end{split}
	\end{equation}
	where in the last inequality we used \eqref{cuesta} abd that $\varsigma_\varepsilon^t$ is supported in $[0,t]$.
	
	Let us now estimate $\| \mathrm{(III)} \|$. Define the function
	\[
		h_\varepsilon (s) \ : = \ | u_{\varepsilon,s } (0) |^2  u_{\varepsilon,s } (0) - | \varphi_s (0) |^2 \varphi_s (0),
	\] 
	so that
	\[
		{\rm{(III)}} = - \ii  \mu \int_0^t \d s \, U (t-s) \, (\delta u_{\eps,s} |u_{\eps,s} (0)|^2 - \delta \varphi_s |\varphi_s (0)|^2) \ = \ -\ii \mu \int_0^t \d s \, U(t-s,x) h_\eps (s)
	\]
	then, by \eqref{comput-propag} we bound
	\begin{equation} \label{eq:estimIIInorm}
		\begin{split}
			\| \mathrm{(III)} \|^2 \ = \ & \mu^2 \int \dx \left|\int_0^t \ds \,
			U (t-s, x) \, h_\varepsilon (s)
			\right|^2 \\
			\ = \ & \mu^2 \int_{[0,t]^2} \ds \, \d s' \, h_\varepsilon (s) \overline{h_\varepsilon (s')} \int \dx \, U (t-s,x) \overline{U(t-s',x)}
			\\
			\ = \  & \sqrt{2 \pi}\mu^2 \int_{[0,t]^2} \ds \, \ds' \, h_\varepsilon (s) \overline{h_\varepsilon (s')}  \, U (s'-s,0) \\
			\ \leq \ & C \int_{[0,t]^2} \ds \, \ds' \, \frac{|h_\varepsilon (s)||h_\varepsilon (s')|}{\sqrt{|s-s'|}}.
		\end{split}
	\end{equation}
	To estimate  $|h_\varepsilon(s)|$ we first use the triangular inequality and Remark \ref{rem:epsuniform}:
	\begin{equation}
		\label{eq:estimh}
		\begin{split}
			|h_\varepsilon (s) | \ \leq \ &
			\left| | u_{\varepsilon,s } (0) |^2  u_{\varepsilon,s } (0) - | \varphi_s (0) |^2 \varphi_s (0) \right| \\
			\ = \ &
			\left| \left(| u_{\varepsilon,s } (0) |^2  - |\varphi_s (0) |^2 \right)
			u_{\varepsilon,s } (0) + | \varphi_s (0) |^2 (u_{\varepsilon,s } (0) - \varphi_s (0) ) \right| \\
			\ \leq \ &  \left| | u_{\varepsilon,s } (0) |^2  - |\varphi_s (0) |^2 \right|
			| u_{\varepsilon,s } (0)| + | \varphi_s (0) |^2 \left| u_{\varepsilon,s } (0) - \varphi_s (0) \right| \\
			\ \leq \ &  \left| | u_{\varepsilon,s } (0) |  - |\varphi_s (0) | \right|
			\left( | u_{\varepsilon,s } (0) |  + |\varphi_s (0) | \right)
			| u_{\varepsilon,s } (0)| + | \varphi_s (0) |^2 \left| u_{\varepsilon,s } (0) - \varphi_s (0) \right| \\
			\ \leq \ &  \left| u_{\varepsilon,s } (0)   - \varphi_s (0)  \right|
			\left( | u_{\varepsilon,s } (0) |  + |\varphi_s (0) | \right)
			| u_{\varepsilon,s } (0)| + | \varphi_s (0) |^2 \left| u_{\varepsilon,s } (0) - \varphi_s (0) \right| \\
			\ \leq \ & \left(  | u_{\varepsilon,s } (0)|^2 +  | u_{\varepsilon,s } (0)| | \varphi_s (0) |+| \varphi_s (0) |^2 \right)
			\left| u_{\varepsilon,s } (0) - \varphi_s (0) \right| \\
			\ \leq \ & C \left| u_{\varepsilon,s } (0) - \varphi_s (0) \right|,
		\end{split}
	\end{equation}
	then, in order to bound $\left| u_{\varepsilon,s } (0) - \varphi_s (0) \right|$, we evaluate at $x = 0$ and $t=s$ the terms in \eqref{eq123}. We start by $\mathrm{(I)}$. Specializing \eqref{eq:IaIb2} at $x=0$ and $t=s$ we obtain
	\begin{equation} \begin{split} \label{eq:IaIb2zero}
			\mathrm{(I0)}  = & \ -\ii \mu \int_0^s \ds_1 \langle w_\varepsilon, |u_{\varepsilon,s_1} |^2 \rangle \int_\R \dy \, w (y) u_{\varepsilon,s_1} (\varepsilon y)
			\left( U (s-s_1, - \varepsilon y) - U (s-s_1,0)
			\right)  \\
			&  \ -\ii \mu \int_0^s \ds_1 \langle w_\varepsilon, |u_{\varepsilon,s_1} |^2 \rangle \; U (s-s_1,0) \int_\R \dy \, w (y)
			\left( u_{\varepsilon,s_1} (\varepsilon y) -
			u_{\varepsilon,s_1} (0)
			\right) \\
			= & \ \mathrm{(Ia0)} + \mathrm{(Ib0)}
		\end{split}
	\end{equation}
	Proceeding like in \eqref{expsuboptimal} one has
	\[
		\left| U (s-s_1, - \varepsilon y) - U (s-s_1, 0) \right| \ \leq \ C \frac {\varepsilon^{2 \eta} | y|^{2 \eta}} {|s-s_1|^{\frac 1 2 + \eta}}, \qquad 0 < \eta < \frac 1 2, 
	\]
	so we immediately conclude
	\begin{equation}
		\label{eq:estimIazero}
		\left| \mathrm{(Ia0)} \right| \ \leq \ C \varepsilon^{2 \eta} \int_0^s  \ds_1 \, \frac{\| u_{\varepsilon,s_1} \|_{H^1}^3}{|s-s_1|^{\frac{1}2 + \eta}} \int \dy \, |y|^{2 \eta} w (y) \ \leq \ C \varepsilon^{2 \eta} s^{\frac 1 2 - \eta}.
	\end{equation}
	To estimate $|\mathrm{(Ib0)}|$, by \eqref{eq:upoint} we obtain
	\begin{equation}
		\label{eq:estimIbzero}
		|\mathrm{(Ib0)}| \ \leq \ C \sqrt{\varepsilon}
		\int_0^s \ds_1 \, \frac{\| u_{\varepsilon,s_1} \|_{H^1}^3}{\sqrt{s-s_1}}
		\int \dy \, \sqrt{|y|} w (y)
		\ \leq \ C \sqrt {\varepsilon s}.
	\end{equation}
	Let us now estimate $| \mathrm{(II0)} |$, namely the value of $\mathrm{(II)}$ at $x = 0$ and $t=s$. From \eqref{eq:2rewrite} we get
	\[
		\mathrm{(II0)} = -\ii \mu \int_0^s \ds_1 \, U (s-s_1, 0) \, \varsigma_\varepsilon^s (s_1)
	\]
	and by \eqref{eq:etauniform}
	\begin{equation}
		\label{eq:II0estimate}
		| \mathrm{(II0)} | \ \leq \ C \sqrt \varepsilon \int_0^s \frac {\ds_1}{\sqrt{s-s_1}} \ = \ C \sqrt{\varepsilon s}.
	\end{equation}
	It remains to estimate the term (III0) defined as the evaluation of (III) in \eqref{eq123} at 
	$x=0$ and $t=s$, namely
	\[
		\mathrm{(III0)} \ = \ - \ii \mu \int_0^s \ds_1 \, \frac{h_\varepsilon(s_1)}{\sqrt{4 \pi \ii (s-s_1)}}.
	\]
	Thus, collecting \eqref{eq:estimh}, \eqref{eq:estimIazero}, \eqref{eq:estimIbzero}, and \eqref{eq:II0estimate}, we get
	\begin{equation}
		\begin{split}
			| h_\varepsilon (s) | \ \leq \ & C
			\left| u_{\varepsilon,s } (0) - \varphi_s (0) \right| \\ \ \leq \ & C \left( \mathrm{|(Ia0)|+|(Ib0)|+|(II0)|+|(III0)|} \right) \\
			\ \leq \ & C \varepsilon^{2 \eta} s^{\frac 1 2 - \eta} + C {\sqrt{\varepsilon s}} +
			C \int_0^s \ds_1 \, \frac{|h_\varepsilon (s_1)|}{\sqrt{s-s_1}},
			\qquad \eta < \frac 1 2 
			\\ \leq & C \sqrt \eps s^{\frac 1 4}		
			+
			C \int_0^s \ds_1 \, \frac{|h_\varepsilon (s_1)|}{\sqrt{s-s_1}},
		\end{split}
	\end{equation}
	where 
	we set $\eta = \frac 1 4.$
	By Lemma \ref{gronwalltype} with $A(t) = 
	C \sqrt{\varepsilon} t^{\frac 1 4} 
	$, and since by well-known properties of Abel integral operators \cite{Gorenflo-Vessella}
	$D (t) = C \sqrt \varepsilon (t^{\frac 1 4} + t^{\frac 3 4})
	$, one gets
	\begin{equation} \label{esthgronwall}
		|h (s)| \ \leq \ C \sqrt{\varepsilon} \left(s^{\frac 1 4 } + s^{\frac 3 4} +  e^{Cs} \int_0^s ds_1 \,  s_1^{\frac 1 4}  
		e^{-C s_1} \right) \ \leq \ C \sqrt \eps \left( s^{\frac 1 4} e^{Cs} + s^{\frac 3 4} \right) \ \leq \ C \sqrt \eps \, s^{\frac 1 4} e^{Cs}
	\end{equation}
	where we used $s \geq 0$, $\int_0^s e^{-C s_1} \, \d s_1 \leq C$, and possibly modified the value of the constant $C$ from line to line.
	Now, plugging \eqref{esthgronwall} in \eqref{eq:estimIIInorm} for $h(s)$ and for $h(s')$
	yields
	\begin{equation}
		\begin{split}
			\| {\rm{(III)}} \|^2  \ \leq & \ C
			\int_0^t \ds \, | h (s) |
			\int_0^s \ds' \, \frac{|h(s')|}{\sqrt{s-s'}}\\ \ \leq & \ C \eps
			\int_0^t  {\d s} \, s^{\frac 1 4} e^{Cs}\int_0^s \d s' \frac{(s')^{\frac 14 } e^{C s'}}{\sqrt{s-s'}} \\
			\leq & \ C \varepsilon \, t^{2} e^{Ct}
		\end{split}
	\end{equation}
	Together with \eqref{eq:I-final} and \eqref{eq:IIestimate}, the last inequality  gives
	\[
		\| u_{\varepsilon,t} - \varphi_t \|_2^2 \ \leq \ C \varepsilon^{2 \eta} (t + t^2) \, e^{Ct} \qquad 0 < \eta < \frac 1 2, \, t \geq 0.
	\]
	By using the fact that exponential function dominates polynomial, we get the result.
\end{proof}

\section{From the $N$-body to the one-body problem: microscopic derivation}\label{sec:derivation}

This section is the core of the present work, as it is devoted to the  proof of the transition from the $N$-body, linear dynamics driven by \eqref{eq:manybody} to the one-body regime 
identified with the concentrated Hartree equation \eqref{eq:intermediateintro}. As proved in Section \ref{sec:1bodyConvergence}, such a solution converges to the solution of \eqref{eq:limiteq} as $\varepsilon$ vanishes. 


As the structure of the proof is quite involved, for the convenience of the reader we list here the main steps to be followed in the next sections.

\begin{enumerate}[Step 1:]
	
	\item 
	First, in Section \ref{sec:apriori}, we prove an a priori lower bound for the powers of the $N$-body energy. 
	
	\item  Second,
	we introduce the formalism of second quantization. For the sake of self-containedness, we review some established results on Fock spaces and on creation and annihilation operators. 
	This is detailed in Section \ref{sec:Fock_space}.
	
	\item 
	We rephrase the $N$-body dynamics and the reduced density matrix 
	in the formalism of the  Fock space. After introducing coherent states and Weyl operators,  we modify  the reduced density matrix 
	by introducing the coherent states (Section \ref{sec:unitary}).
	
	\item Exploiting the fact that the coherent states are eigenstates of the annihilation operator, we express both creation and annihilation operators in terms of the associated eigenvalues. This method is often referred to as the c-number substitution (see \eqref{eq:introducingU}).
	
	
	\item After the c-number substitution, in order to  prove  Theorem \ref{thm:main} we are led to control the number operator fluctuation dynamics of an alternative generator, related to the dynamics of Weyl operators (outlined in \eqref{eq:derivative_decomposition}).
	
	\item Because this new generator does not conserve the parity of particle numbers and this makes it harder to control the number operator fluctuation dynamics, we further modify the new generator as in \eqref{eq:def_mathcalU}.
	
	\item \label{step:previous} Since the modified generator is still hard to be investigated directly, we consider another modified generator which approximates the modified generator given in \eqref{eq:def_mathcalUM} in the large $N$ limit, see Section \ref{sec:fluctuations}.
	
	\item By controlling the number operator fluctuation dynamics associated with the generator given in Step \ref{step:previous} and providing the vicinity of each approximation to its predecessor, we conclude the proof of the main theorem.
	
\end{enumerate}

\subsection{A priori estimate}
\label{sec:apriori}

We second quantize the system. Then we consider a Fock space evolution
with a coherent state. To do so, we will need to modify either Proposition 3.1 of \cite{kirkpatrick2011derivation} or Proposition
2.1 of \cite{ChenPavlovic2011}.

\begin{lemma}[Lemma 2.2 of \cite{ChenPavlovic2011}]
	\label{lem:PotentialBound_by_Sobolev}
	For any $W:\mathbb{R}\times\mathbb{R}\to\mathbb{R}$ in $L^p(\bR^2)$, the estimate
	\[
	\langle\psi_{1},W(x_{1},x_{2})\psi_{2}\rangle\leq C_{p}\|W\|_{L_{x_{1},x_{2}}^{p}}\|\psi_1\|_{L_{x_{1},x_{2}}^{2}}\|\langle\partial_{x_{1}}\rangle\langle\partial_{x_{2}}\rangle\psi_{2}\|_{L_{x_{1},x_{2}}^{2}}
	\]
	holds for any $p>1$.
\end{lemma}

\begin{proposition}[A priori energy bounds]\label{prop:enest}
	There exists a constant $C>0$, and for every $k$, there exists $N_{0}(k)$ such that for all $N\geq N_{0}(k)$,
	\[
	\langle\Psi,(H_{N}+N)^{k}\Psi\rangle\geq C^{k}N^{k}\langle\Psi,(1-\partial^2_{x_{1}})\dots(1-\partial^2_{x_{k}})\Psi\rangle
	\]
	for all $\Psi\in L_{s}^{2}(\bR^{N})$.
\end{proposition}

\begin{proof}
	We proceed by a two-step induction over $k\geq0$. For $k=0$ the statement is trivial and for $k=1$ it follows from the positivity of the potential. Suppose the claim holds for all $k\leq n$. 
	We prove it holds for $k=n+2$. In fact, from the induction assumption, and using the notation $S_{i}=(1-\partial^2_{x_{i}})^{1/2}$, we find
	\begin{equation}
		\langle\psi,(H_{N}+N)^{n+2}\psi\rangle\geq C^{n}N^{n}\langle\psi,(H_{N}+N)S_{1}^{2}\dots S_{n}^{2}(H_{N}+N)\psi\rangle\,.
	\end{equation}
	Now, writing $H_{N}+N=h_{1}+h_{2}$, with
	\[
	h_{1}=\sum_{j=k+1}^{N}S_{j}^{2}\qquad\text{and } \qquad h_{2}=\sum_{j=1}^{k}S_{j}^{2}+\frac{1}{N}\sum_{i<j}^{N}\varepsilon^{-2}w(\varepsilon^{-1}(x_{i}))\,w(\varepsilon^{-1}(x_{j})),
	\]
	it follows that
	\begin{equation}\begin{split}
			\langle\psi, & (H_{N}+N)^{n+2}\psi\rangle\\
			&\; \geq \;  C^{n}N^{n}\langle\psi,h_{1}S_{1}^{2}\dots S_{n}^{2}h_{1}\psi\rangle\\
			& +C^{n}N^{n}\left(\langle\psi,h_{1}S_{1}^{2}\dots S_{n}^{2}h_{2}\psi\rangle+\langle\psi,h_{2}S_{1}^{2}\dots S_{n}^{2}h_{1}\psi\rangle\right)\\
			&\; \geq \;  C^{n}N^{n}(N-n)(N-n-1)\langle\psi,S_{1}^{2}\dots S_{n+2}^{2}\psi\rangle+C^{n}N^{n}(N-n)\langle\psi,S_{1}^{4}S_{2}^{2}\dots S_{n+1}^{2}\psi\rangle\\
			& +C^{n}N^{n}\frac{(N-n)}{N}\varepsilon^{-2}\sum_{i<j}^{N}\left(\langle\psi,S_{1}^{2}\dots S_{n+1}^{2}w(\varepsilon^{-1}(x_{i}))\,w(\varepsilon^{-1}(x_{j}))\,\psi\rangle+\text{complex conjugate}\right).
	\end{split}\end{equation}
	Because of the permutation symmetry of $\psi$, we obtain
	\begin{equation}\label{eq:bd}\begin{split}
			\langle\psi, & (H_{N}+N)^{n+2}\psi\rangle\\
			&\; \geq \;  C^{n+2}N^{n+2}\langle\psi,S_{1}^{2}\dots S_{n+2}^{2}\psi\rangle+C^{n+1}N^{n+1}\langle\psi,S_{1}^{4}S_{2}^{2}\dots S_{n+1}^{2}\psi\rangle\\
			& +C^{n}N^{n-1}\varepsilon^{-2}(N-n)^{2}(N-n-1)\left(\langle\psi,S_{1}^{2}\dots S_{n+1}^{2}w(\varepsilon^{-1}(x_{n+2}))\,w(\varepsilon^{-1}(x_{n+3}))\,\psi\rangle+\text{c.c.}\right)\\
			=& +C^{n}N^{n-1}\varepsilon^{-2}(N-n)^{2}(n+1)\left(\langle\psi,S_{1}^{2}\dots S_{n+1}^{2}w(\varepsilon^{-1}(x_{1}))\,w(\varepsilon^{-1}(x_{n+2}))\,\psi\rangle+\text{c.c.}\right)\\
			& +C^{n}N^{n-1}\varepsilon^{-2}(N-n)(n+1)n\left(\langle\psi,S_{1}^{2}\dots S_{n+1}^{2}w(\varepsilon^{-1}(x_{1}))\,w(\varepsilon^{-1}(x_{2}))\,\psi\rangle+\text{c.c.}\right).
	\end{split}\end{equation}
	The last three terms are the errors we need to control. 
	First of all, we remark that the first error term is positive, and thus can be neglected (because we assumed $w\geq0$). 
	In fact, since $w(\varepsilon^{-1}(x_{n+2}))\,w(\varepsilon^{-1}(x_{n+3}))\,\psi\rangle$ commutes with all derivatives $S_{1},\dots,S_{n}$, we have
	\[\begin{split}
		&\langle\psi,S_{1}^{2}\dots S_{n+1}^{2} w(\varepsilon^{-1}(x_{n+2}))\,w(\varepsilon^{-1}(x_{n+3}))\,\psi\rangle\psi\rangle\\
		&\quad =\int\mathrm{d}\mathbf{x}\;w(\varepsilon^{-1}(x_{n+2}))\,w(\varepsilon^{-1}(x_{n+3}))\,\psi\rangle|(S_{1}\dots S_{n+1}\psi)(\mathbf{x})|^{2}\geq0\,.
	\end{split}\]
	As for the second error term on the r.h.s. of \eqref{eq:bd}, we bound it from below by
	\begin{equation}\begin{split}
			&C^{n}N^{n-1}\varepsilon^{-2} (N-n)^{2}(n+1)\left(\langle\psi,S_{1}^{2}\dots S_{n+1}^{2}w(\varepsilon^{-1}(x_{1}))\,w(\varepsilon^{-1}(x_{n+2}))\,\psi\rangle+\text{c.c.}\right)\\
			&\; \geq \;  -C(n)N^{n+1}\varepsilon^{-2}\left|\langle\psi,S_{n+1}\dots S_{2}S_{1}\,[S_{1},w(\varepsilon^{-1}(x_{1}))]\,S_{2}\dots S_{n+1}\,w(\varepsilon^{-1}(x_{n+2}))\,\psi\rangle\right|\\
			&\; \geq \;  -C(n)N^{n+1}\varepsilon^{-3}\left|\langle\psi,S_{n+1}\dots S_{2}S_{1}\,w'(\varepsilon^{-1}(x_{1}))\,S_{2}\dots S_{n+1}\,w(\varepsilon^{-1}(x_{n+2}))\,\psi\rangle\right|\\
			&\; \geq \;  -C(n)N^{n+1}\varepsilon^{-3}\left(\mu\langle\psi,S_{n+1}^{2}\dots S_{2}^{2}S_{1}^{2}\,|w(\varepsilon^{-1}(x_{n+2}))|^{2}\,\psi\rangle\right.\\
			& \hspace{4cm}\left.+\mu^{-1}\langle\psi,S_{n+1}\dots S_{2}\,\left|\,w'(\varepsilon^{-1}(x_{1}))\,\right|^{2}S_{2}\dots S_{n+1}\,\psi\rangle\right)\\
			&\; \geq \;  -C(n)N^{n+1}\varepsilon^{-3+\frac{1}{p}}\langle\psi,S_{n+2}^{2}S_{n+1}^{2}\dots S_{2}^{2}S_{1}^{2}\psi\rangle
	\end{split}\end{equation}
	for a constant $C(n)$ independent of $N$. Using Lemma \ref{lem:PotentialBound_by_Sobolev} and
	\[
	\langle\psi,w(x)\psi\rangle\leq C\|w\|_{p}\langle\psi,(1-\partial^2_x)\psi\rangle
	\]
	for every $p>1$, we find
	\begin{equation}\label{eq:err0}\begin{split}
			&C^{n}N^{n-1}\varepsilon^{-2} (N-n)^{2}(n+1)\left(\langle\psi,S_{1}^{2}\dots S_{n+1}^{2}w(\varepsilon^{-1}(x_{1}))\,w(\varepsilon^{-1}(x_{n+2}))\,\psi\rangle+\text{c.c.}\right)\\
			&\; \geq \;  -C(n)N^{n+1}\varepsilon^{-2+\frac{2}{p}}\langle\psi,S_{1}^{4}\dots S_{n+2}^{2}\psi\rangle
	\end{split}\end{equation}
	for arbitrary $\eps>0$. On the other hand, the last term on the r.h.s. of \eqref{eq:bd}, can be controlled by
	\begin{equation}\label{eq:bd2}\begin{split}
			&C^{n}N^{n-1}\varepsilon^{-2} (N-n)(n+1)n\left(\langle\psi,S_{1}^{2}\dots S_{n+1}^{2}w(\varepsilon^{-1}(x_{1}))\,w(\varepsilon^{-1}(x_{2}))\,\psi\rangle+\text{c.c.}\right)\\
			&\; \geq \;  -C(n)N^{n}\varepsilon^{-2}\left|\langle\psi,S_{n+1}\dots S_{2}S_{1}\,[S_{1}S_{2},w(\varepsilon^{-1}(x_{1}))\,w(\varepsilon^{-1}(x_{2}))\,]\,S_{3}\dots S_{n+1}\psi\rangle\right|\\
			&\; \geq \;  -C(n)N^{n}\varepsilon^{-4}\left|\langle\psi,S_{n+1}\dots S_{2}S_{1}\,w'(\varepsilon^{-1}(x_{1}))\,w'(\varepsilon^{-1}(x_{2}))\,S_{3}\dots S_{n+1}\psi\rangle\right|\\
			&\; \geq \;  -C(n)N^{n}\varepsilon^{-4}\bigg(\left|\langle\psi,S_{n+1}^{2}\dots S_{1}^{2}\psi\rangle\right|\\
			& \hspace{3cm}-\langle\psi,S_{n+1}\dots S_{3}\,|w'(\varepsilon^{-1}(x_{1}))\,w'(\varepsilon^{-1}(x_{2}))|^{2}\,S_{3}\dots S_{n+1}\psi\rangle\bigg).
	\end{split}\end{equation}
	The second term is bounded by
	\begin{equation}\label{eq:err1}\begin{split}
			&-C(n) N^{n}\varepsilon^{-4}\langle\psi,S_{n+1}\dots S_{3}\,|w'(\varepsilon^{-1}(x_{1}))\,w'(\varepsilon^{-1}(x_{2}))|^{2}\,S_{3}\dots S_{n+1}\psi\rangle\\
			&\; \geq \;  -C(n)N^{n}\varepsilon^{-4+\frac{2}{p}}\langle\psi,S_{1}^{2}\dots S_{n+1}^{2}\psi\rangle.
	\end{split}\end{equation}
	Inserting \eqref{eq:err1} on the r.h.s. of \eqref{eq:bd2}, we find
	\begin{equation}\label{eq:err4}\begin{split}
			&C^{n}N^{n-1}N^{2\beta} (N-n)(n+1)n\left(\langle\psi,S_{1}^{2}\dots S_{n+1}^{2}w(\varepsilon^{-1}(x_{1}))\,w(\varepsilon^{-1}(x_{n+2}))\,\psi\rangle+\text{c.c.}\right)\\
			&\; \geq \;  -C(n)N^{n}\varepsilon^{-4}\langle\psi,S_{1}^{2}\dots S_{n+1}^{2}\psi\rangle-C(n)N^{n}\varepsilon^{-4+\frac{2}{p}}\langle\psi,S_{1}^{2}S_{2}^{2}\dots S_{n+1}^{2}\psi\rangle.
	\end{split}\end{equation}
	Inserting \eqref{eq:err0} and \eqref{eq:err4} on the r.h.s. of \eqref{eq:bd},
	we see that all error terms can be controlled by the two positive
	contributions, and the proposition follows.
\end{proof}

\subsection{Definitions and properties of Fock space\label{sec:Fock_space}}

To investigate the system of $N$-bosons, we want to embed our the system into bosonic Fock space as in \cite{RodnianskiSchlein2009,ChenLeeSchlein2011}.
The bosonic Fock space is a Hilbert space defined by
\[
\mathcal{F}=\bigoplus_{n\geq0}L^{2}(\bR)^{\otimes_{s}n}=\bC\oplus\bigoplus_{n\geq1}L_{s}^{2}(\bR^{n}),
\]
where $L_{s}^{2}(\bR^{n})$ is a subspace of $L^{2}(\bR^{n})$ that is the space of all functions symmetric under any permutation of $x_{1},x_{2},\dots,x_{n}$. Note that we let $L_{s}^{2}(\bR)^{\otimes0}=\bC$ for convenience. 
An element (or state) $\psi\in\mathcal{F}$ is a sequence $\psi=\{\psi^{\left(n\right)}\}_{n\geq0}$ of $n$-particle wave functions $\psi^{\left(n\right)}\in L_{s}^{2}(\bR^{n})$.
The inner product on $\mathcal{F}$ is defined by
\[\begin{aligned}
	\langle\psi_{1},\psi_{2}\rangle & =\sum_{n\geq0}\langle\psi_{1}^{(n)},\psi_{2}^{(n)}\rangle_{L^{2}(\bR^{n})}\\
	& =\overline{\psi_{1}^{(0)}}\psi_{2}^{(0)}+
	\sum_{n\geq0}\int\dx_{1}\dots\dx_{n}\;\overline{\psi_{1}^{(n)}(x_{1},\dots,x_{n})}\,\psi_{2}^{(n)}(x_{1},\dots,x_{n}).
\end{aligned}\]
The vacuum $\Omega:=\left\{ 1,0,0,\dots\right\} \in\mathcal{F}$ is describing zero particle state. Note that an element $\psi\in\mathcal{F}$ is a many body quantum state which can has uncertainty of number of particles. 
Fock space includes the information of the number of particles. Using the following operators, we can add and remove particle from a state in Fock space. For $f\in L^{2}(\bR)$, we define the creation operator $a^{*}(f)$ and the annihilation operator $a(f)$ on $\mathcal{F}$ by
\begin{equation}
	\left(a^{*}(f)\psi\right)^{(n)}(x_{1},\dots,x_{n})=\frac{1}{\sqrt{n}}\sum_{j=1}^{n}f(x_{j})\psi^{(n-1)}(x_{1},\dots,x_{j-1},x_{j+1},\dots,x_{n})\label{eq:creation}
\end{equation}
and
\begin{equation}\label{eq:annihilation}
	\left(a(f)\psi\right)^{(n)}(x_{1},\dots,x_{n})=\sqrt{n+1}\int\dx\overline{f(x)}\psi^{(n+1)}(x,x_{1},\dots,x_{n}).
\end{equation}
By definition, the creation operator $a^{*}(f)$ is the adjoint of the annihilation operator of $a(f)$, and in particular, $a^{*}(f)$ and $a(f)$ are not self-adjoint. We will use the self-adjoint operator $\phi(f)$ defined as
\begin{equation}\label{eq:phi}
	\phi(f)=a^{*}(f)+a(f).
\end{equation}
We also use operator-valued distributions $a_{x}^{*}$ and $a_{x}$ satisfying
\[
	a^{*}(f)=\int\dx\,f(x)a_{x}^{*},\qquad a(f)=\int\dx\,\overline{f(x)}a_{x}
\]
for any $f\in L^{2}(\bR)$. The canonical commutation relation between these operators is
\[
	[a(f),a^{*}(g)]=\langle f,g\rangle,\quad[a(f),a(g)]=[a^{*}(f),a^{*}(g)]=0,
\]
which also assumes the form
\[
[a_{x},a_{y}^{*}]=\delta(x-y),\quad[a_{x},a_{y}]=[a_{x}^{*},a_{y}^{*}]=0.
\]
Starting from the the creation operator and the annihilation operator, we define other useful operators on $\mathcal{F}$. 
For each non-negative integer $n$, we introduce the projection operator onto the $n$-particle sector of the Fock space,
\[
	P_{n}(\psi):=(0,0,\dots,0,\psi^{(n)},0,\dots)
\]
for $\psi=(\psi^{(0)},\psi^{(1)},\dots)\in\mathcal{F}$. 
For simplicity, with a slight abuse of notation, we will use $\psi^{(n)}$ to denote $P_{n}\psi$. 
The {\em number operator} $\mathcal{N}$ is given by
\begin{equation}
	\mathcal{N}=\int\dx\,a_{x}^{*}a_{x},\label{eq:number operator}
\end{equation}
and it satisfies $\left(\mathcal{N}\psi\right)^{(n)}=n\psi^{(n)}$.
In general, for an operator $J$ defined on the one-particle sector $L^{2}(\bR)$, its second quantization $d\Gamma(J)$ is the operator on $\mathcal{F}$ whose action on the $n$-particle sector is given by
\[
	\left(d\Gamma\left(J\right)\psi\right)^{(n)}=\sum_{j=1}^{n}J_{j}\psi^{(n)}
\]
where $J_{j}=1\otimes\dots\otimes J\otimes\dots\otimes1$ is the operator $J$ acting on the $j$-th variable only. The number operator defined above can be understood as the second quantization of the identity, i.e., $\mathcal{N}=d\Gamma(1)$. With a kernel $J(x;y)$ of the operator $J$, the second quantization $d\Gamma(J)$ can be written as
\begin{equation}\label{eq:dgjop}
	d\Gamma(J)=\int\dx\dy\,J(x;y)a_{x}^{*}a_{y},
\end{equation}
which is consistent with \eqref{eq:number operator}.

Since the annihilation and the creation operator form the number operator, it is natural to control the former by the latter. 
To this purpose, we provide the following lemma.
\begin{lemma}\label{lem:basifock-bounds}
	For $\alpha>0$, let $D(\mathcal{N}^{\alpha})=\{\psi\in\mathcal{F}:\sum_{n\geq1}n^{2\alpha}\|\psi^{(n)}\|^{2}<\infty\}$
	denote the domain of the operator $\mathcal{N}^{\alpha}$. 
	For any $f\in L^{2}(\bR)$ and any $\psi\in D(\mathcal{N}^{1/2})$, we have
	\begin{equation}\label{eq:bd-a}\begin{split}
			\|a(f)\psi\| & \leq\|f\|\,\|\mathcal{N}^{1/2}\psi\|,\\
			\|a^{*}(f)\psi\| & \leq\|f\|\,\|(\mathcal{N}+1)^{1/2}\psi\|,\\
			\|\phi(f)\psi\| & \leq2\|f\|\|\left(\mathcal{N}+1\right)^{1/2}\psi\|\,.
	\end{split}\end{equation}
	Moreover, for any bounded one-particle operator $J$ on $L^{2}(\bR)$ and for every $\psi\in D(\mathcal{N})$, we find
	\begin{equation}
		\|d\Gamma(J)\psi\|\leq\|J\|\|\mathcal{N}\psi\|\,.\label{eq:J-bd}
	\end{equation}
\end{lemma}

\begin{proof}
	See \cite[Lemma 2.1]{RodnianskiSchlein2009} for \eqref{eq:bd-a}.
\end{proof}

\subsection{Unitary operators and their generators}
\label{sec:unitary}

After defining bosonic Fock space (we refer Secion \ref{sec:Fock_space}), we have to define appropriate Hamiltonian $\mathcal H_N$ for our system, that, in some sense, contains the many-body Hamiltonian $H_N$ defined in \eqref{eq:manybodyH}. 
The new Hamiltonian for the Fock space evolution can be written as
\begin{equation}\label{eq:Fock_space_Hamiltonian}\begin{aligned}
		\cH_{N} &=
		\int\dx\,a_{x}^{*}(-\partial^2_x)a_{x}
		+\frac{\mu}{2N}\int\dx\dy\,w_{\varepsilon}(x)w_{\varepsilon}(y)a_{x}^{*}a_{y}^{*}a_{y}a_{x}.
\end{aligned}\end{equation}
Since we have $(\cH_{N}\psi)^{(N)}=H_{N}\psi^{(N)}$ for $\psi\in\mathcal{F}$, \eqref{eq:Fock_space_Hamiltonian} can be justified as an appropriate generalization of \eqref{eq:manybodyH}. The one-particle marginal density $\gamma_{\psi}^{(1)}$ associated with $\psi$ is
\begin{equation}\label{eq:Kernel_gamma}
	\gamma_{\psi}^{(1)}\left(x;y\right)=\frac{1}{\langle \psi,\mathcal{N}\psi\rangle }\langle \psi,a_{y}^{*}a_{x}\psi\rangle .
\end{equation}
Note that $\gamma_{\psi}^{(1)}$ is a trace class operator on $L^{2}(\bR)$ and $\text{Tr }\gamma_{\psi}^{(1)}=1$.

Heuristically, if $\psi=\psi^{(N)}\in\mathcal{F}$ were an eigenvector of $a_{x}$ with the eigenvalue $\sqrt{N}\varphi(x)$, then from \eqref{eq:Kernel_gamma} we get $\gamma_{\psi}^{(1)}(x;y)=\varphi(x)\overline{\varphi(y)}$, which is exactly the same with the one-particle marginal density associated with the factorized wave function $\varphi^{\otimes N}$. 
Even though the eigenvectors of the annihilation operator do not have a fixed number of particles, they still can be utilized for our goal. 
These eigenvectors are known as the coherent states, defined by
\[
	\psi_{\mathrm{coh}}(f)=e^{-\left\Vert f\right\Vert ^{2}/2}\sum_{n\geq0}\frac{\left(a^{*}\left(f\right)\right)^{n}}{n!}\Omega=e^{-\left\Vert f\right\Vert ^{2}/2}\sum_{n\geq0}\frac{1}{\sqrt{n!}}f^{\otimes n}.
\]
Here, for the ease of notation, when we say a function $\psi^{(N)}\in L^{2}(\bR^{n})$ is a function the the Fock space $\mathcal{F}$, we mean that $\psi^{(N)}=(0,0,\dots,0,\psi^{(N)},0,\dots)\in\mathcal{F}$.
For example, we used $f^{\otimes n}$ to denote
\[
	(0,0,\dots,0,f^{\otimes n},0,\dots)\in\mathcal{F}
\]
whose only nonzero component, $f^{\otimes n}$, is in the $n$-particle sector of the Fock space. Closely related to the coherent states is
the Weyl operator. For $f\in L^{2}(\bR)$, the Weyl operator $\cW\left(f\right)$ is defined by
\[
	\cW(f):=\exp\left(a^{*}(f)-a(f)\right)
\]
and it also satisfies
\[
	\cW(f)=e^{-\left\Vert f\right\Vert ^{2}/2}\exp\left(a^{*}(f)\right)\exp\left(-a(f)\right),
\]
which is known as the Hadamard lemma in Lie algebra. 
The coherent state can also be generated by acting Weyl operator to the vacuum as
\begin{equation}\label{Weyl_f}
	\psi_{\mathrm{coh}}(f)=\cW(f)\Omega=e^{-\left\Vert f\right\Vert ^{2}/2}\exp\left(a^{*}(f)\right)\Omega=e^{-\left\Vert f\right\Vert ^{2}/2}\sum_{n\geq0}\frac{1}{\sqrt{n!}}f^{\otimes n}.
\end{equation}

We collect the useful properties of the Weyl operator and the coherent states in the following lemma.
\begin{lemma} \label{lem:Basic_Weyl} 
	Let $f,g\in L^{2}(\bR)$.
	\begin{enumerate}
		\item The commutation relation between the Weyl operators is given by
		\[
			\cW(f)\cW(g)=\cW(g)\cW(f)e^{-2\ii \cdot\mathrm{Im}\langle f,g\rangle}=\cW(f+g)e^{-\ii \cdot\mathrm{Im}\langle f,g\rangle}.
		\]
		\item The Weyl operator is unitary and satisfies 
		\[
			\cW(f)^{*}=\cW(f)^{-1}=\cW(-f).
		\]
		\item The coherent states are eigenvectors of annihilation operators, i.e.
		\[
			a_{x}\psi(f)=f(x)\psi(f)\quad\Rightarrow\quad a(g)\psi(f)=\langle g,f\rangle_{L^{2}}\psi(f).
		\]
		The commutation relation between the Weyl operator and the annihilation operator (or the creation operator) is thus
		\[
			\cW^{*}(f)a_{x}\cW(f)=a_{x}+f(x)\quad\text{and}\quad \cW^{*}(f)a_{x}^{*}\cW(f)=a_{x}^{*}+\overline{f(x)}.
		\]
		\item The distribution of $\mathcal{N}$ with respect to the coherent state $\psi\left(f\right)$ is Poisson. In particular,
		\[
			\langle \psi(f),\mathcal{N}\psi(f)\rangle =\left\Vert f\right\Vert ^{2},\qquad\langle \psi(f),\mathcal{N}^{2}\psi(f)\rangle -\langle \psi(f),\mathcal{N}\psi(f)\rangle ^{2}=\left\Vert f\right\Vert ^{2}.
		\]
	\end{enumerate}
\end{lemma}

We omit the proof of the lemma, since it can be derived from the definition of the Weyl operator and elementary calculations.

Let
\begin{equation} \label{eq:d_N}
	d_{N}:=\frac{\sqrt{N!}}{N^{N/2}e^{-N/2}}.
\end{equation}
Note that $C^{-1} N^{1/4} \leq d_N \leq C N^{1/4}$ for some constant $C>0$ independent of $N$, which can be easily checked by using Stirling's formula.
\begin{lemma}[{\cite[Lemma 6.3]{ChenLee2011}}]\label{lem:coherent_all}
	There exists a constant $C>0$ independent of $N$ such that, for any $\varphi\in L^{2}(\bR)$ with $\|\varphi\|=1$, we have
	\[
		\left\Vert (\cN+1)^{-1/2}\cW^{*}(\sqrt{N}\varphi)\frac{(a^{*}(\varphi))^{N}}{\sqrt{N!}}\Omega\right\Vert \leq\frac{C}{d_{N}}.
	\]
\end{lemma}

\begin{lemma}[{\cite[Lemma 7.2]{Lee2013}}]\label{lem:coherent_even_odd}
	Let $P_m$ be the projection onto the $m$-particle sector of the Fock space $\mathcal{F}$ for a non-negative integer $m$. Then, for any non-negative integer $k\leq(1/2)N^{1/3}$ and any $\varphi\in L^{2}(\bR)$ with $\|\varphi\|=1$, we have
	\[
		\left\Vert P_{2k}\cW^{*}(\sqrt{N}\varphi)\frac{(a^{*}(\varphi))^{N}}{\sqrt{N!}}\Omega\right\Vert \leq\frac{2}{d_{N}}
	\]
	and
	\[
		\left\Vert P_{2k+1}\cW^{*}(\sqrt{N}\varphi)\frac{(a^{*}(\varphi))^{N}}{\sqrt{N!}}\Omega\right\Vert \leq\frac{2(k+1)^{3/2}}{d_{N}\sqrt{N}}.
	\]
\end{lemma}

Let $\gamma_{N,t}^{(1)}$ be the kernel of the one-particle marginal density associated with the time evolution of the factorized state $\varphi^{\otimes N}$ with respect to the Hamiltonian $\cH_{N}$.
By definition,
\begin{align}
	\gamma_{N,t}^{(1)} & =\frac{\langle e^{-\ii \cH_{N}t}\varphi^{\otimes N},a_{y}^{*}a_{x}e^{-\ii \cH_{N}t}\varphi^{\otimes N}\rangle }{\langle e^{-\ii \cH_{N}t}\varphi^{\otimes N},\mathcal{N}e^{-\ii \cH_{N}t}\varphi^{\otimes N}\rangle }=\frac{1}{N}\langle \varphi^{\otimes N},e^{\ii \cH_{N}t}a_{y}^{*}a_{x}e^{-\ii \cH_{N}t}\varphi^{\otimes N}\rangle \nonumber \\
	& =\frac{1}{N}\left\langle \frac{\left(a^{*}(\varphi)\right)^{N}}{\sqrt{N!}}\Omega,e^{\ii \cH_{N}t}a_{y}^{*}a_{x}e^{-\ii \cH_{N}t}\frac{\left(a^{*}(\varphi)\right)^{N}}{\sqrt{N!}}\Omega\right\rangle .\label{eq:marginal_factorized}
\end{align}
If we use the coherent states instead of the factorized state in \eqref{eq:marginal_factorized} and expand $a_{y}^{*}a_{x}$ around $N\overline{u_{\varepsilon,t}(y)}u_{\varepsilon,t}(x)$, then we are led to consider the operator
\begin{equation}
	\label{eq:introducingU}
	\begin{aligned}
		& \cW^{*}(\sqrt{N}u_{\varepsilon,s})e^{\ii \cH_{N}\left(t-s\right)}(a_{x}-\sqrt{N}u_{\varepsilon,t}(x))e^{-\ii \cH_{N}\left(t-s\right)}\cW(\sqrt{N}u_{\varepsilon,s})\\
		& =\cW^{*}(\sqrt{N}u_{\varepsilon,s})e^{\ii \cH_{N}\left(t-s\right)}\cW(\sqrt{N}u_{\varepsilon,t})a_{x}\cW^{*}(\sqrt{N}u_{\varepsilon,t})e^{-\ii \cH_{N}\left(t-s\right)}\cW(\sqrt{N}u_{\varepsilon,s}).\nonumber
	\end{aligned}
\end{equation}
To understand further the operator $\cW^{*}(\sqrt{N}u_{\varepsilon,t})e^{-\ii \cH_{N}\left(t-s\right)}\cW(\sqrt{N}u_{\varepsilon,s})$, we see that
\begin{equation}
	\ii \partial_{t}\cW^{*}(\sqrt{N}u_{\varepsilon,t})e^{-\ii \cH_{N}\left(t-s\right)}\cW(\sqrt{N}u_{\varepsilon,s})=:\left(\sum_{k=0}^{4}\cL_{k}(t;s)\right)\cW^{*}(\sqrt{N}u_{\varepsilon,t})e^{-\ii \cH_{N}\left(t-s\right)}\cW(\sqrt{N}u_{\varepsilon,s}),\label{eq:derivative_decomposition}
\end{equation}
where $\cL_{k}$ contains $k$ creation and/or annihilation operators. The exact formulas for $\cL_{k}$ are as follows:
\begin{align*}
	\cL_0 (t)&:=
	\frac{N\mu}{2} \int_s^t \mathrm{d}\tau \, \langle w_ \varepsilon , |u_{\varepsilon,t}|^2 \rangle ^2,
	\\
	\cL_1 (t;s)& :=\cL_1  = 0, \\
	\cL_{2}(t;s)&=\int_{\bR}\dx\,a_{x}^{*}(-\partial^2_{x}) a_{x}
	+\mu \langle w_ \varepsilon , |u_{\varepsilon,t}|^2 \rangle \int_{\bR}\dx\,w_{\varepsilon}(x) a_{x}^{*}a_{x}\\
	&\qquad+\mu\int_{\bR\times\bR}\dx\dy\,w_{\varepsilon}(x)w_{\varepsilon}(y)\left(a_{x}^{*}a_{y}^{*}\,u_{\varepsilon,t}(x)u_{\varepsilon,t}(y)+a_{x}a_{y}\,\overline{u_{\varepsilon,t}(x)}\,\overline{u_{\varepsilon,t}(y)}\right),\\
	\cL_{3}(t;s)&=
	\frac{\mu}{\sqrt{N}}\int_{\bR\times\bR}\dx\dy\,w_{\varepsilon}(x)w_{\varepsilon}(y)a_{x}^{*}\left(a_{y}^{*}u_{\varepsilon,t}(y)+a_{y}\overline{u_{\varepsilon,t}(y)}\right)a_{x},\\
	\cL_{4}(t;s)&=
	\frac{\mu}{N}\int_{\bR\times\bR}\dx\dy\,w_{\varepsilon}(x)w_{\varepsilon}(y)a_{x}^{*}a_{y}^{*}a_{y}a_{x}.
\end{align*}

We consider  the time evolution
\begin{equation}
	\ii \partial_{t}\,\cU(t;s)=\cL\,\cU(t;s)
	\quad\text{with}\quad \cL=\sum_{k=2}^{4}\cL_{k}(t;s).
	\label{eq:def_mathcalU}
\end{equation}
Furthermore, we introduce
\begin{equation}\label{eq:hatL_2}\begin{aligned}
		\widehat{\cL}_{2}(t) &:= \int_{\bR} \dx\, a_{x}^{*}\big(-\partial^2_x\big)a_{x}+\mu \langle w_ \varepsilon , |u_{\varepsilon,t}|^2 \rangle \int_{\bR}\dx\,w_{\varepsilon}(x) a_{x}^{*}a_{x}\\
		&\qquad+\mu\int_{\bR\times\bR}\dx\dy\, w_{\varepsilon} (x) w_{\varepsilon} (y)\overline{u_{\varepsilon,t}(x)}u_{\varepsilon,t}(y)a_{y}^{*}a_{x}\\
		& \qquad+  \frac{\mu\,\varepsilon}{2}
		\int_{\bR\times\bR} \dx\dy\,  w_{\varepsilon} (x) w_{\varepsilon} (y)\big(u_{\varepsilon,t}(x)u_{\varepsilon,t}(y)a_{x}^{*}a_{y}^{*}+\overline{u_{\varepsilon,t}(x)}\,\overline{u_{\varepsilon,t}(y)}a_{x}a_{y}\big)
\end{aligned}\end{equation}
and
$\widehat{\cL}:=\widehat{\cL}_{2}+\cL_{4}$. Then we define the unitary operator $\widehat{\cU}(t;s)$ by
\begin{equation}\label{eq:def_mathcalwidehatU}
	\ii \partial_{t}\,\widehat{\cU}(t;s)=
	\widehat{\cL}(t)\,\widehat{\cU}(t;s)\quad\text{and} \quad\widehat{\cU}(s;s)=1.
\end{equation}
Since $\widehat{\cL}$ does not change the parity of the number of particles,
\begin{equation}
	\langle \Omega,\widehat{\cU}^{*}(t;0)\,a_{y}\,\widehat{\cU}(t;0)\Omega\rangle=\langle \Omega,\widehat{\cU}^{*}(t;0)\,a_{x}^{*}\,\widehat{\cU}(t;0)\Omega\rangle=0.
	\label{eq:Parity_Consevation_hat}
\end{equation}

Now, we have the following bounds for $E_{t}^{(1)}(J)$ and $E_{t}^{(2)}(J)$, which will be defined and Proposition \ref{prop:Et1} and Proposition \ref{prop:Et2}:
\begin{proposition} \label{prop:Et1}
	Suppose that the assumptions in Theorem \ref{thm:main} hold. For any compact Hermitian operator $J$ on $L^{2}(\bR)$, let
	\[
		E_{t}^{(1)}(J):=\frac{d_{N}}{N}\left\langle \cW^{*}(\sqrt{N}\varphi)\frac{(a^{*}(\varphi))^{N}}{\sqrt{N!}}\Omega,\cU^{*}(t;0)d\Gamma(J)\,\cU(t;0)\Omega\right\rangle.
	\]
	Then, there exist constants $C,K$ depending only on $\|w\|_{L^{1}(\bR)}$ and $\|w\|_{L^{2}(\bR)}$ such that
	\[
		|E_{t}^{(1)}(J)|\leq \frac{C}{N}\Big(1+\frac{1}{\varepsilon^{\frac{7}{2}}N}\Big) \Big(1+\frac{1}{\varepsilon^{\frac{37}{2}}N}\Big) \exp\left(K (1+\varepsilon^{-1}) t \right) \|J\|,
	\]
\end{proposition}

\begin{proposition}\label{prop:Et2}
	Suppose that the assumptions in Theorem \ref{thm:main} hold. For any compact Hermitian operator $J$ on $L^{2}(\bR)$, let
	\[
		E_{t}^{(2)}(J):=\frac{d_{N}}{\sqrt{N}}\left\langle 	\cW^{*}(\sqrt{N}\varphi)\frac{(a^{*}(\varphi))^{N}}{\sqrt{N!}}\Omega,\cU^{*}(t;0)\phi(Ju_{\varepsilon,t})\,\cU(t;0)\Omega\right\rangle.
	\]
	Then, there exist constants $C$, $K$ depending only on $\|w\|_{L^{1}(\bR)}$ and $\|w\|_{L^{2}(\bR)}$ such that
	\[
		|E_{t}^{(2)}(J)|\leq
		\frac{C e^{Kt}}{N} \Big(1+\varepsilon^{-3/2}\Big) \|J\|.
	\]
\end{proposition}
The proofs of these propositions will be given in Section \ref{sec:proof-of-Et1-Et2}.

\subsection{Fluctuation dynamics}
\label{sec:fluctuations}

The main goal of this subsection is to provide some important lemmata to prove Propositions \ref{prop:Et1} and \ref{prop:Et2}.

First, we introduce a truncated time-dependent generator with fixed $M>0$ as follows:
\begin{align*}
	\cL_{N}^{(M)}(t) & =\int_{\bR}\dx\,a_{x}^{*}(-\partial^2_{x})a_{x}\\
	&\quad+\mu\langle w_ \varepsilon , |u_{\varepsilon,t}|^2 \rangle\int_{\bR}\dx\,w_ \varepsilon (x)a_{x}^{*}a_{x}+\mu \int_{\bR\times\bR}\dx\dy\, w_ \varepsilon (x)w_ \varepsilon (y)\overline{u_{\varepsilon,t}(x)}u_{\varepsilon,t}(y)a_{y}^{*}a_{x}\\
	& \quad+\frac{\mu}{2}
	\int_{\bR\times\bR}\dx\dy\,w_ \varepsilon (x)w_ \varepsilon (y)\left(u_{\varepsilon,t}(x)u_{\varepsilon,t}(y)a_{x}^{*}a_{y}^{*}+\overline{u_{\varepsilon,t}(x)}\overline{u_{\varepsilon,t}(y)}a_{x}a_{y}\right)\\
	& \quad+\frac{\mu}{\sqrt N}\int_{\bR\times\bR}\dx\dy\,w_ \varepsilon (x)w_ \varepsilon (y)a_{x}^{*}\left(\overline{u_{\varepsilon,t}(y)}a_{y}\chi(\mathcal{N}\leq M)+u_{\varepsilon,t}(y)\chi(\mathcal{N}\leq M)a_{y}^{*}\right)a_{x}\\
	& \quad+\frac{\mu}{2N}\int_{\bR\times\bR}\dx\dy\,w_ \varepsilon (x)w_ \varepsilon (y)a_{x}^{*}a_{y}^{*}a_{y}a_{x}.
\end{align*}
We remark that later, in the proof of Lemma \ref{lem:NjU}, $M$ will be chosen to be $M=N \varepsilon^3 $. Define a unitary operator $\cU^{(M)}_{N}$ by
\begin{equation}\label{eq:def_mathcalUM}
	\mathrm{i}\partial_{t}\cU^{(M)}_{N}\left(t;s\right)=\cL^{(M)}_N(t)\,\cU^{(M)}_{N}(t;s)\quad\text{and}\quad\cU^{(M)}_{N}\left(s;s\right)=1.
\end{equation}
Then wee have the following result.

\begin{lemma}\label{lem:UMNUM}
	Suppose that the assumptions in Theorem \ref{thm:main} hold and let $\,\cU^{(M)}_{N}$ be the unitary operator defined in \eqref{eq:def_mathcalUM}. Then, for any $j\in\mathbb N$ there exists a constant $K=K(\|w\|_1, \|w\|_2, j)>0$ such that for all $N\in\mathbb{N}$, $M>0$, $\psi\in\mathcal{F}$, and $t,s\in\mathbb{R}$,
	\[
		\left\langle \cU^{(M)}_{N}(t;s)\psi,(\cN+1)^{j}\cU^{(M)}_{N}(t;s)\psi\right\rangle
		\leq
		C\langle\psi,(\cN+1)^{j}\psi\rangle \exp\left(K\left({ \varepsilon^{-1}}+\Big(\frac{M}{N{\varepsilon^3}}\Big)^{1/2} \right)t\right).
	\]
\end{lemma}

\begin{proof}
	Following the proof of \cite[Lemma 3.5]{RodnianskiSchlein2009}, see \cite[(3.15)]{RodnianskiSchlein2009}, we have
	\begin{equation}\label{eq:ddtUM}
		\frac{\mathrm{d}}{\dt}\langle \cU^{(M)}_{N}(t;0)\psi,(\cN+1)^{j}\cU^{(M)}_{N}(t;0)\psi\rangle=\mu(A+B),
	\end{equation}
	where
	\begin{align*}
		A
		&:=\mu\sum_{k=0}^{j-1}{\binom{j}{k}}(-1)^{k}\ImPart \int_{\bR\times\bR}\dx\dy\,w_ \varepsilon (x)w_ \varepsilon (y)u_{\varepsilon,t}(x)u_{\varepsilon,t}(y) \\
		&\qquad\times\langle \cU^{(M)}_{N}(t;0)\psi,\left(\cN^{k/2}a_{x}^{*}a_{y}^{*}(\cN+2)^{k/2}+(\cN+1)^{k/2}a_{x}^{*}a_{y}^{*}(\cN+3)^{k/2}\right)\,\cU^{(M)}_{N}(t;0)\psi\rangle\\
		B&:=\frac{2}{\sqrt{N}} \sum_{k=0}^{j-1}{\binom{j}{k}}\ImPart \int_{\bR}\dx\,w_ \varepsilon (x)
		\langle \cU^{(M)}_{N}(t;0)\psi,a_{x}^{*}a(w_ \varepsilon \left(\cdot)\varphi_t\right)\chi(\cN \leq M)(\cN+1)^{k/2}a_{x}\cN^{k/2}\cU^{(M)}_{N}(t;0)\psi\rangle.
	\end{align*}
	To control the contribution from the first term on the right-hand side of \eqref{eq:ddtUM}, we use bounds of the form
	\begin{align*}
		&\left|\int_{\bR\times\bR}\dx\dy\,w_ \varepsilon (x)w_ \varepsilon (y)u_{\varepsilon,t}(x)u_{\varepsilon,t}(y)\langle\cU_{N}^{(M)}(t;0)\psi,(\cN+1)^{\frac{k}{2}}a_{x}^{*}a_{y}^{*}(\cN+3)^{\frac{k}{2}}\cU_{N}^{(M)}(t;0)\psi\rangle\right|\\
		&=\left|\int_{\bR\times\bR}\dx\dy\,w_ \varepsilon (x)w_ \varepsilon (y)u_{\varepsilon,t}(x)u_{\varepsilon,t}(y)\langle a_{x}(\cN+1)^{\frac{k}{2}-\frac{1}{2}}\cU_{N}^{(M)}(t;0)\psi,a_{y}^{*}(\cN+3)^{\frac{k}{2}+\frac{1}{2}}\cU_{N}^{(M)}(t;0)\psi\rangle\right|\\
		&=\left|\langle a\big( w_\varepsilon u_{\varepsilon,t}\big)(\cN+1)^{\frac{k}{2}}\cU_{N}^{(M)}(t;0)\psi,a^{*}\big( w_\varepsilon u_{\varepsilon,t}\big)(\cN+3)^{\frac{k}{2}}\cU_{N}^{(M)}(t;0)\psi\rangle\right|\\
		&\leq\left\Vert a\big( w_\varepsilon u_{\varepsilon,t}\big)(\cN+1)^{\frac{k}{2}}\cU_{N}^{(M)}(t;0)\psi\right\Vert \,\left\Vert a^{*}\big( w_\varepsilon u_{\varepsilon,t}\big)(\cN+1)^{\frac{k}{2}}\cU_{N}^{(M)}(t;0)\psi\right\Vert \\&\leq\| w_\varepsilon u_{\varepsilon,t}\|_{2}^{2}\,\left\Vert (\cN+1)^{\frac{k+1}{2}}\cU_{N}^{(M)}(t;0)\psi\right\Vert ^{2}\\
		&\leq \varepsilon^{-1}\|u_{\varepsilon,t}\|_{\infty}^{2}\,\left\Vert (\cN+1)^{\frac{k+1}{2}}\cU_{N}^{(M)}(t;0)\psi\right\Vert ^{2},
	\end{align*}
	where the last inequality comes from
	\[
		\| w_\varepsilon u_{\varepsilon,t}\|^{2}=\int\dx\,|w_ \varepsilon (x)u_{\varepsilon,t}(x)|^{2}
		\leq
		\varepsilon^{-1}\int\dx\,w_ \varepsilon (x)|u_{\varepsilon,t}(x)|^{2}
		\leq \varepsilon^{-1}\|u_{\varepsilon,t}\|_{\infty}^{2}
	\]
	because $\|w_\varepsilon\|_1=\|w\|_1=1$.
	
	On the other hand, to control the second integral on the right-hand side of \eqref{eq:ddtUM}, we use that
	\begin{align*}
		&\left|\int_{\bR}\dx\,w_ \varepsilon (x)\langle\cU_{N}^{(M)}(t;0)\psi,a_{x}^{*}a(w_ \varepsilon \left(\cdot)u_{\varepsilon,t}\right)\chi(\cN\leq M)(\cN+1)^{k/2}a_{x}\cN^{k/2}\cU_{N}^{(M)}(t;0)\psi\rangle\right|\\
		&\quad\leq\int_{\bR}\dx\,w_ \varepsilon (x)\Vert a_{x}(\cN+1)^{\frac{k}{2}}\cU_{N}^{(M)}(t;0)\psi\Vert\Vert a( w_\varepsilon u_{\varepsilon,t})\chi(\cN\leq M)\Vert\Vert a_{x}\cN^{\frac{k}{2}}\cU_{N}^{(M)}(t;0)\psi\Vert\\
		&\quad\leq M^{1/2}\,\varepsilon^{-1} \| w_\varepsilon u_{\varepsilon,t}\|\|(\cN+1)^{\frac{k+1}{2}}\cU_{N}^{(M)}(t;0)\psi\|^{2}\\
		&\quad\leq M^{1/2}\,\varepsilon^{-3/2}\,\|u_{\varepsilon,t}\|_{\infty}\,\|(\cN+1)^{\frac{k+1}{2}}\cU_{N}^{(M)}(t;0)\psi\|^{2}.
	\end{align*}
	Then,
	\begin{align*}
		&\left\vert\frac{\mathrm{d}}{\dt}\langle \cU^{(M)}_{N}(t;0)\psi,(\cN+1)^{j}\cU^{(M)}_{N}(t;0)\psi\rangle\right\vert\\
		&\leq 2\sum_{k=0}^{j-1}{\binom{j}{k}} 2 \varepsilon^{-1} \|u_{\varepsilon,t}\|_{\infty}^{2}\|(\cN+3)^{\frac{k+1}{2}}\cU^{(M)}_{N}(t;0)\psi\|^{2}\\
		&\qquad+\frac{2}{\sqrt{N}} \sum_{k=0}^{j-1}{\binom{j}{k}}M^{1/2} \varepsilon^{-3/2}\,\|u_{\varepsilon,t}\|_{\infty}\,\|(\cN+1)^{\frac{k+1}{2}}\cU^{(M)}_{N}(t;0)\psi\|^{2}\\
		&\leq 2 \sum_{k=0}^{j-1}{\binom{j}{k}}\left(2 \varepsilon^{-1}\|u_{\varepsilon,t}\|_{\infty}^{2}+\frac{1}{\sqrt{N}}M^{1/2}\varepsilon^{-3/2} \,\|u_{\varepsilon,t}\|_{\infty}\right)\|\left(\cN+3\right)^{\frac{j}{2}}\cU^{(M)}_{N}(t;0)\psi\|^{2}\\
		&\leq 2\cdot2^j\left(2 \varepsilon^{-1} \|u_{\varepsilon,t}\|_{\infty}^{2}+\Big(\frac{M}{N\,\varepsilon^3}\Big)^{1/2} \,\|u_{\varepsilon,t}\|_{\infty}\right)\Big\|\left(3\left(\cN+1\right)\right)^{\frac{j}{2}}\cU^{(M)}_{N}(t;0)\psi\Big\|^{2}\\
		&\leq 2\cdot6^j\left(2 \varepsilon^{-1} \|u_{\varepsilon,t}\|_{\infty}^{2}+\Big(\frac{M}{N\,\varepsilon^3}\Big)^{1/2}\,\|u_{\varepsilon,t}\|_{\infty}\right)\Big\langle \cU^{(M)}_{N}(t;0)\psi,(\cN+1)^{j}\cU^{(M)}_{N}(t;0)\psi\Big\rangle.
	\end{align*}
	Applying the Grönwall lemma, one gets
	\begin{align*}
		&\Big\langle\cU_{N}^{(M)}(t;0)\psi,(\cN+1)^{j}\cU_{N}^{(M)}(t;0)\psi\Big\rangle\\
		&\leq\langle\psi,(\cN+1)^{j}\psi\rangle\exp\left(\int_{0}^{t}\ds\,2\cdot6^{j}
		\left(2\varepsilon^{-1} \|u_{\varepsilon,s}\|_{\infty}^{2}+\Big(\frac{M}{N\,\varepsilon^3}\Big)^{1/2}\,\|u_{\varepsilon,s}\|_{\infty}\right)
		\right)\\
		&\leq\langle\psi,(\cN+1)^{j}\psi\rangle\exp\left(C_{j}\int_{0}^{t}\ds\left(\varepsilon^{-1}\|u_{\varepsilon,s}\|_{\infty}^{2}+\Big(\frac{M}{N \varepsilon^3}\Big)^{1/2}\,\|u_{\varepsilon,s}\|_{\infty}\right)\right)\\
		&\leq\langle\psi,(\cN+1)^{j}\psi\rangle \exp\left(K\left(\varepsilon^{-1} +\Big(\frac{M}{N \varepsilon^3}\Big)^{1/2}\right)t\right).
	\end{align*}
	Here, for the last inequality, we have used $\|u_{\varepsilon,s}\|_\infty \leq \|u_{\varepsilon,s}\|_{H^{1}}<C $.
	The proof is complete.
\end{proof}

\begin{lemma}\label{lem:NjU}
	Suppose that the assumptions in Theorem \ref{thm:main} hold. Let $\cU\left(t;s\right)$ be the unitary evolution defined in \eqref{eq:def_mathcalU}. Then for any $\psi\in\mathcal{F}$ and $j\in\mathbb{N}$, there exists a constant $C \equiv C(j,  \|w\|_1)>0$ such that
	\[
		\langle \cU(t;s)\psi,\cN^{j}\cU(t;s)\psi\rangle
		\leq C\Big(1+\frac{1}{\varepsilon^{\frac{1}{2}+3j}N}\Big)\,\langle \psi,\left(\cN+1\right)^{2j+2}\psi\rangle \exp\left(K(1+\varepsilon^{-1}) t\right).
	\]
\end{lemma}

In this section we modify the proof given in previous articles, for example, \cite{RodnianskiSchlein2009} or \cite{ChenLeeLee2018} and the references therein.

We now start with the proof of Lemma \ref{lem:NjU}.
To prove the Lemma, we compare the dynamics of $\cU$ and $\cU^{(M)}$ in Lemma \ref{lem:UNUUM}.
To this aim, we recall weak bounds on the $\cU$ dynamics.

\begin{lemma}[{\cite[Lemma 3.6]{RodnianskiSchlein2009}}]\label{lem:3.6 i Rodnianski}
	For arbitrary $t,s\in\mathbb{R}$ and $\psi\in\mathcal{F}$, we have
	\[
		\left\langle \psi,\cU(t;s) \mathcal{N}\cU^{*}(t;s)\psi\right\rangle \leq6\left\langle \psi,(\mathcal{N}+N+1)\psi\right\rangle .
	\]
	Moreover, for every $\ell\in\mathbb{N}$, there exists a constant $C(\ell)$ such that
	\[
		\left\langle \psi,\cU(t;s)\mathcal{N}^{2\ell}\cU^{*}(t;s)\psi\right\rangle \leq C(\ell)\left\langle \psi,(\mathcal{N}+N)^{2\ell}\psi\right\rangle
	\]
	and
	\[
		\left\langle \psi,\cU(t;s)\mathcal{N}^{2\ell+1}\cU^{*}(t;s)\psi\right\rangle \leq C(\ell)\left\langle \psi,(\mathcal{N}+N)^{2\ell+1}(\mathcal{N}+1)\psi\right\rangle
	\]
	for all $t,s\in\mathbb{R}$ and $\psi\in\mathcal{F}$.
\end{lemma}

\begin{proof}
	The proof can be found in \cite{RodnianskiSchlein2009}.
\end{proof}

Now we are ready to compare the dynamics of $\cU$ and $\cU^{(M)}_{N}$.

\begin{lemma}\label{lem:UNUUM}
	Suppose that the assumptions in Theorem \ref{thm:main} hold. Then, for every $j\in\mathbb{N}$ and $\psi\in\cF$, there exists a constant $C \equiv C(j, \|w\|_1 )>0$ such that for all $t,s\in\mathbb R$
	\begin{align}
		\begin{split}
			&\Big|\langle\cU(t;0)\psi,\cN^{j}\big(\cU(t;0)-\cU_{N}^{(M)}(t;0)\big)\psi\rangle\Big|
			\leq C \frac{N^j}{ \varepsilon^{1/2} M^{j+1}} \|(\cN+1)^{j+1}\psi\|^{2}\exp\Bigg(K\bigg(1+\sqrt{\frac{M}{N  \varepsilon^{3} }}\bigg)t\Bigg)
		\end{split}
		\intertext{and}
		\begin{split}
			&\Big| \langle \cU^{(M)}_N (t;0) \psi, \cN^j \left( \cU (t;0) - \cU_N^{(M)} (t;0)\right) \psi \rangle \Big|
			\leq \; C \frac{N^j}{M^{j+1} {\varepsilon^{1/2}}}\, \|  (\cN+1)^{j+1} \psi \|^2\exp\Bigg(K\bigg(1+\sqrt{\frac{M}{N  \varepsilon^{3} }}\bigg)t\Bigg).
		\end{split}\label{eq:UNUUM1overM}
	\end{align}
\end{lemma}

\begin{proof}
	To simplify the notation, we consider the case $s=0$ and $t>0$ only; other cases can be treated in a similar manner. To prove the first inequality of the lemma, we expand the difference of the two evolutions as follows:
	\begin{align*}
		& \langle \cU(t;0)\psi,\cN^{j}\big(\cU(t;0)-\cU^{(M)}_{N}(t;0)\big)\psi\rangle=\langle \cU(t;0)\psi,\cN^{j}\cU(t;0)\big(1-\cU^{*}(t;0)\,\cU^{(M)}_{N}(t;0)\big)\psi\rangle\\
		& =-\ii \int_{0}^{t}\ds\,\langle \cU(t;0)\psi,\cN^{j}\cU(t;0)\big(\cL_{N}(s)-\cL_{N}^{(M)}(s)\big)\,\cU^{(M)}_{N}(s;0)\psi\rangle\\
		& =-\frac{\ii \mu }{\sqrt{N}}\int_{0}^{t}\ds\,\int_{\bR\times\bR}\dx\dy\,w_ \varepsilon (x)w_ \varepsilon (y)\\
		& \qquad\times\langle\cU(t;0)\psi,\cN^{j}\cU(t;s)a_{x}^{*}\big(\overline{u_{\varepsilon,s}(y)}a_{y}\chi(\cN >M)+u_{\varepsilon,s}(y)\chi(\cN >M)a_{y}^{*}\big)a_{x}\cU^{(M)}_{N}(s;0)\psi\rangle \\
		& =-\frac{\ii \mu }{\sqrt{N}}\int_{0}^{t}\ds\,\int_{\bR}\dx\,w_ \varepsilon (x)\langle a_{x}\cU^{*}(t;s)\cN^{j}\cU(t;0)\psi,a( w_\varepsilon u_{\varepsilon,s})\chi(\cN >M)a_{x}\cU^{(M)}_{N}(s;0)\psi\rangle \\
		& \quad\quad-\frac{\ii \mu }{\sqrt{N}}\int_{0}^{t}\ds\,\int_{\bR}\dx\,w_ \varepsilon (x)\langle a_{x}\cU^{*}(t;s)\cN^{j}\cU(t;0)\psi,\chi(\cN >M)a^{*}( w_\varepsilon u_{\varepsilon,s})a_{x}\cU^{(M)}_{N}(s;0)\psi\rangle .
	\end{align*}
	Hence,
	\begin{align}
		& \Big|\langle \cU(t;0)\psi,\cN^{j}\big(\cU(t;0)-\cU^{(M)}_{N}(t;0)\big)\psi\rangle\Big|\nonumber \\
		& \leq |\mu| \frac{1}{ \varepsilon \sqrt{N}}\int_{0}^{t}\ds\,\int_{\bR}\dx\,\|a_{x}\cU^{*}(t;s)\cN^{j}\cU(t;0)\psi\|
		\|a( w_\varepsilon u_{\varepsilon,s})a_{x}\chi(\cN >M+1)\,\cU^{(M)}_{N}(s;0)\psi\|\nonumber \\
		& \quad+ |\mu| \frac{1}{ \varepsilon \sqrt{N}}\int_{0}^{t}\ds\,\int_{\bR}\dx\,\|a_{x}\cU^{*}(t;s)\cN^{j}\cU(t;0)\psi\|
		\|a^{*}( w_\varepsilon u_{\varepsilon,s})a_{x}\chi(\cN >M)\,\cU^{(M)}_{N}(s;0)\psi\|\nonumber \\
		& \leq |\mu| \frac{2}{ \varepsilon \sqrt{N}}\int_{0}^{t}\ds\,\| w_\varepsilon u_{\varepsilon,s}\|_{2}
		\int_{\bR}\dx\,\|a_{x}\cU^{*}(t;s)\cN^{j}\cU(t;0)\psi\|\,\|a_{x}(\cN+1)^{1/2}\chi(\cN >M)\,\cU^{(M)}_{N}(s;0)\psi\|\nonumber \\
		& \leq |\mu| \frac{2}{ \varepsilon \sqrt{N}}\int_{0}^{t}\ds\,\| w_\varepsilon u_{\varepsilon,s}\|_{2}
		\|\cN^{1/2}\cU^{*}(t;s)\cN^{j}\cU(t;0)\psi\|\,\|(\cN+1)\chi(\cN >M)\,\cU^{(M)}_{N}(s;0)\psi\|.
		\label{eq:U-UM}
	\end{align}
	Since $\chi(\cN >M)\leq(\cN /M)^{L}$ for any $L>1$, we find that, from \cite[(3.30)]{RodnianskiSchlein2009},
	\begin{align*}
		& \Big|\langle \cU(t;0)\psi,\cN^{j}\big(\cU(t;0)-\cU^{(M)}_{N}(t;0)\big)\psi\rangle\Big|\\
		& \quad
		\leq C N^{j}\|(\cN+1)^{j+1}\psi\|
		\int_{0}^{t}\ds\,\| w_\varepsilon u_{\varepsilon,s}\|_{2}\langle \cU^{(M)}_{N}(s;0)\psi,(\cN+1)^{2}\chi(\cN >M)\,\cU^{(M)}_{N}(s;0)\psi\rangle ^{1/2}\\
		& \quad\leq C N^{j}\|(\cN+1)^{j+1}\psi\|
		\int_{0}^{t}\ds\,{\varepsilon^{-1/2}}\|w\|_{1}^{1/2}\|u_{\varepsilon,s}\|_{\infty}\langle \cU^{(M)}_{N}(s;0)\psi,(\cN+1)^{2}\chi(\cN >M)\,\cU^{(M)}_{N}(s;0)\psi\rangle ^{1/2}\\
		& \quad\leq C N^{j}{\varepsilon^{-1/2}}\|(\cN+1)^{j+1}\psi\|\langle \cU^{(M)}_{N}(s;0)\psi,\frac{(\cN+1)^{2+{2j+2}}}{M^{2j+2}}\cU^{(M)}_{N}(s;0)\psi\rangle ^{1/2}{\int_{0}^{t}\ds\,\|u_{\varepsilon,s}\|_{\infty}}
	\end{align*}
	where $C = C(j, \|w\|_1)$.
	By Lemma \ref{lem:UMNUM}, we conclude that
	\begin{align*}
		&\Big|\langle\cU(t;0)\psi,\cN^{j}\big(\cU(t;0)-\cU_{N}^{(M)}(t;0)\big)\psi\rangle\Big|
		\leq C \frac{N^j}{ \varepsilon^{1/2} M^{j+1}} \|(\cN+1)^{j+1}\psi\|^{2}\exp\Bigg(K\bigg(1+\sqrt{\frac{M}{N  \varepsilon^{3} }}\bigg)t\Bigg).
	\end{align*}
	To prove \eqref{eq:UNUUM1overM}, we proceed similarly; analogously to \eqref{eq:U-UM} we find
	\begin{align*}
		&\langle \cU^{(M)}_N (t;0) \psi, \cN^j \left( \cU (t;0) -\cU_N^{(M)}(t;0)\right) \psi \rangle \\= \;
		&-\frac{\ii\mu }{\sqrt{N}}\int_0^t \ds \int_{\bR} \dr x  \langle a_x
		\cU (t;s)^* \cN^j \cU^{(M)}_N (t;0) \psi, a (w_ \varepsilon (x-.)\varphi_t) \chi(\cN > M) a_x \cU_N^{(M)} (s;0) \psi \rangle \\
		&-\frac{\ii\mu }{\sqrt{N}}\int_0^t \ds \int_{\bR} \dr x  \langle a_x \cU(t;s)^* \cN^j \cU^{(M)}_N (t;0) \psi, \chi (\cN >M) a^*	(w_ \varepsilon (x-.)\varphi_t) a_x \cU_N^{(M)} (s;0) \psi \rangle
	\end{align*}
	and thus, by $|\mu|\le1$, $\chi (\cN > M)\le\left(\frac{\cN}{M}\right)^{L}$ for any $L>1$, $\|w_ \varepsilon \|_{2}=N^{3\beta / 2}\|w\|_{2}$,
	\begin{equation}\label{eq:0}\begin{split}
			&\Big| \langle \cU^{(M)}_N (t;0) \psi, \cN^j \left( \cU (t;0) -
			\cU_N^{(M)} (t;0)\right) \psi \rangle \Big| \\
			&\leq \; \frac{C}{\sqrt{N}} \int_0^t \ds\,\| w_\varepsilon u_{\varepsilon,s}\|_{2}\, \|\cN^{1/2} \cU
			(t;s)^* \cN^j \cU_N^{(M)} (t;0) \psi \| \, \|  \cN \chi (\cN > M)
			\cU_N^{(M)} (s;0) \psi \|\,	\\
			&\leq \; \frac{C}{ \varepsilon^{1/2} \sqrt{N}} \int_0^t \ds\,(\|w_ \varepsilon \|_{1})^{1/2}\, \|u_{\varepsilon,s}\|_{\infty}\, \|\cN^{1/2} \cU (t;s)^* \cN^j \cU_N^{(M)} (t;0) \psi \| \, \|  \cN \left(\frac{\cN}{M}\right)^{j+1} \cU_N^{(M)} (s;0) \psi \|\,\\
			&\leq \; \frac{C}{M^{j+1}{\varepsilon^{1/2}} \sqrt{N}}\, \int_0^t \ds\, \| u_{\varepsilon,s}\|_{\infty}\, \|\cN^{1/2} \cU(t;s)^* \cN^j \cU_N^{(M)} (t;0) \psi \| \, \|  \cN^{j+1} \cU_N^{(M)} (s;0) \psi \|\,.
	\end{split}\end{equation}
	By Lemma \ref{lem:3.6 i Rodnianski}, Lemma \ref{lem:UMNUM} and $j+1/2\le {j+1}$ , we have
	\begin{equation}\label{eq:1}\begin{split}
			& \|\cN^{1/2} \cU
			(t;s)^* \cN^j \cU_N^{(M)} (t;0) \psi \| \le 6 \|(\cN+N+1)^{1/2}  \cN^j \cU_N^{(M)} (t;0) \psi \| \\
			& \le 12 N^{1/2} \|(\cN+1)^{j+1/2}  \cU_N^{(M)} (t;0) \psi \| \\
			& \le C N^{1/2} \|(\cN+1)^{j+1/2}   \psi \| \exp\Bigg(K\bigg(1+\sqrt{\frac{M}{N  \varepsilon^{3} }}\bigg)t\Bigg) \\
			& \le C N^{1/2} \|(\cN+1)^{j+1}   \psi \| \exp\Bigg(K\bigg(1+\sqrt{\frac{M}{N  \varepsilon^{3} }}\bigg)t\Bigg).
	\end{split}\end{equation}
	By Lemma \ref{lem:UMNUM}, we have
	\begin{equation}\label{eq:2}
		\|  \cN^{j+1}
		\cU_N^{(M)} (s;0) \psi \|\le C\|  (\cN+1)^{j+1}
		\psi \|\exp\Bigg(K\bigg(1+\sqrt{\frac{M}{N  \varepsilon^{3} }}\bigg)t\Bigg).
	\end{equation}
	Combining \eqref{eq:0}, \eqref{eq:1} and \eqref{eq:2}, we obtain
	\begin{equation}\begin{split}
			&\Big| \langle \cU^{(M)}_N (t;0) \psi, \cN^j \left( \cU (t;0) -\cU_N^{(M)} (t;0)\right) \psi \rangle \Big|\\
			&\leq \; \frac{C}{M^{j+1} {\varepsilon^{1/2}} \sqrt{N}}\, N^{j+1/2} \|  (\cN+1)^{j+1} \psi \|^2\exp\Bigg(K\bigg(1+\sqrt{\frac{M}{N  \varepsilon^{3} }}\bigg)t\Bigg)\\
			&\leq \; \frac{C N^j}{M^{j+1} {\varepsilon^{1/2}}}\, \|  (\cN+1)^{j+1} \psi \|^2\exp\Bigg(K\bigg(1+\sqrt{\frac{M}{N  \varepsilon^{3} }}\bigg)t\Bigg).
	\end{split}\end{equation}
	This gives the desired Lemma.
\end{proof}

Let us now prove Lemma \ref{lem:NjU}.
\begin{proof}[Proof of Lemma \ref{lem:NjU}]
	Let $M=N {\varepsilon^3}$. Then by Lemmata \ref{lem:UMNUM} and \ref{lem:UNUUM}, we get
	\begin{align*}
		& \Big\langle \cU\left(t;s\right)\psi,\cN^{j}\cU\left(t;s\right)\psi\Big\rangle\\
		& =\Big\langle \cU\left(t;s\right)\psi,\cN^{j}(\cU-\cU^{(M)}_{N})\left(t;s\right)\psi\Big\rangle +\Big\langle (\cU-\cU^{(M)}_{N})\left(t;s\right)\psi,\cN^{j}\cU^{(M)}_{N}\left(t;s\right)\psi\Big\rangle\\
		& \qquad+\Big\langle \cU^{(M)}_{N}\left(t;s\right)\psi,\cN^{j}\cU^{(M)}_{N}\left(t;s\right)\psi\Big\rangle\\
		& \leq C\Big(1+\frac{1}{\varepsilon^{\frac{1}{2}+3j}N}\Big)\,\Big\langle \psi,\left(\cN+1\right)^{ 2j+4 }\psi\Big\rangle
		\exp\Big(K t\Big)
		+ C\langle \psi,\left(\cN+1\right)^{j}\psi\rangle \exp\left(K(1+\varepsilon^{-1}) t  \right)\\
		&\leq C\Big(1+\frac{1}{\varepsilon^{\frac{1}{2}+3j}N}\Big)\,\langle \psi,\left(\cN+1\right)^{2j+2}\psi\rangle \exp\left(K (1+\varepsilon^{-1}) t \right).
	\end{align*}
	This yields the desired result.
\end{proof}

Recall the definition of $\widehat{\cU}\left(t;s\right)$ in \eqref{eq:def_mathcalwidehatU}. In the next lemma, we prove an estimate for the evolution with respect to $\widehat{\cU}$.

\begin{lemma}\label{lem:tildeNj}
	Suppose that the assumptions in Theorem \ref{thm:main} hold. Then, for any $\psi\in\cF$ and $j\in \bN$, there exists a constant $C = C(\|w\|_1 )>0$ such that
	\[
		\langle \widehat{\cU}\left(t;s\right)\psi,\cN^{j}\widehat{\cU}\left(t;s\right)\psi\rangle\leq C e^{Kt}\langle \psi,\left(\cN+1\right)^{j}\psi\rangle .
	\]
\end{lemma}

\begin{proof}
	Let $\widehat{\psi}=\widehat{\cU}(t;s)\psi$ and assume without loss of generality that $t\ge s$. We have
	\begin{align*}
		& \frac{\mathrm{d}}{\dt}\langle \widehat{\psi},(\cN+1)^{j}\widehat{\psi}\rangle
		=\langle \widehat{\psi},[\ii (\widehat{\cL}_{2}+\cL_{4}),(\cN+1)^{j}]\widehat{\psi}\rangle\\
		& =-{ \varepsilon }\ImPart \int_{\bR\times\bR}\dx\dy\,w_ \varepsilon (x)w_ \varepsilon (y)u_{\varepsilon,t}(x)u_{\varepsilon,t}(y)\langle \widehat{\psi},[a_{x}^{*}a_{y}^{*},(\cN+1)^{j}]\widehat{\psi}\rangle\\
		& ={ \varepsilon }\ImPart \int_{\bR\times\bR}\dx\dy\,w_ \varepsilon (x)w_ \varepsilon (y)u_{\varepsilon,t}(x)u_{\varepsilon,t}(y)\langle \widehat{\psi},a_{x}^{*}a_{y}^{*}((\cN+3)^{j}-(\cN+1)^{j})\widehat{\psi}\rangle\\
		& ={ \varepsilon }\ImPart \int_{\bR\times\bR}\dx\dy\,w_ \varepsilon (x)w_ \varepsilon (y)u_{\varepsilon,t}(x)u_{\varepsilon,t}(y)\\
		& \qquad\times\langle (\cN+3)^{(j+1)/2}a_{x}\widehat{\psi},a_{y}(\cN+3)^{(1-j)/2}((\cN+3)^{j}-(\cN+1)^{j})\widehat{\psi}\rangle.
	\end{align*}
	Then, one gets
	\begin{align*}
		\frac{\mathrm{d}}{\dt}\langle \widehat{\psi},(\cN+1)^{j}\widehat{\psi}\rangle
		&\leq { \varepsilon }\int_{\bR\times\bR}\dx\dy\,|w_ \varepsilon (x)w_ \varepsilon (y)| |u_{\varepsilon,t}(x)| |u_{\varepsilon,t}(y)|\|(\cN+3)^{(j-1)/2}a_{x}\widehat{\psi}\|\\
		& \qquad\times\|a_{y}(\cN+3)^{(1-j)/2}((\cN+3)^{j}-(\cN+1)^{j})\widehat{\psi}\|\\
		& \leq \|u_{\varepsilon,t}\|_{\infty}^{2} \Big(\int_{\bR}\dx\,\|(\cN+3)^{(j-1)/2}a_{x}\widehat{\psi}\|^{2}\Big)^{1/2}\\
		& \qquad\times\Big(\int_{\bR}\dy\,\|a_{y}(\cN+3)^{(1-j)/2}((\cN+3)^{j}-(\cN+1)^{j})\widehat{\psi}\|^2\Big)^{1/2}.
	\end{align*}
	Since, for all $j\in\bN$, one can see that $|(\cN+3)^{j}-(\cN+1)^{j})|\leq C_j (\cN+1)^{j-1}$ for some $C_j>0$, one gets
	\begin{align*}
		\frac{\mathrm{d}}{\dt}\langle \widehat{\cU}(t;s)\psi,(\cN+1)^{j}\widehat{\cU}(t;s)\psi\rangle
		&\leq C_j \|u_{\varepsilon,t}\|_{\infty}^{2}\|(\cN+1)^{j/2}\widehat{\cU}(t;s)\psi\|^{2}\\
		&=C_j \|u_{\varepsilon,t}\|_{\infty}^{2}\langle \widehat{\cU}(t;s)\psi,(\cN+1)^{j}\widehat{\cU}(t;s)\psi\rangle.
	\end{align*}
	Here, note that $C_j$ can change from line to line. Applying Gr\"onwall's lemma, we conclude that
	\begin{align*}
		\langle \widehat{\cU}(t;s)\psi,(\cN+1)^{j}\widehat{\cU}(t;s)\psi\rangle
		&\leq \langle \psi,(\cN+1)^{j}\psi\rangle \,
		e^{K t }.
	\end{align*}
	Hence, we get the result.
\end{proof}

\begin{lemma}
	For all $\psi\in\cF$ and $f\in L^2(\bR)$, we have the following inequalities with a constant $C = C(    \|w\|_1, \|w\|_2 ) > 0$. Then
	\begin{align}
		&\Vert (\cN+1)^{j/2}\cL_{3}(t)\psi\Vert \leq
		\frac{C}{ \varepsilon^{3/2} \sqrt{N}} \|u_{\varepsilon,t}\|_{L^{\infty}}\Vert \left(\cN+1\right)^{(j+3)/2}\psi\Vert,\label{eq:NL3-p}\\
		&\Vert (\cN+1)^{j/2}\big(\cL_2(t) - \widehat{\cL}_2(t)\big)\psi\Vert \leq
		C \|u_{\varepsilon,t}\|_{L^{\infty}}\Vert \left(\cN+1\right)^{(j+3)/2}\psi\Vert,\label{eq:NL2hatL2-p}\\
		\intertext{and}
		&\Vert (\cN+1)^{j/2}\big(\cU^{*}(t;0)\phi(f)\,\cU(t;0)-\widehat{\cU}^{*}(t;0)\phi(f)\widehat{\cU}(t;0)\big)\Omega\Vert \leq\frac{C e^{Kt}\|f\|}{ \varepsilon^{3/2} \sqrt{N}}.\label{eq:NjUphiUtildeUphitildeU-p}
	\end{align}
\end{lemma}

\begin{proof}
	For \eqref{eq:NL3-p}, the proof is the same as in \cite[Lemma 5.3]{Lee2013} with $\|w_ \varepsilon \|_2 = {\varepsilon^{-1/2}} \|w\|_2$.
	Furthermore, for \eqref{eq:NL2hatL2-p}, we follow the proof from \cite[Lemma 5.3]{Lee2013} with $\|W_ \varepsilon \|_1 = \|w\|_1$ , and we replace $\cL_3$ by
	\[
		\cL_2 - \widehat{\cL}_2 = \mu {  \frac{1-\varepsilon}{2} }\int_{\bR \times \bR} \dx\dy\,  w_ \varepsilon (x)w_ \varepsilon (y)\big(u_{\varepsilon,t}(x)u_{\varepsilon,t}(y)a_{x}^{*}a_{y}^{*}+\overline{u_{\varepsilon,t}(x)}\,\overline{u_{\varepsilon,t}(y)}a_{x}a_{y}\big).
	\]	
	From \eqref{eq:NjUphiUtildeUphitildeU-p}, we follow the proof of \cite[Lemma 5.4]{Lee2013}
	with Lemma \ref{lem:tildeNj}, \eqref{eq:NL3-p}, and \eqref{eq:NL2hatL2-p}.
\end{proof}

\subsection{Proof of Propositions \ref{prop:Et1} and \ref{prop:Et2}}\label{sec:proof-of-Et1-Et2}

\begin{proof}[Proof of Proposition \ref{prop:Et1}]
	Preliminarily, et us recall that
	\[
		E_{t}^{(1)}(J)=\frac{d_{N}}{N}\Big\langle \cW^{*}(\sqrt{N}\varphi)\frac{(a^{*}(\varphi))^{N}}{\sqrt{N!}}\Omega,\cU^{*}(t;0)d\Gamma(J)\,\cU(t;0)\Omega\Big\rangle.
	\]
	We start by noting that
	\begin{align*}
		|E_{t}^{(1)}(J)| & =\Big|\frac{d_{N}}{N}\langle \cW^{*}(\sqrt{N}\varphi)\frac{(a^{*}(\varphi))^{N}}{\sqrt{N!}}\Omega,\cU^{*}(t;0)d\Gamma(J)\,\cU(t;0)\Omega\rangle\Big|\\
		& \leq\frac{d_{N}}{N}\Big\Vert (\cN+1)^{-\frac{1}{2}}\cW^{*}(\sqrt{N}\varphi)\frac{(a^{*}(\varphi))^{N}}{\sqrt{N!}}\Omega\Big\Vert\,\Big\Vert (\cN+1)^{\frac{1}{2}}\cU^{*}(t;0)d\Gamma(J)\,\cU(t;0)\Omega\Big\Vert .
	\end{align*}
	Furthermore, by Lemma \ref{lem:coherent_all},
	\[
		\left\Vert (\cN+1)^{-\frac{1}{2}}\cW^{*}(\sqrt{N}\varphi)\frac{(a^{*}(\varphi))^{N}}{\sqrt{N!}}\Omega\right\Vert \leq\frac{C}{d_{N}}
	\]
	and, applying Lemma \ref{lem:NjU} twice and Lemma \ref{lem:basifock-bounds},
	\begin{align*}
		& \Vert (\cN+1)^{\frac{1}{2}}\cU^{*}(t;0)d\Gamma(J)\,\cU(t;0)\Omega\Vert \\
		&\leq C\Big(1+\frac{1}{\varepsilon^{\frac{7}{2}}N}\Big) \exp\left(K(1+\varepsilon^{-1}) t\right)
		\,\Vert (\cN+1)^{2} d\Gamma(J)\,\cU(t;0)\Omega\Vert \\
		&\leq C\Big(1+\frac{1}{\varepsilon^{\frac{7}{2}}N}\Big) \exp\left(K(1+\varepsilon^{-1}) t\right)
		\,\Vert J\Vert \, \Vert (\cN+1)^{3}\cU(t;0)\Omega\Vert \\
		&\leq C\Big(1+\frac{1}{\varepsilon^{\frac{7}{2}}N}\Big) \Big(1+\frac{1}{\varepsilon^{\frac{37}{2}}N}\Big) \exp\left(K(1+\varepsilon^{-1}) t\right)
		\,\Vert J\Vert \, \Vert (\cN+1)^{13}\cU(t;0)\Omega\Vert \\
		&=C\Big(1+\frac{1}{\varepsilon^{\frac{7}{2}}N}\Big) \Big(1+\frac{1}{\varepsilon^{\frac{37}{2}}N}\Big) \exp\left(K(1+\varepsilon^{-1}) t\right) \|J\|,
	\end{align*}
	we obtain
	\[
		|E_{t}^{(1)}(J)|\leq \frac{C}{N}\Big(1+\frac{1}{\varepsilon^{\frac{7}{2}}N}\Big) \Big(1+\frac{1}{\varepsilon^{\frac{37}{2}}N}\Big) \exp\left(K (1+\varepsilon^{-1})t \right) \|J\|,
	\]
	which is the desired result.
\end{proof}
\begin{proof}[Proof of Proposition \ref{prop:Et2}]
	Let
	\[
		\cR (f)=\cU^{*}(t;0)\phi(f)\,\cU(t;0)-\widehat{\cU}^{*}(t;0)\phi(f)\widehat{\cU}(t;0).
	\]
	Then
	\begin{align}
		&|E_{t}^{(2)}(J)|
		=\frac{d_{N}}{\sqrt{N}} \Big\langle \cW^{*}(\sqrt{N}\varphi)\frac{(a^{*}(\varphi))^{N}}{\sqrt{N!}}\Omega,\mathcal{\widehat{U}}^{*}(t;0)\phi(Ju_{\varepsilon,t})\mathcal{\widehat{U}}(t;0)\Omega\Big\rangle\nonumber \\
		& \qquad+\frac{d_{N}}{\sqrt{N}}\Big\langle \cW^{*}(\sqrt{N}\varphi) \frac{(a^{*}(\varphi))^{N}}{\sqrt{N!}}\Omega,\cR (Ju_{\varepsilon,t})\Omega\Big\rangle\nonumber \\
		& \leq \frac{d_{N}}{\sqrt{N}} \Big\Vert \sum_{k=0}^{\infty}(\cN+1)^{-\frac{5}{2}}P_{2k+1}\cW^{*}(\sqrt{N}\varphi)\frac{(a^{*}(\varphi))^{N}}{\sqrt{N!}}\Omega\Big\Vert \, \Big\Vert (\cN+1)^{\frac{5}{2}}\mathcal{\widehat{U}}^{*}(t;0)\phi(Ju_{\varepsilon,t})\widehat{\cU}(t;0)\Omega\Big\Vert \nonumber \\
		& \qquad+\frac{d_{N}}{\sqrt{N}}\Big\Vert (\cN+1)^{-\frac{1}{2}}\cW^{*}(\sqrt{N}\varphi)\frac{(a^{*}(\varphi))^{N}}{\sqrt{N!}}\Omega\Big\Vert \, \Big\Vert (\cN+1)^{\frac{1}{2}}\cR (Ju_{\varepsilon,t})\Omega \Big\Vert \label{eq:e_t^21}.
	\end{align}
	Let $\mathsf{K}=\frac{1}{2}N^{1/3}$. By Lemmata \ref{lem:coherent_all} and \ref{lem:coherent_even_odd}, one gets
	\begin{align*}
		& \Big\Vert \sum_{k=0}^{\infty}(\cN+1)^{-\frac{5}{2}}P_{2k+1}\cW^{*}(\sqrt{N}\varphi)\frac{(a^{*}(\varphi))^{N}}{\sqrt{N!}}\Omega\Big\Vert^{2}\\
		& \qquad\leq\sum_{k=0}^{\mathsf{K}}\Big\Vert (\cN+1)^{-\frac{5}{2}}P_{2k+1}\cW^{*}(\sqrt{N}\varphi)\frac{(a^{*}(\varphi))^{N}}{\sqrt{N!}}\Omega\Big\Vert^{2}\\
		& \qquad\qquad+\frac{1}{\mathsf{K}^{4}}\sum_{k=\mathsf{K}}^{\infty}\Big\Vert(\cN+1)^{-1/2} P_{2k+1}\cW^{*}(\sqrt{N}\varphi)\frac{(a^{*}(\varphi))^{N}}{\sqrt{N!}}\Omega\Big\Vert^{2}\\
		& \qquad\leq \Big(\sum_{k=0}^{\mathsf{K}}\frac{C}{(k+1)^{2}d_{N}^{2}N}\Big)+\frac{C}{N^{4/3}}\Big\Vert(\cN+1)^{-1/2} \cW^{*}(\sqrt{N}\varphi)\frac{(a^{*}(\varphi))^{N}}{\sqrt{N!}}\Omega\Big\Vert^2 \leq\frac{C}{d_{N}^{2}N^{4/3}}.
	\end{align*}
	Furthemore, Lemmata \ref{lem:tildeNj} and \ref{lem:basifock-bounds} yield,
	\begin{alignat*}{1}
		& \Vert (\cN+1)^{\frac{5}{2}}\mathcal{\widehat{U}}^{*}(t;0)\phi(Ju_{\varepsilon,t})\widehat{\cU}(t;0)\Omega\Vert
		\leq C t^{K} \Vert (\cN+1)^{\frac{5}{2}}\phi(Ju_{\varepsilon,t})\widehat{\cU}(t;0)\Omega\Vert \\
		&\leq C t^{K}  \|Ju_{\varepsilon,t}\|\Vert (\cN+2)^{3}\mathcal{\widehat{U}}(t;0)\Omega\Vert
		\leq C t^{K} \|J\|\Vert (\cN+2)^{3}\Omega\Vert
		= C t^{K} \|J\|.
	\end{alignat*}
	For the second term on the right-hand side of \eqref{eq:e_t^21}, we use Lemma \ref{lem:coherent_all} and \eqref{eq:NjUphiUtildeUphitildeU-p}, for $f=J\varphi_t$.
	Altogether, we have
	\begin{equation*}
		\Vert (\cN+1)^{j/2}\cR (f)\Omega\Vert \leq
		C\|J\|\varepsilon^{-3/2}N^{-\frac{1}{2}}
	\end{equation*}
	which is the desired conclusion.
\end{proof}

\section{Proof of Theorem \ref{thm:main}}\label{sec:pfmainthm}

\begin{proof}[Proof of Theorem \ref{thm:main}]
	First, recall that
	\begin{equation}
		\gamma_{N,t}^{(1)}(x;y) =\frac{1}{N}\left\langle \frac{\left(a^{*}(\varphi)\right)^{N}}{\sqrt{N!}}\Omega,e^{i\cH_{N}t}a_{y}^{*}a_{x}e^{-i\cH_{N}t}\frac{\left(a^{*}(\varphi)\right)^{N}}{\sqrt{N!}}\Omega\right\rangle.
	\end{equation}
	From the definition of the creation operator in \eqref{eq:creation} and of $ d_N$ in \eqref{eq:d_N}, one easily finds
	\begin{equation} \label{eq:coherent_vec}
		\{0,0,{{\dots}},0,\varphi^{\otimes N},0,{{\dots}}\}=\frac{\left(a^{*}(\varphi)\right)^{N}}{\sqrt{N!}}\Omega,
	\end{equation}
	where the function $\varphi^{\otimes N}$ on the left-hand side is in the $N$-th sector of the Fock space. Recall that $P_N$ is the projection onto the $N$-particle sector of the Fock space. From \eqref{Weyl_f}, one finds
	\[
		\frac{\left(a^{*}(\varphi)\right)^{N}}{\sqrt{N!}}\Omega=\frac{\sqrt{N!}}{N^{N/2}e^{-N/2}}P_{N}\cW(\sqrt{N}\varphi)\Omega=d_{N}P_{N}\cW(\sqrt{N}\varphi)\Omega.
	\]
	Since $\cH_N$ does not change the number of particles, we also have 
	\begin{align*}
		\gamma_{N,t}^{(1)}(x;y) & =\frac{1}{N}\left\langle \frac{\left(a^{*}(\varphi)\right)^{N}}{\sqrt{N!}}\Omega,e^{\ii \cH_{N}t}a_{y}^{*}a_{x}e^{-\ii \cH_{N}t}\frac{\left(a^{*}(\varphi)\right)^{N}}{\sqrt{N!}}\Omega\right\rangle \\
		& =\frac{d_{N}}{N}\left\langle \frac{\left(a^{*}(\varphi)\right)^{N}}{\sqrt{N!}}\Omega,e^{\ii \cH_{N}t}a_{y}^{*}a_{x}e^{-\ii \cH_{N}t}P_{N}\cW(\sqrt{N}\varphi)\Omega\right\rangle \\
		& =\frac{d_{N}}{N}\left\langle \frac{\left(a^{*}(\varphi)\right)^{N}}{\sqrt{N!}}\Omega,P_{N}e^{\ii \cH_{N}t}a_{y}^{*}a_{x}e^{-\ii \cH_{N}t}\cW(\sqrt{N}\varphi)\Omega\right\rangle \\
		& =\frac{d_{N}}{N}\left\langle \frac{\left(a^{*}(\varphi)\right)^{N}}{\sqrt{N!}}\Omega,e^{\ii \cH_{N}t}a_{y}^{*}a_{x}e^{-\ii \cH_{N}t}\cW(\sqrt{N}\varphi)\Omega\right\rangle, \\
	\end{align*}
	where we used that $P_{N}\frac{\left(a^{*}(\varphi)\right)^{N}}{\sqrt{N!}}\Omega =\frac{\left(a^{*}(\varphi)\right)^{N}}{\sqrt{N!}}\Omega$ in the last step.
	To simplify it further, we use the relation
	\begin{equation}\label{eq:eae}
		e^{\ii \cH_{N}t}a_{x}e^{-\ii \cH_{N}t}=\cW(\sqrt{N}\varphi)\,\cU^{*}(t;0)(a_{x}+\sqrt{N}u_{\varepsilon,t}(x))\,\cU(t;0)\cW^{*}(\sqrt{N}\varphi),
	\end{equation}
	which follows from the first equality in Lemma \ref{lem:Basic_Weyl}(4), the definition of $\cU$ in \eqref{eq:def_mathcalU} and the unitarity of the Weyl operator (see Lemma \ref{lem:Basic_Weyl}(2)). By \eqref{eq:eae} and the analogous result for the creation operator, we obtain
	\begin{align*}
		\gamma_{N,t}^{(1)}(x;y) & =\frac{d_{N}}{N}\langle \frac{\left(a^{*}(\varphi)\right)^{N}}{\sqrt{N!}}\Omega,e^{\ii \cH_{N}t}a_{y}^{*}a_{x}e^{-\ii \cH_{N}t}\cW(\sqrt{N}\varphi)\Omega\rangle \\
		& =\frac{d_{N}}{N}\langle \frac{\left(a^{*}(\varphi)\right)^{N}}{\sqrt{N!}}\Omega,\cW(\sqrt{N}\varphi)\,\cU^{*}(t;0)\left(a_{y}^{*}+\sqrt{N}\,\overline{u_{\varepsilon,t}\left(y\right)}\right)\left(a_{x}+\sqrt{N}u_{\varepsilon,t}(x)\right)\,\cU(t;0)\Omega\rangle.
	\end{align*}
	Thus,
	\begin{equation}\label{eq:gaphyphx}\begin{split}
		\gamma_{N,t}^{(1)}(x;y)-\overline{u_{\varepsilon,t}\left(y\right)}u_{\varepsilon,t}(x) & =\frac{d_{N}}{N}\langle \frac{\left(a^{*}(\varphi)\right)^{N}}{\sqrt{N!}}\Omega,\cW(\sqrt{N}\varphi)\,\cU^{*}(t;0)a_{y}^{*}a_{x}\cU(t;0)\Omega\rangle \\
		& \quad+\overline{u_{\varepsilon,t}(y)}\frac{d_{N}}{\sqrt{N}}\langle \frac{\left(a^{*}(\varphi)\right)^{N}}{\sqrt{N!}}\Omega,\cW(\sqrt{N}\varphi)\,\cU^{*}(t;0)a_{x}\cU(t;0)\Omega\rangle \\
		& \quad+u_{\varepsilon,t}(x)\frac{d_{N}}{\sqrt{N}}\langle \frac{\left(a^{*}(\varphi)\right)^{N}}{\sqrt{N!}}\Omega,\cW(\sqrt{N}\varphi)\,\cU^{*}(t;0)a_{y}^{*}\cU(t;0)\Omega\rangle.
	\end{split}\end{equation}
	Recalling the definition of $E_{t}^{(1)}(J)$ and $E_{t}^{(2)}(J)$ in Propositions \ref{prop:Et1} and \ref{prop:Et2},  for any compact one-particle Hermitian operator $J$ on $L^{2}(\mathbb{R})$, we have
	\begin{align*}
		\Tr\left( J\left({\gamma}_{N,t}^{(1)}-|u_{\varepsilon,t}\rangle \langle u_{\varepsilon,t}|\right)\right) & =\int_{\bR\times\bR}\dx\dy\,
		J(x;y)\left(\gamma_{N,t}^{(1)}(y;x)-u_{\varepsilon,t}(y)\overline{u_{\varepsilon,t}\left(x\right)}\right)\\
		& =\frac{d_{N}}{N}\left\langle \frac{\left(a^{*}(\varphi)\right)^{N}}{\sqrt{N!}}\Omega,\cW(\sqrt{N}\varphi)\,\cU^{*}(t;0)d\Gamma(J)\,\cU(t;0)\Omega\right\rangle \\
		& \quad+\frac{d_{N}}{\sqrt{N}}\left\langle \frac{\left(a^{*}(\varphi)\right)^{N}}{\sqrt{N!}}\Omega,\cW(\sqrt{N}\varphi)\,\cU^{*}(t;0)\phi(Ju_{\varepsilon,t})\,\cU(t;0)\Omega\right\rangle \\
		& =E_{t}^{(1)}(J)+E_{t}^{(2)}(J).
	\end{align*}
	The second step can be carried out by using the expression for $d\Gamma(J)$ in terms of operator-valued distributions in \eqref{eq:dgjop}, the definition of annihilation and creation operators in terms of operator-valued distributions in \eqref{eq:creation} and \eqref{eq:annihilation} and the definition of $\phi$ in \eqref{eq:phi}. Thus, from Propositions \ref{prop:Et1} and \ref{prop:Et2}, we find that
	\begin{equation*}\label{eq:trjpro3132}\begin{split}
		&\left|\Tr\left( J\big({\gamma}_{N,t}^{(1)}-|u_{\varepsilon,t}\rangle \langle u_{\varepsilon,t}|\big)\right)\right|
		\leq \frac{C\|J\|}{N}\left(\Big(1+\frac{1}{\varepsilon^{\frac{7}{2}}N}\Big) \Big(1+\frac{1}{\varepsilon^{\frac{37}{2}}N}\Big) \exp\left(K  (1+\varepsilon^{-1}) t \right)
		+
		e^{Kt} \Big(1+\varepsilon^{-3/2}\Big)\right).
	\end{split}\end{equation*}
	Then 
	\begin{equation}\label{eq:Gamma-u}\begin{split}
			&\Tr\left|{\gamma}_{N,t}^{(1)}-|u_{\varepsilon,t}\rangle \langle u_{\varepsilon,t}|\right|
			\leq \frac{C\|J\|}{N}\left(\Big(1+\frac{1}{\varepsilon^{\frac{7}{2}}N}\Big) \Big(1+\frac{1}{\varepsilon^{\frac{37}{2}}N}\Big) \exp\left(K  (1+\varepsilon^{-1}) t \right)
			+
			e^{Kt} \Big(1+\varepsilon^{-3/2}\Big)\right).
	\end{split}\end{equation}
	By the triangle inequality and the Cauchy-Schwarz inequality we have
	\begin{equation}\label{eq:u-phi}\begin{split}
		\Tr\Big||u_{\varepsilon,t}\rangle\langle u_{\varepsilon,t}| - |\varphi_{t}\rangle \langle {\varphi}_{t}|\Big|
		&=
		\Tr\Big||u_{\varepsilon,t} \rangle \langle u_{\varepsilon,t}| - |u_{\varepsilon,t}\rangle \langle {\varphi}_{t}| +|u_{\varepsilon,t}\rangle \langle {\varphi}_{t}| -|{\varphi}_{t}\rangle \langle {\varphi}_{t}| \Big|\\
		&=
		\Tr\Big||u_{\varepsilon,t} \rangle  \langle u_t - {\varphi}_{t} | +|u_{\varepsilon,t} - \varphi_{t}\rangle \langle \varphi_{t}|\Big|\\
		&\leq
		2\|{u}_{\varepsilon,t}- \varphi_{t}\|_2\\
		&\leq C \varepsilon^\eta e^{Kt}
	\end{split}\end{equation}
	for $0<\eta<1/2$ where we have used Theorem \ref{thm:Convergence}.
	For the last inequality, we have used Theorem \ref{thm:Convergence}.
	Then combining \eqref{eq:Gamma-u} and \eqref{eq:u-phi} and noting that $0<\varepsilon<1$ and $0<\eta<1/2$, we have
	\[\begin{split}
		\Tr\left|{\gamma}_{N,t}^{(1)}-|\varphi_{t}\rangle \langle {\varphi}_{t}|\right|
		&\leq \Tr\left|{\gamma}_{N,t}^{(1)}-|u_{\varepsilon,t}\rangle \langle u_{\varepsilon,t}|\right| + \Tr\Big||u_{\varepsilon,t}\rangle\langle u_{\varepsilon,t}| - |\varphi_{t}\rangle \langle {\varphi}_{t}|\Big|\\
		&\leq \frac{C\|J\|}{N}\left( \Big(1+\frac{1}{\varepsilon^{\frac{37}{2}}N}\Big) \exp\left(K  (1+\varepsilon^{-1}) t \right)
		+ e^{Kt} \Big(1+\varepsilon^{-3/2}\Big)\right)
		+ C \sqrt{\varepsilon} (1+t)\\
		&\leq
		\frac{C}{N}\,\Big(1+\frac{1}{\varepsilon^{3/2}}+\frac{1}{\varepsilon^{\frac{37}{2}}N}+\frac{1}{\varepsilon^{37}N^2}\Big) \exp\left(K (1+\varepsilon^{-1}) t \right)
		+ C \varepsilon^\eta \, e^{Kt}
	\end{split}\]
	which concludes the proof of Theorem \ref{thm:main}.
\end{proof}

\small

\bibliographystyle{abbrv}
\bibliography{refs-point,refs-manybody}

\begin{thebibliography}{10}

\bibitem{Adami2020blow}
R.~Adami, R.~Carlone, M.~Correggi, and L.~Tentarelli.
\newblock Blow-up for the pointwise {NLS} in dimension two: {A}bsence of
  critical power.
\newblock {\em J. Differ. Equ.}, 269(1):1--37, 2020.

\bibitem{Adami2021Stability2D}
R.~Adami, R.~Carlone, M.~Correggi, and L.~Tentarelli.
\newblock Stability of the standing waves of the concentrated {NLSE} in
  dimension two.
\newblock {\em Math. eng.}, 3(2):1--15, 2021.

\bibitem{ADFT03}
R.~Adami, G.~Dell'Antonio, R.~Figari, and A.~Teta.
\newblock The {C}auchy problem for the {S}chrödinger equation in dimension
  three with concentrated nonlinearity.
\newblock {\em Ann. Inst. Henri Poincare (C) Anal.}, 20(3):477--500, 2003.

\bibitem{ADFT04}
R.~Adami, G.~Dell'Antonio, R.~Figari, and A.~Teta.
\newblock Blow-up solutions for the {S}chrödinger equation in dimension three
  with a concentrated nonlinearity.
\newblock {\em Ann. Inst. Henri Poincare (C) Anal.}, 21(1):121--137, 2004.

\bibitem{AFH21}
R.~Adami, R.~Fukuizumi, and J.~Holmer.
\newblock Scattering for the ${L}^2$ supercritical point {NLS}.
\newblock {\em Trans. Am. Math. Soc.}, 374(1):35--60, 2021.

\bibitem{ANO13}
R.~Adami, D.~Noja, and C.~Ortoleva.
\newblock {Orbital and asymptotic stability for standing waves of a nonlinear
  {S}chrödinger equation with concentrated nonlinearity in dimension three}.
\newblock {\em J. Math. Phys.}, 54(1):013501, 01 2013.

\bibitem{ANO16}
R.~Adami, D.~Noja, and C.~Ortoleva.
\newblock Asymptotic stability for standing waves of a {NLS} equation with
  subcritical concentrated nonlinearity in dimension three: {N}eutral modes.
\newblock {\em Discrete Contin. Dyn. Syst.}, 36(11):5837--5879, 2016.

\bibitem{AdamiTeta2001class}
R.~Adami and A.~Teta.
\newblock A class of nonlinear {S}chr{\"o}dinger equations with concentrated
  nonlinearity.
\newblock {\em J. Funct. Anal.}, 180(1):148--175, 2001.

\bibitem{AGH-KH88}
S.~Albeverio, F.~Gesztesy, R.~Høegh-Krohn, and H.~Holden.
\newblock {\em Solvable Models in Quantum Mechanics}.
\newblock Theoretical and Mathematical Physics. Springer Berlin, Heidelberg,
  1988.

\bibitem{BOS2015}
N.~Benedikter, G.~de~Oliveira, and B.~Schlein.
\newblock Quantitative derivation of the {G}ross-{P}itaevskii equation.
\newblock {\em Commun. Pure Appl. Math.}, 68(8):1399--1482, 2015.

\bibitem{BCS2016AHP}
C.~Boccato, S.~Cenatiempo, and B.~Schlein.
\newblock Quantum many-body fluctuations around nonlinear {S}chr{\"o}dinger
  dynamics.
\newblock {\em Ann. Henri Poincar\'e}, 18(1):113--191, 2017.

\bibitem{BNNS19}
C.~Brennecke, P.~T. Nam, M.~Napiórkowski, and B.~Schlein.
\newblock Fluctuations of ${N}$-particle quantum dynamics around the nonlinear
  {S}chrödinger equation.
\newblock {\em Ann. Inst. Henri Poincare (C) Anal.}, 36(5):1201--1235, 2019.

\bibitem{BS19}
C.~Brennecke and B.~Schlein.
\newblock {G}ross--{P}itaevskii dynamics for {B}ose--{E}instein condensates.
\newblock {\em Anal. PDE}, 12(6):1513--1596, 2019.

\bibitem{BKKS08}
V.~Buslaev, A.~Komech, E.~Kopylova, and D.~Stuart.
\newblock On asymptotic stability of solitary waves in {S}chr{\"o}dinger
  equation coupled to nonlinear oscillator.
\newblock {\em Commun. Partial. Differ. Equ.}, 33(4):669--705, 2008.

\bibitem{CFNT14}
C.~Cacciapuoti, D.~Finco, D.~Noja, and A.~Teta.
\newblock The {NLS} equation in dimension one with spatially concentrated
  nonlinearities: the pointlike limit.
\newblock {\em Lett. Math. Phys.}, 104:1557--1570, 2014.

\bibitem{CFNT17}
C.~Cacciapuoti, D.~Finco, D.~Noja, and A.~Teta.
\newblock The point-like limit for a {NLS} equation with concentrated
  nonlinearity in dimension three.
\newblock {\em J. Funct. Anal.}, 273(5):1762--1809, 2017.

\bibitem{CCT18}
R.~Carlone, M.~Correggi, and L.~Tentarelli.
\newblock Well-posedness of the two-dimensional nonlinear {S}chrödinger
  equation with concentrated nonlinearity.
\newblock {\em Ann. Inst. Henri Poincare (C) Anal.}, 36(1):257--294, 2019.

\bibitem{ChenLee2011}
L.~Chen and J.~O. Lee.
\newblock Rate of convergence in nonlinear {H}artree dynamics with factorized
  initial data.
\newblock {\em J. Math. Phys.}, 52(5):052108, 25, 2011.

\bibitem{ChenLeeLee2018}
L.~Chen, J.~O. Lee, and J.~Lee.
\newblock Rate of convergence toward {H}artree dynamics with singular
  interaction potential.
\newblock {\em J. Math. Phys.}, 59(3):031902, 2018.

\bibitem{ChenLeeSchlein2011}
L.~Chen, J.~O. Lee, and B.~Schlein.
\newblock Rate of convergence towards {H}artree dynamics.
\newblock {\em J. Stat. Phys.}, 144(4):872--903, 2011.

\bibitem{ChenPavlovic2011}
T.~Chen and N.~Pavlovi{\'c}.
\newblock The quintic {NLS} as the mean field limit of a boson gas with
  three-body interactions.
\newblock {\em J. Funct. Anal.}, 260(4):959--997, 2011.

\bibitem{Xchen2012second}
X.~Chen.
\newblock Second order corrections to mean field evolution for weakly
  interacting bosons in the case of three-body interactions.
\newblock {\em Arch. Ration. Mech. Anal.}, 203(2):455--497, 2012.

\bibitem{XChenHolmer2019derivation}
X.~Chen and J.~Holmer.
\newblock The derivation of the $\mathbb{T}^{3}$ energy-critical nls from
  quantum many-body dynamics.
\newblock {\em Invent. Math.}, 217:433--547, 2019.

\bibitem{GV79-1}
J.~Ginibre and G.~Velo.
\newblock The classical field limit of scattering theory for non-relativistic
  many-boson systems. {I}.
\newblock {\em Commun. Math. Phys.}, 66:37--76, 1979.

\bibitem{GV79-2}
J.~Ginibre and G.~Velo.
\newblock The classical field limit of scattering theory for non-relativistic
  many-boson systems. {II}.
\newblock {\em Commun. Math. Phys.}, 68:45--68, 1979.

\bibitem{Gorenflo-Vessella}
R.~Gorenflo and S.~Vessella.
\newblock {\em Abel Integral Equations: Analysis and Applications}.
\newblock Lecture Notes in Mathematics. Springer Berlin, Heidelberg, 1 edition,
  1991.

\bibitem{grillakis2010second}
M.~Grillakis, M.~Machedon, and D.~Margetis.
\newblock Second-order corrections to mean field evolution of weakly
  interacting bosons {I}.
\newblock {\em Commun. Math. Phys.}, 294(1):273--301, 2010.

\bibitem{grillakis2011second}
M.~Grillakis, M.~Machedon, and D.~Margetis.
\newblock Second-order corrections to mean field evolution of weakly
  interacting bosons {II}.
\newblock {\em Adv. Math.}, 228(3):1788--1815, 2011.

\bibitem{Hepp}
K.~Hepp.
\newblock The classical limit for quantum mechanical correlation functions.
\newblock {\em Commun. Math. Phys.}, 35:265--277, 1974.

\bibitem{Holmer2020blow}
J.~Holmer and C.~Liu.
\newblock Blow-up for the 1{D} nonlinear {S}chr{\"o}dinger equation with point
  nonlinearity {I}: {B}asic theory.
\newblock {\em J. Math. Anal. Appl.}, 483(1):123522, 2020.

\bibitem{kirkpatrick2011derivation}
K.~Kirkpatrick, B.~Schlein, and G.~Staffilani.
\newblock Derivation of the two-dimensional nonlinear {S}chr{\"o}dinger
  equation from many body quantum dynamics.
\newblock {\em Am. J. Math.}, 133(1):91--130, 2011.

\bibitem{KnowlesPickl2010meanfield}
A.~Knowles and P.~Pickl.
\newblock Mean-field dynamics: singular potentials and rate of convergence.
\newblock {\em Commun. Math. Phys.}, 298:101--138, 2010.

\bibitem{KKS12}
A.~Komech, E.~Kopylova, and D.~Stuart.
\newblock On asymptotic stability of solitons in a nonlinear {S}chr\"odinger
  equation.
\newblock {\em Commun. Pure Appl.}, 11(3):1063--1079, 2012.

\bibitem{Lee2020rate}
J.~Lee.
\newblock Rate of convergence towards mean-field evolution for weakly
  interacting bosons with singular three-body interactions.
\newblock {\em arXiv preprint arXiv:2006.13040}, 2020.

\bibitem{Lee2013}
J.~O. Lee.
\newblock Rate of convergence towards semi-relativistic {H}artree dynamics.
\newblock {\em Ann. Henri Poincar\'e}, 14(2):313--346, 2013.

\bibitem{LSY04}
E.~H. Lieb, R.~Seiringer, and J.~Yngvason.
\newblock One-dimensional behavior of dilute, trapped {B}ose gases.
\newblock {\em Commun. Math. Phys.}, 244:347--393, 2004.

\bibitem{MA93}
B.~A. Malomed and M.~Y. Azbel.
\newblock Modulational instability of a wave scattered by a nonlinear center.
\newblock {\em Phys. Rev. B}, 47(16):10402, 1993.

\bibitem{NN17-1}
P.~T. Nam and M.~Napi{\'o}rkowski.
\newblock Bogoliubov correction to the mean-field dynamics of interacting
  bosons.
\newblock {\em Adv. Theor. Math. Phys.}, 21(3):683--738, 2017.

\bibitem{NN17-2}
P.~T. Nam and M.~Napi{\'o}rkowski.
\newblock A note on the validity of {B}ogoliubov correction to mean-field
  dynamics.
\newblock {\em J. Math. Pures. Appl.}, 108(5):662--688, 2017.

\bibitem{NamRicaudTriay2022}
P.~T. Nam, J.~Ricaud, and A.~Triay.
\newblock Ground state energy of the low density {B}ose gas with three-body
  interactions.
\newblock {\em J. Math. Phys.}, 63(7):071903, 2022.

\bibitem{nam2020derivation}
P.~T. Nam and R.~Salzmann.
\newblock Derivation of {3D} energy-critical nonlinear {S}chr{\"o}dinger
  equation and bogoliubov excitations for {B}ose gases.
\newblock {\em Commun. Math. Phys.}, 375(1):495--571, 2020.

\bibitem{N93}
F.~Nier.
\newblock The dynamics of some quantum open systems with short-range
  nonlinearities.
\newblock {\em Nonlinearity}, 11(4), 1998.

\bibitem{Pickl2011LMP}
P.~Pickl.
\newblock A simple derivation of mean field limits for quantum systems.
\newblock {\em Lett. Math. Phys.}, 97(2):151--164, 2011.

\bibitem{PJ-LC91}
C.~Presilla, G.~Jona-Lasinio, and F.~Capasso.
\newblock Nonlinear feedback oscillations in resonant tunneling through double
  barriers.
\newblock {\em Phys. Rev. B}, 43(5200(R)), 1991.

\bibitem{RodnianskiSchlein2009}
I.~Rodnianski and B.~Schlein.
\newblock Quantum fluctuations and rate of convergence towards mean field
  dynamics.
\newblock {\em Commun. Math. Phys.}, 291(1):31--61, 2009.

\bibitem{S80}
H.~Spohn.
\newblock Kinetic equations from {H}amiltonian dynamics: {M}arkovian limits.
\newblock {\em Rev. Mod. Phys.}, 52(3):569, 1980.

\end{thebibliography}

\end{document}